%% file: main.tex
\documentclass[trsc,nonblindrev]{informs3} % current default for manuscript submission
\pdfoutput=1
\OneAndAHalfSpacedXI % current default line spacing
\usepackage{amsmath, amsfonts, url, graphicx, balance}
\usepackage{mathrsfs,upgreek}
\usepackage[utf8x]{inputenc}
\usepackage[bookmarks]{hyperref}
\usepackage{xfrac}
\usepackage{subfig}
\usepackage{mathptmx}
\usepackage{pgf}
\usepackage[normalem]{ulem}
\usepackage{makecell}
\usepackage{pifont}% http://ctan.org/pkg/pifont
\usepackage{tabu}
\usepackage{float}
\usepackage{tikz}
\usepackage{array}
\usepackage{bm}
\usepackage{upgreek}
\usepackage{epstopdf}
\usepackage[overload]{empheq}
\newcolumntype{L}[1]{>{\raggedright\let\newline\\\arraybackslash\hspace{0pt}}m{#1}}
\newcolumntype{C}[1]{>{\centering\let\newline\\\arraybackslash\hspace{0pt}}m{#1}}
\newcolumntype{R}[1]{>{\raggedleft\let\newline\\\arraybackslash\hspace{0pt}}m{#1}}
\usepackage{natbib}
 \bibpunct[, ]{(}{)}{,}{a}{}{,}%
 \def\bibfont{\small}%
 \def\newblock{\ }%
\TheoremsNumberedThrough     % Preferred (Theorem 1, Lemma 1, Theorem 2)
\EquationsNumberedThrough    % Default: (1), (2), ...
\MANUSCRIPTNO{1} % When the article is logged in and DOI assigned to it,
\include{notation}
\begin{document}
%%%%%%%%%%%%%%%%

%\setcounter{page}{1}
%\DOI{10.1287/xxxx.0000.0000}%

\RUNAUTHOR{Liu, Amin, and Schwartz}

\RUNTITLE{Information Heterogeneity in Bayesian Congestion Games}

% \TITLE{Bundling Information Goods of Decreasing Value}
\TITLE{Effects of Information Heterogeneity in Bayesian Congestion Games}

\ARTICLEAUTHORS{%
\AUTHOR{Jeffrey Liu, Saurabh Amin}
\AFF{Massachusetts Institute of Technology, \EMAIL{\{jeffliu,amins\}@mit.edu}}
\AUTHOR{Galina Schwartz}
\AFF{University of California, Berkeley, \EMAIL{schwartz@eecs.berkeley.edu}}
} % end of the block

\ABSTRACT{%
This article studies the value of information in route choice decisions when a fraction of players have access to high accuracy information about traffic incidents relative to others. To model such environments, we introduce a Bayesian congestion game, in which players have private information about incidents, and each player chooses her route on a network of parallel links. The links are prone to incidents that occur with an ex-ante known probability. The demand is comprised of two player populations: one with access to high accuracy incident information and another with low accuracy information, i.e. the populations differ only by their access to information. The common knowledge includes: (i) the demand and route cost functions, (ii) the fraction of highly-informed players, (iii) the incident probability, and (iv) the marginal type distributions induced by the information structure of the game. We present a full characterization of the Bayesian Wardrop Equilibrium of this game under the assumption that low information players receive no additional information beyond common knowledge. We also compute the cost to individual players and the social cost as a function of the fraction of highly-informed players when they receive perfectly accurate information. Our first result suggests that below a certain threshold of highly-informed players, both populations experience a reduction in individual cost, with the highly-informed players receiving a greater reduction. However, above this threshold, both populations realize the same equilibrium cost. Secondly, there exists another (lower or equal) threshold above which a further increase in the fraction of highly-informed players does not reduce the expected social costs. Thus, once a sufficiently large number of players are highly informed, wider distribution of more accurate information is ineffective at best, and otherwise socially harmful. 

}%

%\KEYWORDS{deterministic inventory theory; infinite linear programming duality; 
%  existence of optimal policies; semi-Markov decision process; cyclic schedule}

\KEYWORDS{Value of information, incomplete information games, route choice behavior}
\HISTORY{Submitted for review March 28, 2016}

\maketitle
%%%%%%%%%%%%%%%%%%%%%%%%%%%%%%%%%%%%%%%%%%%%%%%%%%%%%%%%%%%%%%%%%%%%%%

% Samples of sectioning (and labeling) in TRSC
% NOTE: (1) \section and \subsection do NOT end with a period
%       (2) \subsubsection and lower need end punctuation
%       (3) capitalization is as shown (title style).
%
%\section{Introduction.}\label{intro} %%1.
%\subsection{Duality and the Classical EOQ Problem.}\label{class-EOQ} %% 1.1.
%\subsection{Outline.}\label{outline1} %% 1.2.
%\subsubsection{Cyclic Schedules for the General Deterministic SMDP.}
%  \label{cyclic-schedules} %% 1.2.1
%\section{Problem Description.}\label{problemdescription} %% 2.

% Text of your paper here
\include{intro_rev}
\include{environment_rev}
\include{model_rev}
\include{beliefs}

\include{equilibrium}
\include{welfare}
% \include{benchmarks}
\include{discussion}
\ACKNOWLEDGMENT{%
We are grateful to Dr. Michael Schwarz and Prof. Hani Mahmassani for useful discussions.
This work was supported in part by the Google Research Award ``Estimating Social Welfare of Traffic Information Systems'', National Science Foundation (NSF) sponsored project ``CPS: Frontiers: Collaborative Research: Foundations of Resilient CybEr-Physical Systems (FORCES)'' (award \# CNS-1238959, CNS-1238962,
CNS-1239054, CNS-1239166), and NSF award ``CAREER: Resilient Design of Networked Infrastructure Systems: Models, Validation, and Synthesis'' (award \# CNS-1453126). J. Liu also acknowledges support by the MIT CEE Gradate Research Fellowship.}% Leave this (end of acknowledgment)

% Appendix here
% Options are (1) APPENDIX (with or without general title) or 
%             (2) APPENDICES (if it has more than one unrelated sections)
% Outcomment the appropriate case if necessary
%
% \begin{APPENDIX}{<Title of the Appendix>}
% \end{APPENDIX}
%
%   or 
%
% \begin{APPENDICES}
% \section{<Title of Section A>}
% \section{<Title of Section B>}
% etc
% \end{APPENDICES}

% References here (outcomment the appropriate case) 

% CASE 1: BiBTeX used to constantly update the references 
%   (while the paper is being written).
\bibliographystyle{ormsv080} % outcomment this and next line in Case 1
% \bibliography{sources,traffic,traffic_old} % if more than one, comma separated

% CASE 2: BiBTeX used to generate mypaper.bbl (to be further fine tuned)
\input{main.bbl} % outcomment this line in Case 2

\end{document}

%% file: notation.tex
\newcommand{\xmark}{\ding{55}}%

\newcommand*{\rom}[1]{\mathrm{\uppercase\expandafter{\romannumeral #1\relax}}}

\newcommand{\genGame}[0]{\Gamma}

\newcommand{\playerset}[0]{\mathcal{N}}
\newcommand{\playerindex}[0]{n}
\newcommand{\otherplayerindex}[0]{j}
\newcommand{\playertypeset}[0]{\mathrm{T}}

\newcommand{\playertype}[0]{t}
\newcommand{\numplayers}[0]{N}

\newcommand{\typeInf}[0]{\mathrm{H}}
\newcommand{\typeUninf}[0]{\mathrm{L}}

\newcommand{\playerFrac}[1]{\lambda^{\hspace{-0.2em}#1}}
\newcommand{\fracInf}[0]{\playerFrac{\typeInf}}
\newcommand{\fracUninf}[0]{\playerFrac{\typeUninf}}

\newcommand{\fracBoundSym}[0]{\Lambda}
\newcommand{\fracBound}[1]{{\fracBoundSym}_{#1}\hspace{-0.5pt}(\accuracy^\typeInf,\pInc)}
\newcommand{\fracBoundEta}[1]{{\fracBoundSym}_{#1}\hspace{-0.5pt}(\pInc)}

\newcommand{\serviceset}[0]{\mathcal{I}}
\newcommand{\service}[0]{i}
\newcommand{\accuracy}[0]{\eta}
\newcommand{\obs}[0]{y}
\newcommand{\obsrv}{y}
\newcommand{\totaldemand}[0]{\mathrm{D}}

\newcommand{\stateset}[0]{\Omega}
\newcommand{\staterv}[0]{\theta}
\newcommand{\state}[0]{\theta}
\newcommand{\stateinc}[0]{\mathbf{a}}
\newcommand{\statenorm}[0]{\mathbf{n}}

\newcommand{\pInc}[0]{p}

\newcommand{\actionset}[0]{A}
\newcommand{\actionindex}[0]{a}

\newcommand{\routeset}[0]{\mathcal{R}}
\newcommand{\routeindex}[0]{\mathrm{r}}
\newcommand{\numroutes}[0]{M}

\newcommand{\strat}[0]{s}

\newcommand{\loadset}[0]{\mathcal{Q}}

\newcommand{\totLoad}[0]{Q}
\newcommand{\load}[0]{\mathrm{q}}

\newcommand{\latencysymb}[0]{\ell}
\newcommand{\latency}[3]{\latencysymb_{#1}^{#2}\big({#3}\big)} %routeindex, state, load
\newcommand{\slope}[2]{\alpha_{#1}^{#2}} %routeindex, state
\newcommand{\intercept}[1]{\beta_{#1}} %routeindex
\newcommand{\costVar}[0]{C}
\newcommand{\valueVar}[0]{V}
\newcommand{\welfare}[0]{W}

\newcommand{\bneIndicator}{\ast}
\newcommand{\socOpt}{\dagger}

\newcommand{\splitfractionVar}[0]{x}
\newcommand{\splitfraction}[3]{\splitfractionVar_{#2}^{#1#3}} %type, route, state

\newcommand{\solConst}{K_0}
\newcommand{\regConst}{K}

\newcommand{\commprior}[0]{\pi}
\newcommand{\belief}[0]{\mu}
\newcommand{\expect}[1]{\mathbb{E}[#1]}
\newcommand{\pr}[0]{\text{Pr}}
\newcommand{\prob}[1]{\pr\left({#1}\right)}

\newcommand{\regime}[1]{\phi_{#1}} %regime I,II,III,IV

%% file: intro_rev.tex
%!TEX root = main.tex

\section{Introduction}

This article studies the effects of heterogeneous information on traffic route choices and travel time costs. Our work is motivated by the recent advancements in (and widespread adoption of) Traveler Information Systems (TIS). TIS inform their subscribed commuters of traffic conditions and routes based on historical and current measurements of the network. Recently, the increased penetration of smartphones and other GPS-enabled devices over the past decade has allowed for improved traffic data collection (\cite{herrera2010evaluation}) and has promoted the growth of traffic navigation services (\cite{balakrishna2013information}). Several public agencies also use road side infrastructure (for example, changeable message signs) located at select locations to alert to commuters about traffic conditions on onward routes (\cite{emmerink1996variable}, \cite{chen2003system}, \cite{mortazavi2009travel}).

However, even if all TIS providers typically want to provide the ``true state'' of the network, different providers use different technologies and have access to different data. Users may choose different TIS due to marketing, costs, availability, etc., and some may choose not to use TIS at all. This leads to an inherently heterogeneous information structure of the commuting populace, since some commuters have access to more accurate information than others. Furthermore, commuters may react to their beliefs of others' knowledge; e.g. a commuter may avoid a route because she believes that it will be congested because she believes that the subscribers of a certain TIS are likely to take that route. It is therefore important not only to model commuters' beliefs of the network state, but also their beliefs about the other commuters. This leads to the question: \emph{how does heterogeneous information about traffic incidents affect the commuters' equilibrium route choices and costs?}

To address this question, we introduce a Bayesian congestion game model which incorporates heterogeneous beliefs about the network state and about the other players (commuters). Under certain assumptions about the information structure, we characterize the equilibrium of the game and analyze the equilibrium costs. In particular, we focus on identifying the effects of providing incident information to a larger portion of the population. We examine the value of information to players with access to high accuracy information, and to those without this information. Finally, we analyze the effect of information on the social cost.

Our model is a static Bayesian congestion game on a parallel-route network, where the cost of each route is an increasing function of the number of players utilizing that route and the incident state. Each route has a commonly known prior probability of being in either an \emph{abnormal} (or incident) or \emph{normal} state. Players are divided into two populations: those subscribed to a ``high-accuracy'' TIS (population $\typeInf$) and those subscribed to a ``low-accuracy'' TIS (population $\typeUninf$). Each TIS sends out a private signal to all of its subscribers; the probability that the signal reports the true incident state is greater for the high-accuracy TIS than for the low-accuracy one. Each player then updates her belief of the incident state based on the signal she receives. We use the Bayesian Wardrop Equilibrium (BWE) as our solution concept. A strategy profile is a BWE if, for each player type, the expected costs of all utilized routes are equal and less than the expected cost of each unutilized route.

Previous work has identified the need to consider information in modeling route choices. \cite{ben1991dynamic} identify several potential phenomena that may occur due to commuters having access to traffic information, including \emph{concentration}, where commuters who receive the same information end up taking the same route and causing congestion; and \emph{overreaction}, where commuters may incorrectly estimate how others will react to information, leading to oscillations and suboptimal decisions. Additionally, \cite{ben1996impact} demonstrate that predictive traffic information provides a modest reduction in travel time for commuters over myopic information, and importantly, drivers responding to myopic information may do worse than those without any access to information. While our model captures the effect of concentration, we do not address dynamic aspects such as oscillations, since they do not show up in the equilibrium of our static model. Such analysis is better suited for dynamic models of traffic, such as repeated or multi-stage games. However, our model is able to capture the equilibrium effects of players having incorrect beliefs about other players or the network state. 

One important line of work is the seminal ``bottleneck'' model of \cite{vickery69}, which accounts for congestion costs due the time spent in queues and scheduling costs due to differences between desired and actual arrival times. Notably, \cite{arnott1991does} begin to address the effect of potentially noisy information about route capacities on route choices and congestion. They demonstrate that providing perfect information to commuters about the network state reduces congestion, but inaccurate information can increase costs. They also find that the utility of information is dependent on the probability of incident. To the best of our knowledge, there has been no analogous analysis for network congestion games. In addition, their results mainly focus on two boundary cases: one where only one commuter has access to information, and the other case where all commuters have access. We contribute to the current literature on the value of information in route choice problems by considering all fractions of informed commuters (players). Specifically, we develop a network congestion model that incorporates incident probability, information accuracy, and information distribution.

Additionally, \cite{mahmassani1991system} experimentally identify several effects of information on traffic congestion through simulation. From their results, we would like to emphasize three particular effects: (i) there exists an ``optimal'' fraction (less than 1) of players with information that results in the maximum reduction in social cost; (ii) the cost for players with information increases as the fraction of informed players increases; and (iii) even players without information receive a reduction in cost when others have access to information. Although our model is a static, analytic model instead of a dynamic simulation, we are still able to qualitatively capture these effects.

Another relevant body of literature is that of congestion games, which model the negative externalities resulting from selfish choices of players.  Classical results include the existence of a potential function in every congestion game~(\cite{rosenthal1973class}), and the isomorphism between congestion and potential games~(\cite{monderer1996potential}). More recently, \cite{Monderer07multipotentialgames} considers a generalization of congestion games with player-specific cost functions; also see~\cite{Milchtaich1996111} for an earlier work in this direction. In our work, players have identical preferences, and in the absence of information heterogeneity, our game reduces to the classical potential game. Still, the existence, uniqueness, and structural properties of our game's equilibrium cannot be derived by straightforward application of the known results.

Significant work has also gone into characterizing the efficiency of equilibrium in congestion games, including the well-known  ``price of anarchy'' metric (\cite{Koutsoupias}), (\cite{roughgarden2007routing}). Work on improving the performance in congestion games includes congestion pricing and Stackelberg routing. Congestion pricing schemes include the ``Pigouvian tax:'' a tolling scheme where players pay a toll equal to the externality they impose on the network (\cite{pigou1932economics}). \cite{beckmann56} built further on this idea, showing that for convex and increasing route functions, a social planner can reach a socially optimal outcome by using \emph{marginal cost pricing}. In Stackelberg routing, a central agency controls a portion of the total demand and routes it in a way that improves system performance (\cite{Korilis1997}), (\cite{roughgarden2004stackelberg}), (\cite{krichene2014stackelberg}). In addition to congestion pricing and Stackelberg routing, providing information to commuters offers another method for system performance improvement. As we demonstrate in our work, providing incident information to even a fraction of commuters can improve system performance significantly; however, providing information to too many commuters can be counterproductive. Thus, it is important to fully understand the effects of heterogeneous information on route choices and system performance.

Some work has addressed the effects of information heterogeneity in congestion games, but from different perspectives. For example, \cite{gairing2008selfish} extend the congestion game framework to games of incomplete information by introducing a Bayesian congestion game where players are uncertain about other players' demands. Importantly, \cite{acemoglu2016informational} define and explore the ``Informational Braess' Paradox,'' (IPB) a phenomenon where users who receive information about the existence of additional routes may be worse off. They show that the IPB cannot occur in networks that are composed of series of linearly independent (SLI) networks. For networks that are not SLI, there exists a configuration of edge-specific cost functions where IBP will occur. While their work focuses on information about network structure, ours focuses on (potentially noisy) information about network state. Both aspects must be considered in order to understand the effects of information in congestion games.

Our model also contributes to the existing literature on congestion games by incorporating all of the following factors: the probability of increased cost due to an incident on a route, the fraction of the population with access to information, and the accuracy of that information. In contrast to the boundary cases studied by \cite{arnott1991does}, we consider the full spectrum of information accuracy and distribution in a congestion game framework. We also provide a game theoretical model that can capture many of the experimental effects observed in \cite{mahmassani1991system}. Finally, our treatment complements the recent work of \cite{acemoglu2016informational} on information about network structure by analyzing the effects of noisy information about network state.

In this paper, we limit our attention to a single origin-destination pair connected by two parallel routes, where the main route is nominally shorter but prone to incidents and the alternative route is nominally longer but not prone to incidents. We assume that there is a commonly known inelastic total demand, and that population $\typeUninf$ receives an uninformative signal, i.e. population $\typeUninf$ only knows common knowledge. We also assume that players know the marginal type distribution of the other players, where type encapsulates the TIS and private signal received by the player. We also present other belief structures that can be considered in our modeling environment, although we do not analyze their equilibrium effects. Note that our model only considers route choice, and not departure time. Nonetheless, this decision does not unduly limit the practicality of our model, since real-time TIS now allow commuters to make route choices after they are already on the road.

Our main results are as follows: we present a full characterization of the equilibrium route choices under the aforementioned assumptions. Additionally, we analyze the individual and social costs for player compositions ranging from all-uninformed to all-informed, under the assumption that population $\typeInf$ receives the exact realization of the network state. Relative to the baseline cost where all players are uninformed, we find that there exists a threshold fraction of population $\typeInf$ players below which both populations experience a reduction in individual cost, with the highly informed players receiving a greater reduction. However, above this threshold, both populations experience an equal reduction in cost, and the \emph{relative value of information} (difference between the two populations' individual costs) is zero. Thus, below this threshold, a member of population $\typeUninf$ would benefit from gaining access to information (i.e. joining population $\typeInf$), but above this threshold, both populations are equally well off, so there is no incentive to gain access to the information.

We also characterize the social cost and observe that there is a reduction in social cost from increasing the fraction of highly informed players, but only up to a threshold. This can be viewed as the ``optimal'' level of information distribution for social cost reduction. Importantly, this second threshold is less than or equal to the first, depending on the route cost and incident probability parameters. If this threshold is equal to the first, then the relative value of information goes to zero and the social cost stops improving at the same point. However, if this threshold is lower than the first, then there exists a range of fractions where population $\typeUninf$ players would benefit if they became highly informed, but the social cost would increase if they became highly informed.

The paper is structured as follows: Section~\ref{sec:environment} describes the modeling environment; Section~\ref{sec:bayesian_congestion_game} presents the Bayesian congestion game model; Section~\ref{sec:beliefs} describes belief structures that can be considered in our modeling environment; Section~\ref{sec:equilibrium} characterizes the equilibrium, and Section~\ref{sec:cost_analysis} analyzes the equilibrium costs and value of information; Section~\ref{sec:discussion} presents our concluding remarks.

%% file: environment_rev.tex
\section{Modeling Environment}
\label{sec:environment}

In Sec.~\ref{sub:environment}, we introduce a parallel-route road network, populated by a set of players (commuters), who incur the travel costs depending on the aggregate load in their chosen route and the state of the network. The network is in the normal condition (state~$\statenorm$) when there are no incidents on any route of the network. The presence of traffic incident on a route $\routeindex$ is represented by an abnormal condition (state~$\stateinc_\routeindex$) of the network. The state of the network is random, governed by a known distribution.  Next, in Sec.~\ref{sub:traffic_information_services}, we introduce a set of Traffic Information Services (TISs) available to the players. The TISs in this set have different levels of accuracy about the network state. Each player has access to a TIS of some exogenously given level of accuracy. While we describe the modeling environment for finite player settings, the same description extends to the population game of nonatomic players that we introduce in Section~\ref{ssub:populations_of_nonatomic_players}.
 
%For simplicity, we do not consider incident events that simultaneously affect more than one route. 
%two states (conditions): a normal condition or an abnormal one, respectively, and the abnormal state occurs due to the presence of traffic incidents on the route. 
%
%In this section,  we first describe the environment governing the travel costs of commuters (players) on a parallel-route road network, which randomly assumes a normal condition (state $\statenorm$) or an abnormal one (state $\stateinc$). The abnormal state occurs due to the presence of traffic incidents on the road network. Next, we introduce a set of traffic information services (TIS) with different levels of accuracy about the state of the network. 

\subsection{Effect of Random Incidents} % (fold)
\label{sub:environment}
Consider a set of $\playerset = \{1,\ldots,\numplayers\}$ players, where each player $\playerindex \in \playerset$ chooses among a set of $\numroutes$ parallel routes, $\routeset = \{\routeindex_1,\ldots,\routeindex_\numroutes\}$, connecting a single origin node $o$ and destination node $d$. Let $\totaldemand$ be the total traffic demand induced by the players. We assume the demand is inelastic (i.e., fixed), and each player induces the same demand on the network.

For each player $\playerindex\in\playerset$, denote the set of pure strategies (i.e. available route choices) as $\actionset^\playerindex \equiv \routeset$. A pure strategy profile is denoted $\actionindex = (\actionindex^1,\ldots,\actionindex^\numplayers)$. Let $\Delta(\actionset^\playerindex)$ denote the set of probability distributions over $\actionset^\playerindex$, and define $\Delta := \prod_\playerindex \Delta(\actionset^\playerindex)$. Let $\splitfractionVar^{\playerindex} \in \Delta(\actionset^\playerindex)$ denote a mixed strategy of player $\playerindex$, and $\splitfractionVar = (\splitfractionVar^1,\ldots,\splitfractionVar^\numplayers) \in \Delta$ a mixed strategy profile for the set of players. A mixed strategy of player $\playerindex$ is an $M$-dimensional probability vector $\splitfractionVar^\playerindex = (\splitfractionVar_1^\playerindex,\ldots,\splitfractionVar_\numroutes^\playerindex)$, where $\splitfractionVar_\routeindex^\playerindex$ denotes the player $\playerindex$'s probability of taking the route $\routeindex \in \routeset$.

The routes of the network are prone to random incidents. When a route~$\routeindex$ is experiencing an incident (abnormal state~$\stateinc_\routeindex$), each player's cost of traveling via this route is higher than the cost when there is no incident on it. We model the allocation of incidents in the network via a fictitious player ``Nature,'' who instantiates the ``state'' of the network (and thus, the state of the game) prior to players choosing their routes. More generally, in incomplete information games, the concept of fictitious player (a.k.a. Nature) is routinely used to instantiate all ``payoff-relevant'' parameters of the game (\cite{fudenberg1991game}). We denote the random state of network as $\left(\staterv_1,\dots,\staterv_\numroutes\right)$, which takes values in the set $\prod_{\routeindex=1}^\numroutes\{\statenorm_\routeindex,\stateinc_\routeindex\}$. The state $\statenorm:=\left(\statenorm_1,\ldots,\statenorm_\numroutes\right)$ represents the nominal state where there are no incidents on any route of the network. The presence of an incident (i.e., abnormal state) on a route~$\routeindex$ of the network is represented by $\staterv_\routeindex=\stateinc_\routeindex$. 

%The road network is prone to random incidents. \textcolor{red}{Each commuter's cost of traveling through the network in an incident (abnormal state) is higher than the cost when the network has no incident (normal state).} 

%
%We denote the state of network as $\staterv$, which takes values in the set $\stateset = \{\statenorm,\stateinc\}$. The outcomes ``$\statenorm$'' and ``$\stateinc$'' represent the ``normal'' and ``abnormal'' states of the network, respectively. The probability mass function of $\staterv$ follows a Bernoulli distribution with parameter $\pInc$, i.e.: 

We assume that route~$\routeindex$ is prone to an incident with an exogenously given Bernoulli probability $\pInc_\routeindex$, independent of all the other routes. Then, the probability of the network state in which the routes $\routeindex_1,\ldots,\routeindex_k$ are in the abnormal state (i.e. $\staterv_1=\stateinc_1,\dots,\staterv_k=\stateinc_k$), and the routes $\routeindex_{k+1},\ldots,\routeindex_{\numroutes}$ are in the nominal state (i.e. $\staterv_{k+1}=\statenorm_{k+1},\dots,\staterv_\numroutes=\statenorm_\numroutes$) is given by: 
\begin{align}
\pr(\staterv_1=\stateinc_1,\dots,\staterv_k=\stateinc_k,\staterv_{k+1}=\statenorm_{k+1},\dots,\staterv_\numroutes=\statenorm_\numroutes) = \prod_{\routeindex=1}^k \pInc_\routeindex \times\prod_{\routeindex={k+1}}^{\numroutes} (1-\pInc_\routeindex).
%\begin{cases}
%  \pInc_r, &\text{ if } \state = \stateinc_r, \quad r=1,\ldots,\numroutes\\
%  \vdots & \vdots\\
%  \pInc_\numroutes, &\text{ if } \state = \stateinc_\numroutes\\
%  1-\sum_{\routeindex}^{\numroutes}\pInc_\routeindex, &\text{ if } \state = \statenorm.
%\end{cases} 
\label{eq:state_prob_dist}
\end{align}

% affect only one route at a time, i.e. ``$\stateinc_\routeindex$'' represents that route $\routeindex$ is facing an incident, but all other routes are operating normally. The probability mass function of $\staterv$ is described by a discrete distribution with exogenously given parameters $\pInc_1,\ldots,\pInc_\numroutes$, i.e.: 
%\begin{align}
%\pr(\staterv) &= \begin{cases}
%  \pInc_r, &\text{ if } \state = \stateinc_r, \quad r=1,\ldots,\numroutes\\
%%  \vdots & \vdots\\
%%  \pInc_\numroutes, &\text{ if } \state = \stateinc_\numroutes\\
%  1-\sum_{\routeindex}^{\numroutes}\pInc_\routeindex, &\text{ if } \state = \statenorm.
%\end{cases} \label{eq:state_prob_dist}
%\end{align}

%\stateset = 

For simplicity, we assume that each route $\routeindex\in\routeset$ has a linear cost function with \emph{state-dependent} slopes $\slope{\routeindex}{\state} \in \mathbb{R}_{+}$, and \emph{state-independent} free-flow travel times $\intercept{\routeindex} \in \mathbb{R}_{+}$: 
\begin{align}
\forall \routeindex \in \routeset, \quad \latency{\routeindex}{\state}{\totLoad_\routeindex} &=
  \slope{\routeindex}{\state}\totLoad_\routeindex + \intercept{\routeindex},
\label{eq:state_latency_func}
\end{align}
where $\totLoad_\routeindex=\frac{\totaldemand}{\numplayers} \sum_{\playerindex=1}^\numplayers \delta_{\{\actionindex^\playerindex = \routeindex\}}$ is the total load induced by the players on route $\routeindex$. Our framework naturally extends to allow for generalized cost functions. We expect that results qualitatively similar to the ones in this paper hold for the case in which the cost function on each route is continuous and increasing in that route's load.
%:
%\begin{align}
%\totLoad_\routeindex(\actionindex) = \demand \sum_{\playerindex=1}^\numplayers \delta_{\{\actionindex^\playerindex = \routeindex\}}. \label{eq:total_demand}
%\end{align} 
It is conventional to impose that the slope of a route when it is facing an incident is greater than or equal to the slope when it is not facing an incident, i.e., $\forall \routeindex \in \routeset, \slope{\routeindex}{\stateinc_\routeindex} \geq \slope{\routeindex}{\statenorm}$. Similar assumptions are made by \cite{arnott1991does} regarding the capacity of routes affected by incidents in the bottleneck model.
%\begin{align}
%\forall \routeindex \in \routeset,\quad \slope{\routeindex}{\stateinc_\routeindex} \geq
%%\slope{\routeindex}{\stateinc_{\routeindex'\neq\routeindex}}=
%\slope{\routeindex}{\statenorm}, \label{eq:slope_comparison}
%\end{align} 
Thus, for a given demand, a route in the abnormal state is at least as costly as in the normal state. We assume that the free-flow travel time, $\intercept{\routeindex}$, is not affected by the state of network; this is mainly due to the observation that incidents occurring during free flow tend to have little to no impact on the travel time of the route (\cite{jin2014piecewise}).

%\textcolor{blue}{
%\begin{remark}{\textbf{[Network states]}}
%Note that although we restrict our discussion to $M+1$ network states, our setup is extensible to a broader range of incident events. For example, the state could be an $M$-dimensional vector taking values in $\prod_{\routeindex=1}^M\{\staterv_0,\dots,\staterv_{I_\routeindex}\}$ where $\staterv_0<\staterv_1\dots<\staterv_{I_\routeindex}$ represent the possible incident scenarios on route $\routeindex$, ranging from no incident ($\state_0$) to the most severe incident ($\staterv_{I_\routeindex}$).
%\end{remark}
%}
% \begin{align}
% \latency{\routeindex}{\staterv}{\load} = \begin{cases}
%   \slope{\routeindex}{\staterv_1}\load + \intercept{\routeindex}, &\text{ if }\staterv = \staterv_1\\
%   \vdots & \\
%   \slope{\routeindex}{\staterv_m}\load + \intercept{\routeindex}, &\text{ if }\staterv = \staterv_m\\
% \end{cases}
% \label{eq:state_latency_func}
% \end{align}
For the ease of exposition, we will limit our attention to a two-route network ($\numroutes = 2$) as shown in Fig.~\ref{fig:network_cartoon}, and impose the following assumption on the network state, and route-cost parameters: 

\begin{assumption}
\label{ass:route_structure} 
%We assume the network has the following properties:

\begin{itemize}
\item[$(A\ref{ass:route_structure})_a$] Only the first route, $\routeindex_1$, is prone to an incident with a probability~$\pInc$, and the state of network takes values in the set~$\stateset=\{\statenorm,\stateinc\}$. With a slight abuse of notation, we use $\state = \stateinc$ to represent an incident on $\routeindex_1$ and no incident on $\routeindex_2$. State $\state = \statenorm$ denotes that there are no incidents on any route.

\item[$(A\ref{ass:route_structure})_b$] The slopes in the route cost functions \eqref{eq:state_latency_func} satisfy the ordering $\slope{1}{\stateinc}>\slope{2}{}\geq\slope{1}{\statenorm}$, 
%\begin{align}
%\slope{1}{\stateinc}>\slope{2}{}>\slope{1}{\statenorm}, \label{eq:slope_order}
%\end{align}
i.e. the slope of the first route's cost function in the abnormal state is larger than the slope of the second route, which is in turn larger than or equal to the slope of the first route in the normal state. In addition, the free-flow travel time for the second route at least as high as that of the first route, i.e. $\intercept{2} \geq \intercept{1}$. 
%\item[$(A\ref{ass:route_structure})_c$] 
\end{itemize}
\end{assumption}
Fig.~\ref{fig:latency_funcs} illustrates $(A\ref{ass:route_structure})_a-(A\ref{ass:route_structure})_b$.

\begin{figure}[htb]
\centering
\subfloat[Two route network]{
\includegraphics[width=0.4\linewidth]{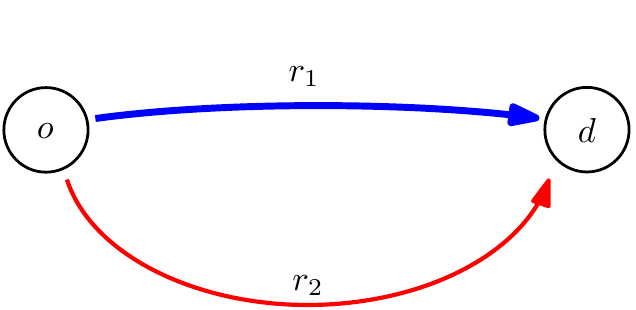}\label{fig:network_cartoon}}
\hspace{1 cm}
\subfloat[Route cost functions]{
\includegraphics[width=0.4\linewidth]{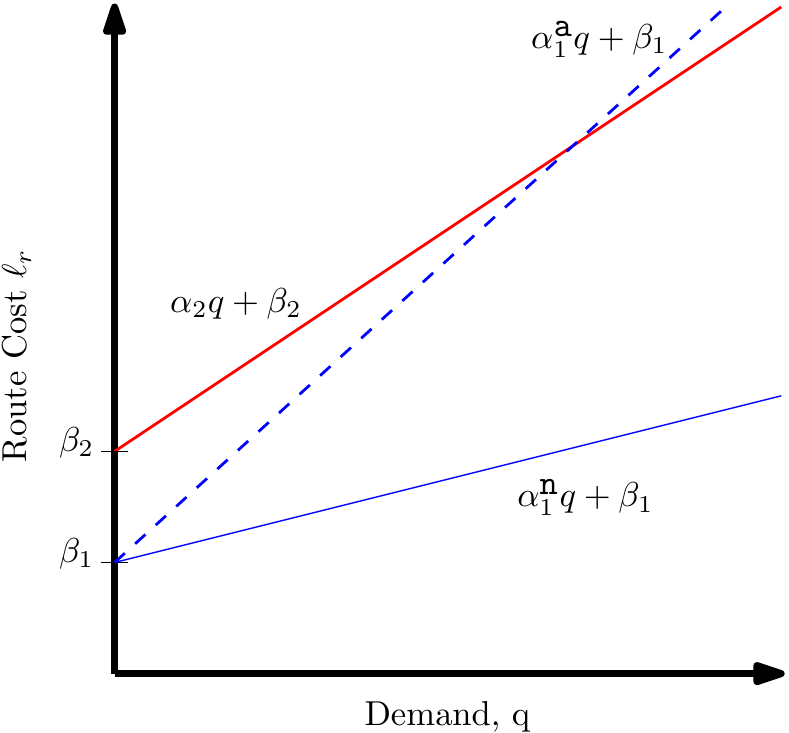}\label{fig:latency_funcs}}
\caption{Model environment.}
\label{fig:environment}
\end{figure}

% subsection environment (end)
\subsection{Effect of Heterogenenous Information} % (fold)
\label{sub:traffic_information_services}
Existing models of route choices in transportation networks approach incidents as purely stochastic events (\cite{Emmerink1998}). Thanks to traffic information services (TIS) such as GPS-enabled devices,  smartphone applications, and in-vehicle navigation systems, commuters are increasingly aware of when and where incidents happen. This  enables the commuters to \emph{strategically} choose their routes in response to current traffic conditions in the network.

A classical paradigm of modeling the effect of incident information is by symmetrically improving the information available across all players (\cite{de2012risk}). However, in reality, commuters use different TIS with different accuracies, or they may choose not to use TIS at all. Additionally, commuters may have different beliefs regarding other players' knowledge. This leads to a heterogeneous information structure, and it is therefore unrealistic to ignore the information heterogeneity. 
 
%Although all TIS providers might want to supply the true state of nature, these entities face limitations in terms of the availability and quality of data needed to estimate the network state with a high level of accuracy. In addition, commuters may choose one service over another due to brand loyalty, cost, or other factors, thus resulting in heterogenous levels of access or utilization.

Motivated by the need to account for heterogeneous traffic information, we approach modeling incident prone networks game theoretically. In our model, each player is subscribed exclusively to an information service $\service \in \serviceset \equiv \{\typeInf,\typeUninf\}$, where $\typeInf,\typeUninf$ denote ``High accuracy'' and ``Low accuracy'' TIS, respectively. 
%For simplicity, we assume that only the first route has a positive probability of incident, and all other routes are always functioning normally. Thus, the only states that can occur are $\stateinc_1$ and $\statenorm$. For notational ease, we drop the subscript of $\stateinc_1$ and simply refer to it as state $\stateinc$.
%
We model the accuracy of the information service $\service \in \serviceset$ by a scalar parameter $\accuracy^\service \in [0.5,1]$, which is the \emph{likelihood} that the TIS $\service$ will report the correct state of nature to any player subscribed to it. To account for a higher accuracy of TIS $\typeInf$ than TIS $\typeUninf$, we let: 
\begin{align}
0.5\leq\accuracy^\typeUninf<\accuracy^\typeInf\leq1. \label{eq:accuracy}
\end{align}
We will refer to players subscribed to information services $\service = \typeInf, \typeUninf $  as ``population $\typeInf$" and "population $\typeUninf$'', respectively. Each player belongs to one of these populations only. 

We let $\fracInf$ and $\fracUninf$ denote respective fractions of players in populations $\typeInf$ and $\typeUninf$, where $\fracUninf := 1-\fracInf$ since each player belongs to one of these populations only. Each player's TIS is assigned with Bernoulli probability $\fracInf$; i.e. $\forall \playerindex \in \playerset, \pr(\service^\playerindex = \typeInf) = \fracInf, \pr(\service^\playerindex = \typeUninf) = 1-\fracInf$. Note that this assignment is independent of the distribution of network state given in \eqref{eq:state_prob_dist}. Also note that for large $\numplayers$, $\fracInf$ can be viewed as the fraction of players subscribed to TIS $\typeInf$.
%\begin{align}
%\forall \playerindex \in \playerset, \quad\pr(\service^\playerindex = \typeInf) = \fracInf, \quad \pr(\service^\playerindex = \typeUninf) = 1-\fracInf. \label{eq:tech_prob_dist}
%\end{align}

%We refer to players subscribed to information service $\service = \typeInf$ as being players of ``population $\typeInf$,'' and likewise, we refer to those subscribed to the service $\service = \typeUninf$ as ``population $\typeUninf$.'' We denote the fraction of players in population $\typeInf$ as $\fracInf$, and the fraction of players in population $\typeUninf$ as $\fracUninf := 1-\fracInf$, and assign one of the two TIS to each player according to a Bernoulli process with parameter $\fracInf$; i.e.:

We model the signal reported by each service as an outcome from a random variable $\obsrv^\service$ defined over the set $\{\statenorm,\stateinc\}$. The likelihood of a service reporting the true state of nature is given by:
\begin{align}
\forall \service \in \{\typeInf,\typeUninf\}, \quad \begin{cases}
\pr(\obsrv^\service = \stateinc | \staterv = \stateinc) = \accuracy^\service, 
%&\pr(\obsrv^\service = \statenorm | \staterv = \stateinc) = 1-\accuracy^\service
\\
%\pr(\obsrv^\service = \statenorm | \staterv = \statenorm) = \accuracy^\service, 
%&
\pr(\obsrv^\service = \stateinc | \staterv = \statenorm) = 1 - \accuracy^\service.
\end{cases}\label{eq:likelihoods}
\end{align}
That is, if the true state of network is $\staterv = \stateinc$, it will report $\obs^\service=\stateinc$ with probability $\accuracy^\service$ and report $\obs^\service=\statenorm$ with complementary probability. Likewise, if the true state of network is $\statenorm$, the information service will report the signal $\obs^\service=\statenorm$ with probability $\accuracy^\service$, and report the signal $\obs^\service=\stateinc$ with complementary probability $1-\accuracy^\service$. Note that here we have assumed that that the accuracy of reporting the correct state is the same regardless of the true state of network, i.e. the accuracy of each TIS is not affected by the state of network.

With the likelihoods \eqref{eq:likelihoods} and the prior distribution on $\staterv$ in \eqref{eq:state_prob_dist}, each player calculates her posterior belief on the true state $\staterv$ given the received signal $\obsrv^\service$ from her TIS using Bayes' rule (see Sec.~\ref{sec:beliefs} for a more precise description on posterior beliefs):
\begin{align*}
\forall \playerindex \in \playerset, \service^\playerindex \in \{\typeInf,\typeUninf\}, \quad\quad  \prob{\staterv|\obsrv^{\service^\playerindex}} = \frac{\prob{\staterv}\prob{\obsrv^{\service^\playerindex}|\staterv}}{\prob{\obsrv^{\service^\playerindex}}}. 
%\label{eq:TIS_assignment}
\end{align*}

% Note that a perfectly informative service ($\accuracy^\service = 1$) will always report the true state of nature. Conversely, a completely uninformative service ($\accuracy^\service = 0.5$) will be equally likely to report either state, regardless of the true state of nature. 

We are now ready to introduce the notion of ``type,'' which captures all the private information available to each player; see \cite{harsanyi1967games}. In our modeling environment, the signal sent to each player constitutes the private information of the game, and thus defines the type for the player. We define type space for each respective TIS as follows: 
\begin{align}
\playertypeset_\typeInf := \{\typeInf\statenorm,\typeInf\stateinc\}, \quad
\playertypeset_\typeUninf := \{\typeUninf\statenorm,\typeUninf\stateinc\}. \label{eq:typeset}
\end{align}
That is, for $\playerindex \in \playerset$, if player $\playerindex$'s information service $\service^\playerindex = \typeInf$, then her type will be in the set $\playertypeset_\typeInf$; conversely, if $\service^\playerindex = \typeUninf$, then her type will be in the set $\playertypeset_\typeUninf$. If a player is assigned a type $\typeInf\statenorm$, we mean that the player is subscribed to the TIS $\service=\typeInf$ with accuracy parameter~$\accuracy^\typeInf$, and has received a signal $\obs^\service=\statenorm$ from it. Similarly, a player with type $\typeUninf\stateinc$ indicates that the TIS~$\typeUninf$ (with accuracy parameter $\accuracy^\typeUninf$) has provided a signal $\obs^\service=\stateinc$ to the player, and so on. Thus, the \emph{player type encapsulates the information service that the player is subscribed to, the corresponding accuracy parameter, and the signal she receives from the service}.

% subsection information_technologies (end)

%% file: model_rev.tex
%!TEX root = main.tex
\section{Bayesian Routing Game} % (fold)
\label{sec:bayesian_congestion_game}
We now present the Bayesian congestion game with uncertain state and heterogeneous information access about the state. We introduce two formulations of this model: (i) a game $\genGame_f$ with a large but finite set of atomic players, where all players have identical travel time preferences, and each player is subscribed to a TIS of high ($\typeInf$) or low ($\typeUninf$) accuracy; (ii) a game $\genGame_p$ with two populations of non-atomic players also with identical travel time preferences, but one population has access to a more accurate TIS than the other. In both formulations, players have private information about the state of the network, and make route choices to minimize their expected individual cost (travel time) from origin $o$ to destination $d$. We largely follow the notational conventions of Bayesian games (\cite{fudenberg1991game}).  

\subsection{Finite Atomic Players.} % (fold)
\label{ssub:finite_atomic_players}
We assume that the network state and route cost parameters are subject to Assumption~$(A\ref{ass:route_structure})$. Formally, the Bayesian congestion game of $\numplayers$ players is defined as follows:
\begin{align}
\genGame_f := \left(\playerset,\actionset,\stateset,\playertypeset,\costVar,\commprior \right),
\end{align}
where:
\begin{itemize}
  \item[-] $\playerset$ is the set of $\numplayers$ players (with generic member $\playerindex$)
  \item[-] $\actionset$ = $(\actionset^\playerindex)_{\playerindex\in\playerset}$ is the set of action profiles, where $\actionset^\playerindex \equiv \routeset$
  \item[-] $\stateset = \{\statenorm,\stateinc\}$ is the set of game states (with generic member $\state$)
  \item[-] $\playertypeset = (\playertypeset^\playerindex)_{\playerindex\in\playerset}$, where $\playertypeset^\playerindex$ is the type space of player~$\playerindex$, with $\playertypeset^\playerindex = \playertypeset_\typeInf$ if $\playerindex$ is subscribed to TIS $\typeInf$, and $\playertypeset^\playerindex = \playertypeset_\typeUninf$ otherwise. ($\playertypeset_\typeInf$ and $\playertypeset_\typeUninf$ are defined in~\eqref{eq:typeset}.)
  \item[-] $\costVar = (\latencysymb^\staterv_\routeindex)_{\routeindex\in\routeset}$ is the set of cost functions, where cost of each route is given in~\eqref{eq:state_latency_func}. The cost to an individual player is equal to the travel time of the route that she chooses.
  % \begin{align*}
  % \forall \state \in \stateset:
  % \costVar^\playerindex(\state,\actionindex^{\playerindex},\actionindex^{-\playerindex}) = \latencysymb_{\actionindex^\playerindex = \routeindex}^{\state}(\totLoad_\routeindex(\actionindex^{\playerindex},\actionindex^{-\playerindex}))
  % \end{align*}
  \item[-] $\commprior\in\Delta (\stateset \times \playertypeset)$ is a common prior which is a joint probability distribution $\commprior(\state,\playertype^\playerindex,\playertype^{-\playerindex})$ over the state of the network and player types. 
%  \item[-] $\belief = \left(\belief^\playerindex\right)_{\playerindex\in\playerset}$ is the set of beliefs $\belief^\playerindex(\cdot|\playertype^\playerindex)\in \Delta (\stateset \times \playertypeset^{-\playerindex})$ for each player $\playerindex$ over the state of the network, and other players' types, conditioned on player $\playerindex$'s type.
\end{itemize}

The common knowledge includes: the total demand $\totaldemand$; the parameter $\fracInf$ governing the distribution of players subscribed to each TIS; the set of routes $\routeset$ and the corresponding route cost parameters; the type space for each TIS, $\playertypeset_\typeInf, \playertypeset_\typeUninf$; and the common prior distribution~$\commprior$, which includes the probability of incident $\pInc$. Importantly, the TIS's accuracy parameters $\accuracy^\typeInf$ and $\accuracy^\typeUninf$ may or may not be common knowledge. We further explain the common prior specifications and the belief structure in Sec.~\ref{sec:beliefs}. 

%As we will see in~Sec.~\ref{sec:beliefs}, in our modeling environment, the common prior~$\commprior$ can be specified in two ways. The first specification includes the probability of an incident on the network $\pInc$, and the players' TIS conditional likelihoods. 
%Equivalently, in this specification,  The second specification includes the probability of an incident~$\pInc$ and the marginal type distributions. 

%; see Sec.~\ref{sec:beliefs} for more details. 
% Based on the common knowledge, we can specify a common prior distribution $\commprior_f$ as follows: 
% \begin{align}
% \commprior_f(\staterv,\playertype) &= \pr(\staterv)\prod_{\playerindex = 1}^\numplayers \pr(\playertype^\playerindex). \label{eq:comm_prior_form}
% \end{align}

\begin{figure}[htb]
\centering \setlength{\unitlength}{1cm} % zoom 
\begin{picture}(10,2) % second number is how tight together the lines are   \thicklines

        \put(0,1){\line(1,0){10}}

        %Date t
        %\put(1,1.5){\makebox(0,0){$a_{-}$}}
        %\%put(1,0.9){\line(0,1){0.2}}

      % \put(1,0.9){\line(0,1){0.2}}
       %\put(1,0.5){\makebox(0,0){$a_{-}$}}
        %\put(1,0){\makebox(0,0){policy set}}

        %Division 1 investment
       % \put(1,1.6){\makebox(0,0){\footnotesize{Period $1$} }}
   \put(0,1.3){\makebox(0,0){\textit{ex ante} }}

        % \put(0,1){\line(0,1){0.2}} %vertical line
            \put(0,1){\makebox(0,0){$\bullet$}}
        \put(0,0.8){\makebox(0,0){\footnotesize{Players are assigned TIS $\typeInf$ or $\typeUninf$}}}   
        \put(0,0.5){\makebox(0,0){\footnotesize{Nature draws $\staterv$}}}
        \put(0,0.15){\makebox(0,0){\footnotesize{TIS $\service$ reports $\obs^\service$ to its subscribers}}}
       % \put(2.4,0){\makebox(0,0){chosen}} 0.7 is higher than 0

        %Renegotiation
       % \put(3.8,1.5){\makebox(0,0){2}}

        %\put(3.8,0.9){\line(0,1){0.2}}
        %\put(3.8,0.5){\makebox(0,0){$a_{+}$}}
        %\put(5.8,0){\makebox(0,0){(re)negotiation}}

    \put(5.5,1){\makebox(0,0){$\bullet$}}
     \put(5.5,1.3){\makebox(0,0){{\textit{interim}}}}
 \put(5.5,0.8){\makebox(0,0){\footnotesize{Players: -know their type}}}
 \put(6.43,0.48){\makebox(0,0){\footnotesize{-obtain beliefs $\belief^\service(\cdot|\playertype^\service)$}}}
 \put(6,0.1){\makebox(0,0){\footnotesize{-play strategies}}}

        %Transfer Decision
        %\put(6,0.9){\line(0,1){0.2}}
        %\put(6,0.5){\makebox(0,0){$b_{-}$}}
        %\put(8.2,0){\makebox(0,0){decision}}

        %Division 2 invests
      \put(10,0.9){\line(0,1){0.2}}%vertical line
         \put(10,1){\makebox(0,0){$\bullet$}} % bullet on the timeline
       % \put(10,1.6){\makebox(0,0){\footnotesize{Period $2$} }}
        \put(10,1.3){\makebox(0,0){\textit{ex post}}}
        
        \put(10.0,0.8){\makebox(0,0){\footnotesize{Players realize costs}}}

     % \put(10,0.9){\line(0,1){0.2}}
        %\put(10,0.5){\makebox(0,0){$b_{+}$}}
        %\put(10.6,0){\makebox(0,0){chosen}}
\end{picture} 
\caption{Timing of the game.}
\label{fig:game_timing}

\end{figure}
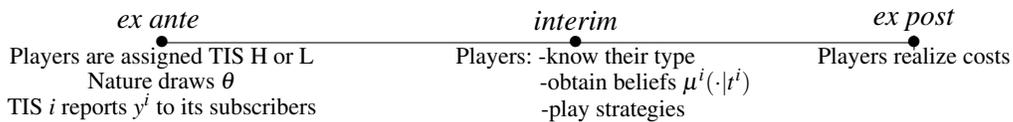

The game is played as follows, see Fig~\ref{fig:game_timing}: first, each player subscribes to one of the two TIS: a high accuracy service ($\typeInf$), or a low accuracy one ($\typeUninf$), according to Bernoulli probability $\fracInf$; this determines each players' TIS subscription $\service$. Nature then draws a realization of the network state $\staterv$ from the distribution \eqref{eq:state_prob_dist}. For a given realization of the state of the network ($\stateinc$ or $\statenorm)$, each TIS broadcasts a signal $\obs^\service$ to all its subscribers. For each player $\playerindex$, her TIS and signal recieved ($\obs^\service$)  determine the player's ``type'' $\playertype^\playerindex \in \playertypeset_\typeInf$ or $\playertypeset_\typeUninf$. The players then simultanously choose their route based on type-dependent (mixed) strategies, $\splitfractionVar^\playerindex:\playertypeset^\playerindex\rightarrow\Delta(\actionset^\playerindex)$, and realize their individual costs.

Following \cite{fudenberg1991game}, once each player's private information is realized, i.e., the players learn their types based on the signal received from their respective TIS, the game enters in an ``interim'' stage. Each player $\playerindex$ knows her own type $\playertype^\playerindex$ but does not know the network state $ \state$ or the other players' types $ \playertype^{-\playerindex}$. This incomplete information is represented by the ``interim'' belief $\belief^\playerindex(\cdot|\playertype^\playerindex)$ of each player $\playerindex$, which she obtains after observing her type $\playertype^\playerindex$, but before choosing her route (see Sec.~\ref{sec:beliefs} for the interim belief structures assumed in this article). We will use the notation $\belief = \left(\belief^\playerindex\right)_{\playerindex\in\playerset}$ to denote the set of beliefs $\belief^\playerindex(\cdot|\playertype^\playerindex)\in \Delta (\stateset \times \playertypeset^{-\playerindex})$ for each player $\playerindex$ over the state of the network and other players' types, conditioned on player $\playerindex$'s type.

The interim stage allows for an equivalent complete information formulation of the Bayesian game, where the game is played between the player types, i.e. the ``players'' of the interim game are the player types, and individual player costs can be calculated conditioned on the player types (\cite{fudenberg1991game}). This allows us to express the interim expected cost function for each player for any mixed strategy profile $\splitfractionVar(\playertype^\playerindex , \playertype^{-\playerindex} ) = (\splitfractionVar^\playerindex(\playertype^\playerindex),\splitfractionVar^{-\playerindex}(\playertype^{-\playerindex}))$ as follows:
\begin{align*}\forall \playerindex \in \playerset, \quad
\expect{\latencysymb_\routeindex^\staterv(\totLoad_\routeindex(\splitfractionVar^\playerindex,\splitfractionVar^{-\playerindex} )| \playertype^\playerindex} 
= \sum_{(\staterv,\playertype^{-\playerindex})}  \belief^{\playerindex}({\staterv,\playertype^{-\playerindex}|\playertype^\playerindex})\sum_{\actionindex \in \actionset}\left(\prod_{\otherplayerindex\in\playerset} \splitfractionVar^\otherplayerindex(\actionindex^\otherplayerindex|\playertype^\otherplayerindex)\right) \latencysymb^\staterv_\routeindex(\totLoad_\routeindex(\actionindex^\playerindex,\actionindex^{-\playerindex})),
\end{align*}
where $\belief^{\playerindex}({\staterv,\playertype^{-\playerindex}|\playertype^\playerindex})$ is player $\playerindex$'s belief on the state and other players' types given its own type.

We are now ready to define the Bayesian Wardrop Equilibrium for the game $\genGame_f$.

% Any Nash equilibrium of the interim game is also a Bayesian Nash equilibrium of $\genGame_f$. Since the interim game is finite, a mixed Nash equilibrium of the interim game always exists. It thus follows that a mixed Bayesian Nash equilibrium of $\genGame_f$ exists. \texttt{[Need citations]}

\begin{definition}{Bayesian Wardrop Equilibrium (BWE) for $\genGame_f$}
\label{def:BNE}

The mixed strategy profile $\splitfractionVar^\bneIndicator = (\splitfractionVar^{1\bneIndicator}(\playertype^1),\dots,\splitfractionVar^{\numplayers\bneIndicator}(\playertype^\numplayers)) \in \prod_{\actionset}\Delta(\actionset^\playerindex)$ is a BWE of the game $\genGame_f$ if for each player $\playerindex \in \playerset$, all routes that is played by player $\playerindex$ with positive probability have equal expected individual cost, which is less than the expected cost for any other route, where all expectations are taken over player $\playerindex$'s belief, $\belief^\playerindex$, i.e.:
\begin{align*}
\forall \playerindex \in \playerset: \splitfractionVar^{\playerindex\bneIndicator}_{\routeindex}(\actionindex^\playerindex=\routeindex|\playertype^\playerindex) & > 0 
\implies \\&\expect{\latency{\actionindex^\playerindex=\routeindex}{\state}{\totLoad_\routeindex(\splitfractionVar^{\playerindex\bneIndicator}(\playertype^\playerindex), \splitfractionVar^{-\playerindex\bneIndicator}(\playertype^{-\playerindex}))}|\playertype^\playerindex} \leq \expect{\latency{\actionindex^\playerindex=\routeindex'}{\state}{\totLoad_{\routeindex'}(\splitfractionVar^{\playerindex\bneIndicator}(\playertype^\playerindex),\splitfractionVar^{-\playerindex\bneIndicator}(\playertype^{-\playerindex}))}|\playertype^\playerindex}, \quad\forall \routeindex' \in \routeset
\end{align*}
\end{definition}

In the limit of large number of players, the contribution of a single player on the load, and thus to the travel time, is negligible. To succinctly model such interactions between large numbers of players, we now introduce in the following section a population game analog of $\genGame_f$, where each population is comprised of nonatomic players. The rest of the article will primarily consider the population game formulation. We refer the reader to \cite{sandholm2010population} for an excellent treatment of population games and the study of their equilibria in a broad range of environments.

%Thus, for large $\numplayers$, instead of undertaking the exhaustive task of computing equilibrium strategies of each individual player, it is more convenient to consider populations determined by TIS ownership.

\subsection{Populations of Nonatomic Players.} % (fold)
\label{ssub:populations_of_nonatomic_players}
Now consider the game played between two populations of nonatomic players: population $\typeInf$ and population $\typeUninf$. We index each population by their respective information service $\service$, i.e. the players in population $\typeInf$ are subscribed to a high accuracy service and those in population $\typeUninf$ are subscribed to a low accuracy one. 
With a slight abuse of notation, each population has a corresponding ``splittable'' demand governed by the parameter~$\fracInf$: 
\begin{align}
\totaldemand^\typeInf := \fracInf\totaldemand,\quad \totaldemand^\typeUninf := \fracUninf\totaldemand = (1-\fracInf)\totaldemand.
\end{align}

We assume that all players belonging to population~$\service$ receive the same realized signal $\obs^\service$, and the likelihood TIS~$\service$ reporting the true state is given by~\eqref{eq:likelihoods}. The type space for each population is defined as in \eqref{eq:typeset}, and all players in population $\service$ learn their type when they receive the signal $\obs^\service$ from their respective information service. 

Each member of a population routes her demand through the network in a way that minimizes her own expected travel cost. The resulting aggregate assignment of demands is referred to as the \emph{strategy distribution} of the population. With an abuse of notation, let $\splitfractionVar^\service: \playertypeset_\service \rightarrow \Delta(\actionset^\service)$ denote an admissible strategy distribution for population $\service$, where $\splitfractionVar^\service(\playertype^\service) = (\splitfractionVar^\service_1(\playertype^\service),\ldots,\splitfractionVar^\service_\numroutes(\playertype^\service))$ can be viewed as a vector of split fractions, with $\splitfractionVar^\service_\routeindex$ being the fraction of population $\service$'s demand that takes route $\routeindex$ when its type is~$\playertype^\service$. 

In the finite-player game, we used the notation $\splitfractionVar$ to denote a player's mixed strategy; in this game we use it to denote the split fraction for a population. We chose this abuse of notation because in the limit of large number of players, the expected fraction of a population taking a given route approaches the probability that a single representative member takes that route.

It is often convenient to consider population $\service$'s strategy distribution in terms of the vector of loads that it assigns to each route. For a strategy distribution $\splitfractionVar^\service(\playertype^\service)$ of population $\service$, let $\load^{\service}(\playertype^\service) = (\load_1^\service(\playertype^\service),\dots,\load_\numroutes^\service(\playertype^\service))$ denote the corresponding load vector, where 
\begin{align}
\forall \service \in \serviceset, \forall \routeindex \in \routeset,\quad\load_{\routeindex}^{\service}(\playertype^\service) := \splitfraction{\service}{\routeindex}{}(\playertype^\service)\totaldemand^\service. \label{eq:splitfraction_load_relation}
\end{align}
We use $(\load^\typeInf(\playertype^\typeInf),\load^{\typeUninf}(\playertype^\typeUninf))$ to denote a generic profile of loads and $(\splitfractionVar^\typeInf(\playertype^\typeInf),\splitfractionVar^{\typeUninf}(\playertype^\typeUninf))$ to denote a profile of strategy distributions (i.e. a profile of load split fractions). Let $\loadset^\service$ denote the set of all admissible load vectors of population $\service$, i.e, for any $\load^\service(\playertype^\service) \in \loadset^\service$ we have $\sum_\routeindex \load_\routeindex^\service(\playertype^\service) = \playerFrac{\service} \totaldemand$. 

%For any admissible ($\load^\typeInf,\load^\typeUninf$), the total load on a route $\routeindex$ can be written as $\totLoad_\routeindex = \load_\routeindex^\typeInf + \load_\routeindex^\typeUninf$. 

Formally, the Bayesian congestion game of populations of nonatomic players can be defined as:
\begin{align}
{\genGame_p} = \left(\serviceset,{\loadset},{\stateset},{\playertypeset},{\costVar},{\commprior} \right),
\end{align}
where:
\begin{itemize}
  \item[-] $\serviceset = \{\typeInf,\typeUninf\}$ is the set of player populations with generic population $\service$
  \item[-] $\loadset$ = $(\loadset^\service)_{\service\in\serviceset}$ is the set of load vectors for each population, i.e. ($\load_1^\service(\playertype^\service),\dots\load_\numroutes^\service(\playertype^\service)$) $\in \loadset^\service$
  \item[-] $\stateset = \{\statenorm,\stateinc\}$ is the set of game states with generic element $\staterv$
  \item[-] $\playertypeset = (\playertypeset_\service)_{\service\in\serviceset}$ where $\playertypeset_\typeInf$ and $\playertypeset_\typeUninf$ denote the type space for population $\typeInf$ and $\typeUninf$, respectively
  \item[-] $\costVar = (\latencysymb^\staterv_\routeindex)_{\routeindex\in\routeset}$ is the set of cost functions for each route governed by \eqref{eq:state_latency_func}.
  % \begin{align*}
  % \forall \state \in \stateset:
  % \costVar^\service(\state,\load^{\service},\load^{-\service}) = \latencysymb_{\load^\service = \routeindex}^{\state}(\totLoad_\routeindex(\load^{\service},\load^{-\service}))
  % \end{align*}

  % The utility function for player $\service$ takes the form given below:
  % \begin{align}
  % \utilfn^\service(\staterv,\splitfractionVar^{\service},\splitfractionVar^{-\service}) = \bonus-\sum_{\routeindex\in\routeset} \latency{\routeindex}{\staterv}{\splitfractionVar_\routeindex^\service\totaldemand^\service + \splitfractionVar_\routeindex^{-\service}\totaldemand^{-\service}} \label{eq:util_form_x}
% \end{align}
  \item[-] $\commprior\in\Delta (\stateset \times \playertypeset)$ is the common prior which is a joint probability distribution $\commprior(\state,\playertype^\typeInf,\playertype^\typeUninf)$ over the state of the network and player types. 
  % \item[-] $\belief = \left(\belief^\service\right)_{\service\in\serviceset}$ is the set of beliefs $\belief^\service(\state,\playertype^{-\service}|\playertype^\service)\in \Delta (\stateset \times \playertypeset_{-\service})$ for each population $\service$ over the state of nature $\staterv,$ and other populations' types $\playertype^{-\service}$, conditioned on the population $\service$'s type $\playertype^\service\in\playertypeset_\service$.
\end{itemize}

The game $\genGame_p$ differs from the interim formulation of the game $\genGame_f$ only by the assumption about player effects on the routes' loads. That is, each individual member of a population in $\genGame_p$ has negligible effect on the load. The common knowledge in $\genGame_p$ is identical to that of $\genGame_f$. Thus, the timing of the game $\genGame_p$ is also captured by Fig.~\ref{fig:game_timing}. 

%The parameters of common knowledge are identical in both games. That is, each population has an identical prior on the state of the network $\staterv$ (before receiving signal $\obsrv^\service$):
%\begin{align}
%\forall \service \in \serviceset,\quad \pr^\service(\state) \equiv \pr(\state)
%\end{align}
%where $\pr(\state)$ governs the distribution of the network state; see \eqref{eq:state_prob_dist}. In addition, the demand of each population $\totaldemand^\service$, is also common knowledge.

% It is convenient for us to reformulate the utility function in terms of load, $\load$, rather than split fraction, $\splitfractionVar$. Hence we combine \eqref{eq:load_def} and \eqref{eq:util_form_x} to get the following:
% \begin{align}
%   \utilfn^\service(\state,\load^{\service},\load^{-\service}) = \bonus-\sum_{\routeindex\in\routeset} \latency{\routeindex}{\state}{\load_\routeindex^\service + \load_\routeindex^{-\service}}\frac{\load_\routeindex^\service}{\totaldemand^\service} \label{eq:util_form_q}
% \end{align}

For a load profile $(\load^\typeInf(\playertype^\typeInf),\load^{\typeUninf}(\playertype^\typeUninf))\in\loadset^\typeInf \times \loadset^{\typeUninf}$, the expected cost of population $\service$ on a given route is:
\begin{align}
\forall \service \in \serviceset, \quad \expect{\latency{\routeindex}{\state}{\load^{\service}_\routeindex(\playertype^\service)+\load^{-\service}_\routeindex(\playertype^{-\service})}|{\playertype^\service}} = \sum_{(\state,\playertype^{-\service})} \latency{\routeindex}{\state}{\load^\service_\routeindex(\playertype^\service) + \load^{-\service}_\routeindex(\playertype^{-\service})} \belief^\service (\state,\playertype^{-\service}|\playertype^\service).\label{eq:expected_cost_population}
\end{align}
We are now ready to state the Bayesian Wardrop Equilibrium for the game $\genGame_p$.
\begin{definition}{Bayesian Wardrop Equilibrium (BWE) for $\genGame_p$}
\label{def:bayesian_wardrop}

A profile of load vectors $(\load^{\service\bneIndicator}(\playertype^\service),\load^{-\service\bneIndicator}(\playertype^{-\service})) \in \loadset^\service \times \loadset^{-\service}$, or the corresponding profile of strategy distributions $(\splitfractionVar^{\service\bneIndicator}(\playertype^\service),\splitfractionVar^{-\service\bneIndicator}(\playertype^{-\service}))$, is an equilibrium if, for each population $\service$ in $\serviceset$, all routes that are utilized by population $\service$ players have equal expected cost, which is less than the expected cost for any route not utilized by players of population $\service$, where all expectations are taken over population $\service$'s belief $\belief^\service$. That is:
\begin{align}
\forall \service \in \serviceset: \load^{\service\bneIndicator}_\routeindex > 0 \implies \expect{\latency{\routeindex}{\state}{\load^{\service\bneIndicator}_\routeindex(\playertype^\service) + \load^{-\service\bneIndicator}_\routeindex(\playertype^{-\service})}|\playertype^{\service}} \leq \expect{\latency{\routeindex'}{\state}{\load^{\service\bneIndicator}_{\routeindex'}(\playertype^{\service})+\load^{-\service\bneIndicator}_{\routeindex'}(\playertype^{-\service})}|\playertype^\service},\quad \forall \routeindex' \in \routeset. \label{eq:equilib_condition_popgame}
\end{align}
\end{definition}

We observe from the definition of the game $\genGame_p$ that the parameters $\pInc,\fracInf,\accuracy^\typeInf,$ and $\accuracy^\typeUninf$ govern the extend of information heterogeneity that is captured in the belief $\belief$. We will henceforth use the tuple $(\pInc,\fracInf,\accuracy^\typeInf,\accuracy^\typeUninf)$ to represent the \emph{information environment} of the game $\genGame_p$.
% subsubsection populations_of_nonatomic_players (end)

\subsection{Equilibrium costs} % (fold)
\label{sub:equilibrium_costs}
In any equilibrium profile $(\load^{\service\bneIndicator}(\playertype^\service),\load^{-\service\bneIndicator}(\playertype^{-\service})) \in \loadset^\service \times \loadset^{-\service}$ of game $\genGame_p$, each population will have an associated \emph{realized cost} of traveling through the network whose actual state was unknown at the time of making route choice decisions. The players realize their individual outcome after having played according to their equilibrium strategy. 

We define the \emph{equilibrium cost for a player of population $\service$ in state $\staterv$} as follows:
\begin{align}
 \costVar_{\staterv}^{\service\bneIndicator} &:= \sum_{\routeindex} \sum_{(\playertype^\service,\playertype^{-\service})}\splitfractionVar^{\service\bneIndicator}_{\routeindex}(\playertype^\service)\latency{\routeindex}{\staterv}{\load_\routeindex^{\service\bneIndicator}(\playertype^\service)+\load_\routeindex^{-\service\bneIndicator}(\playertype^{-\service}) } \pr(\playertype^\service|\staterv) \pr(\playertype^{-\service}|\staterv), \quad \service \in \serviceset, \staterv\in\stateset. \label{eq:state_type_cost_form}
\end{align}

One can average over populations and/or states to get various composite costs. We refer to the costs averaged over states as expected population-dependent costs; the costs averaged over populations as the state-dependent social cost; and the cost averaged over both populations and states as the expected social cost. More precisely, these costs are defined as follows: 

\begin{itemize}
\item[(i)] The \emph{expected population-dependent cost} in equilibrium is the average cost incurred by a player of a given population across all network states: 
\begin{align}
\costVar^{\service\bneIndicator} := \sum_{\staterv\in\stateset} \prob{\staterv} \costVar^{\service\bneIndicator}_{\staterv}, \quad \service \in \serviceset.\label{eq:exp_pop_cost}
\end{align}
\item[(ii)] The \emph{state-dependent social cost} in equilibrium is the average cost incurred by a player of any population for a given network state:
\begin{align}
\costVar_{\staterv}^{\bneIndicator} := \sum_{\service\in \serviceset} \playerFrac{\service} \costVar^{\service\bneIndicator}_{\staterv}, \quad \staterv\in\stateset.\label{eq:state_soc_cost}
\end{align}
\item[(iii)] The \emph{expected social cost} in equilibrium is the average cost incurred by a player of any population across all network states:
\begin{align}
\bar{\costVar}^{\bneIndicator} = \sum_{\service\in \serviceset}\playerFrac{\service} \sum_{\staterv\in\stateset} \prob{\staterv} \costVar^{\service\bneIndicator}_{\staterv}. \label{eq:exp_soc_cost}
\end{align}
\end{itemize}

For the sake of comparison, we define a \emph{baseline equilibrium cost} in a given state $\costVar_\staterv^{0\bneIndicator}$ as the equilibrium cost of a player of population $\typeUninf$ under the information environment $(\pInc,\fracInf=0,\accuracy^\typeInf,\accuracy^\typeUninf=0.5$). This baseline corresponds to the cost in the case where all players belong to population $\typeUninf,$ and TIS $\typeUninf$ is uninformative. Equivalently, it is the same as the realized cost for the corresponding subgame in the classical imperfect information game, where each network state corresponds to a subgame. Similarly, we define the \emph{baseline expected equilibrium cost} as $\costVar^{0\bneIndicator} = \sum_{\staterv\in\stateset} \prob{\staterv} \costVar_\staterv^{0\bneIndicator}$, which is also equivalent to the expected cost for the imperfect information game.

We now define the socially optimal cost. The socially optimal play is that which minimizes the social cost (\cite{Koutsoupias}). We distinguish the quantities associated with socially optimal play using the superscript dagger~$^\socOpt$. For a parallel-route network with linear route cost functions, the socially optimal loads $\load^{\staterv\socOpt}$ for each network state $\staterv\in\stateset$ are given by the solution to the quadratic program:
\begin{align}
\begin{split}
\min \quad&\frac{1}{2} (\load^{\staterv\socOpt})^\intercal S^{\staterv} \load^{\staterv\socOpt} + \intercept{}^\intercal \load^{\staterv\socOpt}\\
\text{subject to}\quad & \sum_\routeindex \load^{\staterv\socOpt}_\routeindex  = \totaldemand, \quad \load^{\staterv\socOpt}_\routeindex \geq 0,
\end{split}\label{eq:quad_prog_so}
\end{align}
where $S^{\staterv}$ is a diagonal matrix with elements $(2\slope{1}{\staterv},\dots,2\slope{\numroutes}{\staterv})$, and $\intercept{}$ is a column vector with elements $(\intercept{1},\dots,\intercept{\numroutes})$. The corresponding socially optimal split fraction is denoted as $\splitfractionVar^{\staterv\socOpt}$.

The \emph{state-dependent socially optimal cost} is defined as: 
\begin{align}
\costVar^\socOpt_\staterv = \sum_\routeindex \splitfractionVar^{\staterv\socOpt}_\routeindex \latency{\routeindex}{\staterv}{\load^{\staterv\socOpt}_\routeindex},\quad \staterv\in\stateset,\label{eq:socOpt_statecost}
\end{align}
and the \emph{socially optimal cost} for the game is defined as: 
\begin{align}
\costVar^\socOpt = \sum_\staterv \costVar^\socOpt_\staterv \pr(\staterv).\label{eq:socOpt_overallcost}
\end{align}
Note that the socially optimal loads and costs depend only on physical parameters of the game. We present the socially optimal loads for the two-route network in Fig.~\ref{fig:network_cartoon} for each state below:
\begin{align*}
\load_{1}^{\statenorm\socOpt} &= \frac{2 \slope{2}{} \totaldemand - \intercept{1}+\intercept{2}}{2 (\slope{1}{\statenorm}+\slope{2}{})}, \quad
 \load_{2}^{\statenorm\socOpt} = \totaldemand - \load_{1}^{\statenorm\socOpt}\\
\load_{1}^{\stateinc\socOpt} &= \frac{2 \slope{2}{} \totaldemand - \intercept{1}+\intercept{2}}{2 (\slope{1}{\stateinc}+\slope{2}{})}, \quad
\load_{2}^{\stateinc\socOpt} = \totaldemand - \load_{1}^{\stateinc\socOpt}.
\end{align*}
Thus, the socially optimal cost of the game is the $\pInc$-weighted average of the state-dependent socially optimal costs:
\begin{align}
\costVar^{\socOpt} &= \pInc \costVar_\stateinc^{\socOpt} + (1-\pInc) \costVar_\statenorm^{\socOpt}, \label{eq:socopt_cost}
\end{align}
where the state-dependent socially optimal costs are given by:
\begin{align*}
\costVar_\statenorm^{\socOpt} &= \splitfractionVar_{1}^{\statenorm\socOpt} \latency{1}{\statenorm}{\load_{1}^{\statenorm\socOpt} } + \splitfractionVar_{2}^{\statenorm\socOpt}\latency{2}{}{\load_{2}^{\statenorm\socOpt} } \\
\costVar_\stateinc^{\socOpt} &= \splitfractionVar_{1}^{\stateinc\socOpt} \latency{1}{\stateinc}{\load_{1}^{\stateinc\socOpt}} + \splitfractionVar_{2}^{\stateinc\socOpt} \latency{2}{}{\load_{2}^{\stateinc\socOpt} }.
\end{align*}

\subsection{Value of information} % (fold)
\label{sub:value_of_information}

% subsection value_of_information (end)
We now define the individual and social value of information for an information environment $(\pInc, \fracInf, \accuracy^\typeInf, \accuracy^\typeUninf)$. These quantities represent the change in equilibrium cost that an individual and society respectively receive from an information environment.

The \emph{individual value} of information, denoted $\valueVar_\staterv^\service(\pInc, \fracInf, \accuracy^\typeInf, \accuracy^\typeUninf)$, for population $\service$ in state $\staterv$ is the difference between the baseline equilibrium cost, $\costVar_\staterv^{0\bneIndicator}$, and the corresponding equilibrium cost in that state, $\costVar_\staterv^{\service\bneIndicator}$, given by \eqref{eq:state_type_cost_form}, under the information environment $(\pInc, \fracInf, \accuracy^\typeInf, \accuracy^\typeUninf)$:
\begin{align}
\valueVar_\staterv^\service(\pInc, \fracInf, \accuracy^\typeInf, \accuracy^\typeUninf) = \costVar_\staterv^{0\bneIndicator} - \costVar_\staterv^{\service\bneIndicator} \label{eq:ind_val_info_struct_state}
\end{align}
Similarly, the \emph{expected individual value} of information, denoted $\valueVar^\service(\pInc, \fracInf, \accuracy^\typeInf, \accuracy^\typeUninf)$, for population $\service$ is the difference between the baseline expected equilibrium cost and the expected equilibrium cost for a player in population $\service$ under the information environment $(\pInc, \fracInf, \accuracy^\typeInf, \accuracy^\typeUninf)$.
\begin{align}
\valueVar^\service(\pInc, \fracInf, \accuracy^\typeInf, \accuracy^\typeUninf) = \costVar^{0\bneIndicator} - \costVar^{\service\bneIndicator} \label{eq:ind_val_info_struct_exp}
\end{align}
The individual value of information represents the reduction in equilibrium cost for a player in each population due to the information environment when compared to a population of all uninformed players.

We also define the \emph{relative individual value} of information, denoted $\valueVar_\staterv(\pInc, \fracInf, \accuracy^\typeInf, \accuracy^\typeUninf)$, in a given state as the difference between the corresponding individual values for players in populations $\typeUninf$ and $\typeInf$ in that state. Similarly, we define the \emph{relative expected individual value}, denoted $\valueVar(\pInc, \fracInf, \accuracy^\typeInf, \accuracy^\typeUninf)$, as the difference between the expected individual values. Thus, these quantities can be written as follows:
\begin{align}
\valueVar_\staterv(\pInc, \fracInf, \accuracy^\typeInf, \accuracy^\typeUninf) := \valueVar_\staterv^\typeUninf(\pInc, \fracInf, \accuracy^\typeInf, \accuracy^\typeUninf)-\valueVar_\staterv^\typeInf(\pInc, \fracInf, \accuracy^\typeInf, \accuracy^\typeUninf), \label{eq:rel_ind_val_inf_state}\\
\valueVar(\pInc, \fracInf, \accuracy^\typeInf, \accuracy^\typeUninf) := \valueVar^\typeUninf(\pInc, \fracInf, \accuracy^\typeInf, \accuracy^\typeUninf)-\valueVar^\typeInf(\pInc, \fracInf, \accuracy^\typeInf, \accuracy^\typeUninf). \label{eq:rel_ind_val_inf}
\end{align}
The relative individual value represents the change in equilibrium cost that a player in population $\typeUninf$ would experience if she could become a member of population $\typeInf$, \emph{ceteris paribus.}

The \emph{social value} of information represents the change in social cost due to an information environment compared to the baseline equilibrium cost. Specifically, the social value of an information in a given state, denoted $\welfare_{\staterv}(\pInc, \fracInf, \accuracy^\typeInf, \accuracy^\typeUninf)$, is the difference between the baseline social cost $\costVar_\staterv^{0\bneIndicator}$ and the corresponding equilibrium social cost $\costVar_{\staterv}^\bneIndicator$ in that state:
\begin{align}
\welfare_{\staterv}(\pInc, \fracInf, \accuracy^\typeInf, \accuracy^\typeUninf) := \costVar_\staterv^{0\bneIndicator} - \costVar_{\staterv}^\bneIndicator. \label{eq:soc_val_of_info_state}
\end{align}
Similarly, we define \emph{expected social value} of information, denoted $\welfare(\pInc, \fracInf, \accuracy^\typeInf, \accuracy^\typeUninf)$, as the difference between the expected baseline equilibrium cost and the equilibrium expected social cost. Equivalently, the expected social value can be calculated as the average of the social value in each state, weighted by the probability of each state:
\begin{align}
\welfare(\pInc, \fracInf, \accuracy^\typeInf, \accuracy^\typeUninf) := \costVar^{0\bneIndicator} - \bar{\costVar}^\bneIndicator = \sum_{\staterv}\welfare_\staterv(\pInc, \fracInf, \accuracy^\typeInf, \accuracy^\typeUninf) \pr{(\staterv)}. \label{eq:exp_soc_val_of_info}
\end{align}

% subsection bayesian_congestion_game (end)

%% file: beliefs.tex
%!TEX root = main.tex

\section{Beliefs}
\label{sec:beliefs}

Recall that in Sec.~\ref{ssub:populations_of_nonatomic_players} the Bayesian congestion game $\genGame_p$ is defined terms of the interim beliefs $\belief^\typeInf\left(\state,\playertype^{\typeUninf}|\playertype^\typeInf\right)$ for population $\typeInf$ and $\belief^\typeUninf\left(\state,\playertype^{\typeInf}|\playertype^\typeUninf\right)$ for population $\typeUninf$, where the type~$\playertype^\service\in\playertypeset_\service$ encapsulates the private information of players in population~$\service$ about their TIS's accuracy and the received signal. Each player's interim belief can be viewed as the posterior distribution obtained using a Bayesian update to a common prior $\commprior\in\Delta(\stateset\times\playertypeset^\typeUninf\times \playertypeset^\typeInf)$ after receiving their type. In our presentation of interim beliefs and subsequent equilibrium characterization, we limit our attention to game $\genGame_p$ of non-atomic players; however, our set-up is extensible to the game $\genGame_f$ of finite-atomic players (Sec.~\ref{sec:bayesian_congestion_game}).

In this section, we present three different treatments (environments) for obtaining interim beliefs; each treatment supposes a different information structure: 
\begin{itemize}
	\item[(i)] In Sec.~\ref{sub:conditional_likelihoods_as_common_knowledge}, we consider the TIS likelihoods~\eqref{eq:likelihoods} as part of prior belief, i.e. both populations have a common prior belief $\tilde{\commprior}=\pr(\state)\cdot\pr(\playertype^{\typeInf}|\staterv)\cdot\pr(\playertype^{\typeUninf}|\staterv)$. This corresponds to an environment where each population knows the distribution of nature state, the likelihood distribution of its own TIS, and that of the other TIS. In this case, players from each population form their interim beliefs using the conditional distribution of the other population's type for each state of the nature.
 	\item[(ii)] In Sec.~\ref{sub:marginal_likelihoods_ck}, we suppose a more restrictive information structure in that the marginal type distributions $\pr(\playertype^\service)=\sum_{\state} \pr(\state)\pr(\playertype^\service|\state),\; \service\in\{\typeInf,\typeUninf\}$ are part of prior belief, i.e., the common prior belief is  $\hat{\commprior}=\pr(\state)\cdot\pr(\playertype^{\typeInf})\cdot\pr(\playertype^{\typeUninf})$. In this case, players do not know the TIS likelihood of those in the other population. Thus, each population's interim beliefs are constructed using the distribution of nature state, the likelihood distribution of its own TIS, and the \emph{marginal} type distribution of the other population. 
	\item[(iii)] Finally, in Sec.~\ref{sub:uninformative_tis_}, we consider a \emph{special case} of the information structure in Sec.~\ref{sub:marginal_likelihoods_ck} where the TIS~$\typeUninf$ is uninformative, i.e., $0.5=\accuracy^\typeUninf < \accuracy^\typeInf \leq 1$. The population $\typeUninf$'s information service does not give any information about Nature's state, i.e. TIS $\typeUninf$ is equally likely to report any state, regardless of the true state of nature, and population $\typeInf$ is better informed about the Nature state for any $\accuracy^\typeInf>0.5$.
	\end{itemize}
	
We will denote the population~$\service$'s interim beliefs for these cases as $\tilde{\belief}^\service$, $\hat{\belief}^\service$, and $\bar{\belief}^\service$, respectively. In all three cases, the interim beliefs can be generically expressed as follows:
\begin{align*}
\forall \service \in \serviceset, \quad \belief^\service\left(\state,\playertype^{-\service}|\playertype^\service\right) &=\frac{\pr^\service(\state,\playertype^\service,\playertype^{-\service})}{\pr(\playertype^\service)}\nonumber\\
&=\frac{\pr(\state)\pr^\service(\playertype^\service,\playertype^{-\service}|\state)}{\pr(\playertype^\service)}\nonumber\\
&= \underbrace{\frac{\pr(\state)\pr^\service(\playertype^{\service}|\state)}{\pr(\playertype^{\service})}}_{\pr^\service(\state|\playertype^\service)}\pr^\service(\playertype^{-\service}|\playertype^{\service},\state). 
%\label{eq:belief_form_w_bayes}
\end{align*}
%, $\pr^\service(\playertype^\service)$ is calculated as:
%\begin{align*}
%\pr^\service(\playertype^\service) = \sum_{\state} \pr(\state)\pr(\playertype^\service|\state).
%\end{align*}
%The expressions of $\pr(\playertype^\service)$ for all four types are given below:

Now notice that in our modeling environment $\playertype^\service$ and $\playertype^{-\service}$ are independent, conditioned on $\state$. Therefore, we can express the interim beliefs as follows: 
\begin{align}
\forall \service \in \serviceset, \quad\belief^\service\left(\state,\playertype^{-\service}|\playertype^\service\right)=\frac{\pr(\state)\pr^\service(\playertype^{\service}|\state)}{\pr(\playertype^{\service})}\pr^\service(\playertype^{-\service}|\state). \label{eq:belief_formimp}
\end{align}
%That is, given the state of Nature, knowledge of one's own type does not give any additional information about other players' types.

Before moving to the specific cases, we compute the marginal type distributions~$\pr(\playertype^{\service})$. From~\eqref{eq:likelihoods}, we note that $\pr(\playertype^\typeInf=\typeInf\stateinc|\state=\stateinc)=\pr(\obsrv^\typeInf|\state=\stateinc)=\accuracy^\typeInf$. Writing similar expressions for the conditional probability of each $\playertype^\service\in\playertypeset_\service$ given the occurrence of a nature state, we can express the marginal type distributions $\pr(\playertype^\typeInf = \typeInf\stateinc)$ and $\pr(\playertype^\typeUninf = \typeInf\stateinc)$ as follows: 
\begin{align}
\begin{split}
%\pr(\playertype^\typeInf &= \typeInf\statenorm) = \pInc (1-\accuracy^\typeInf) + (1-\pInc) \accuracy^\typeInf\\
\pr(\playertype^\typeInf &= \typeInf\stateinc) = \pInc\accuracy^\typeInf + (1-\pInc)(1- \accuracy^\typeInf)\\
%\pr(\playertype^\typeUninf &= \typeUninf\statenorm) = \pInc (1-\accuracy^\typeUninf) + (1-\pInc) \accuracy^\typeUninf\\
\pr(\playertype^\typeUninf &= \typeUninf\stateinc) = \pInc\accuracy^\typeUninf + (1-\pInc)(1- \accuracy^\typeUninf).\label{eq:prob_of_types}
\end{split}
\end{align}
Also note that $\pr(\playertype^\typeInf=\typeInf\statenorm) = 1-\pr(\playertype^\typeInf=\typeInf\stateinc)$, and $ \pr(\playertype^\typeUninf=\typeUninf\statenorm)= 1-\pr(\playertype^\typeUninf=\typeUninf\stateinc)$. For notational ease, we will henceforth refer to $\pr(\playertype^\typeInf = \typeInf\statenorm)$ as $\pr(\typeInf\statenorm)$, and $\pr(\playertype^\typeInf = \typeInf\stateinc)$ as $\pr(\typeInf\stateinc)$, etc. 

\subsection{TIS Conditional Likelihoods are Common Knowledge}
\label{sub:conditional_likelihoods_as_common_knowledge}

In this case, players in each population know the TIS likelihood distribution of the other population. Thus, $\pr^\typeUninf(\playertype^\typeInf=\typeInf\stateinc|\state=\stateinc)=\accuracy^\typeInf$ and $\pr^\typeUninf(\playertype^\typeInf=\typeInf\stateinc|\state=\statenorm)=(1-\accuracy^\typeInf)$. Substituting the expressions of $\pr^\typeUninf(\playertype^{\service}|\state)$ and $\pr^\typeUninf(\playertype^{\typeInf}|\state)$ in~\eqref{eq:belief_formimp}, the population~$\typeUninf$'s interim belief $\tilde{\belief}^\typeUninf$ can be written as: 
% 
%Suppose the TIS accuracies, or conditional likelihoods, are common knowledge; this allows players to explicitly compute \eqref{eq:belief_formimp}. This scenario corresponds to the case where players know the properties of not only their own TIS, but also that of the other players. 
%
%We present the expressions of the posterior belief distributions for both populations below and denote beliefs under common knowledge of conditional likelihoods as $\tilde{\belief}$. For population $\typeUninf$:
\begin{align*}
\tilde{\belief}^\typeUninf(\state = \stateinc,\playertype^\typeInf = \typeInf\stateinc | \playertype^\typeUninf = \typeUninf\stateinc) &= \frac{\pInc \accuracy^\typeUninf}{\pr(\typeUninf\stateinc)}\accuracy^\typeInf\\
\tilde{\belief}^\typeUninf(\state = \stateinc,\playertype^\typeInf = \typeInf\statenorm | \playertype^\typeUninf = \typeUninf\stateinc) &= \frac{\pInc \accuracy^\typeUninf}{\pr(\typeUninf\stateinc)}(1-\accuracy^\typeInf)\\
\tilde{\belief}^\typeUninf(\state = \statenorm,\playertype^\typeInf = \typeInf\stateinc | \playertype^\typeUninf = \typeUninf\stateinc) &= \frac{(1-\pInc)(1- \accuracy^\typeUninf)}{\pr(\typeUninf\stateinc)}(1-\accuracy^\typeInf)\\
\tilde{\belief}^\typeUninf(\state = \statenorm,\playertype^\typeInf = \typeInf\statenorm | \playertype^\typeUninf = \typeUninf\stateinc) &= \frac{(1-\pInc)(1- \accuracy^\typeUninf)}{\pr(\typeUninf\stateinc)}\accuracy^\typeInf\\
\tilde{\belief}^\typeUninf(\state = \stateinc,\playertype^\typeInf = \typeInf\stateinc | \playertype^\typeUninf = \typeUninf\statenorm) &= \frac{\pInc(1- \accuracy^\typeUninf)}{\pr(\typeUninf\statenorm)}\accuracy^\typeInf\\
\tilde{\belief}^\typeUninf(\state = \stateinc,\playertype^\typeInf = \typeInf\statenorm | \playertype^\typeUninf = \typeUninf\statenorm) &= \frac{\pInc(1- \accuracy^\typeUninf)}{\pr(\typeUninf\statenorm)}(1-\accuracy^\typeInf)\\
\tilde{\belief}^\typeUninf(\state = \statenorm,\playertype^\typeInf = \typeInf\stateinc | \playertype^\typeUninf = \typeUninf\statenorm) &= \frac{\pInc\accuracy^\typeUninf}{\pr(\typeUninf\statenorm)}(1-\accuracy^\typeInf)\\
\tilde{\belief}^\typeUninf(\state = \statenorm,\playertype^\typeInf = \typeInf\statenorm | \playertype^\typeUninf = \typeUninf\statenorm) &= \frac{\pInc\accuracy^\typeUninf}{\pr(\typeUninf\statenorm)}\accuracy^\typeInf.
\end{align*}
% \begin{align}
% \begin{split}
% \belief^\typeUninf(\state=\stateinc,\playertype^\typeInf=\typeInf\stateinc|\playertype^\typeUninf=\typeUninf) &= \pInc \accuracy^\typeInf\\
% \belief^\typeUninf(\state=\stateinc,\playertype^\typeInf=\typeInf\statenorm|\playertype^\typeUninf=\typeUninf) &= \pInc (1-\accuracy^\typeInf)\\
% \belief^\typeUninf(\state=\statenorm,\playertype^\typeInf=\typeInf\stateinc|\playertype^\typeUninf=\typeUninf) &= (1-\pInc) (1-\accuracy^\typeInf)\\
% \belief^\typeUninf(\state=\statenorm,\playertype^\typeInf=\typeInf\statenorm|\playertype^\typeUninf=\typeUninf) &= (1-\pInc) \accuracy^\typeInf
% \end{split}\label{eq:full_obj_beliefs_L}
% \end{align}
Similarly, under the knowledge of $\pr(\playertype^{\typeUninf}|\state)$, the population~$\typeInf$'s interim belief $\tilde{\belief}^\typeInf$ can be written as follows:
\begin{align*}
\tilde{\belief}^\typeInf(\state = \stateinc,\playertype^\typeUninf = \typeUninf\stateinc | \playertype^\typeInf = \typeInf\stateinc) &= \frac{\pInc \accuracy^\typeInf}{\pr(\typeInf\stateinc)}\accuracy^\typeUninf\\
\tilde{\belief}^\typeInf(\state = \stateinc,\playertype^\typeUninf = \typeUninf\statenorm | \playertype^\typeInf = \typeInf\stateinc) &= \frac{\pInc \accuracy^\typeInf}{\pr(\typeInf\stateinc)}(1-\accuracy^\typeUninf)\\
\tilde{\belief}^\typeInf(\state = \statenorm,\playertype^\typeUninf = \typeUninf\stateinc | \playertype^\typeInf = \typeInf\stateinc) &= \frac{(1-\pInc) (1-\accuracy^\typeInf)}{\pr(\typeInf\stateinc)}(1-\accuracy^\typeUninf)\\
\tilde{\belief}^\typeInf(\state = \statenorm,\playertype^\typeUninf = \typeUninf\statenorm | \playertype^\typeInf = \typeInf\stateinc) &= \frac{(1-\pInc) (1-\accuracy^\typeInf)}{\pr(\typeInf\stateinc)}\accuracy^\typeUninf\\
\tilde{\belief}^\typeInf(\state = \stateinc,\playertype^\typeUninf = \typeUninf\stateinc | \playertype^\typeInf = \typeInf\statenorm) &= \frac{\pInc (1-\accuracy^\typeInf)}{\pr(\typeInf\statenorm)}\accuracy^\typeUninf\\
\tilde{\belief}^\typeInf(\state = \stateinc,\playertype^\typeUninf = \typeUninf\statenorm | \playertype^\typeInf = \typeInf\statenorm) &= \frac{\pInc (1-\accuracy^\typeInf)}{\pr(\typeInf\statenorm)}(1-\accuracy^\typeUninf)\\
\tilde{\belief}^\typeInf(\state = \statenorm,\playertype^\typeUninf = \typeUninf\stateinc | \playertype^\typeInf = \typeInf\statenorm) &= \frac{(1-\pInc) \accuracy^\typeInf}{\pr(\typeInf\statenorm)}(1-\accuracy^\typeUninf)\\
\tilde{\belief}^\typeInf(\state = \statenorm,\playertype^\typeUninf = \typeUninf\statenorm | \playertype^\typeInf = \typeInf\statenorm) &= \frac{(1-\pInc) \accuracy^\typeInf}{\pr(\typeInf\statenorm)}\accuracy^\typeUninf.
\end{align*}

% Now, for a strategy profile $(\load^\service,\load^{-\service}) \in \loadset^\service \times \loadset^{-\service}$, we can use the beliefs specified in \eqref{eq:full_obj_beliefs_L}-\eqref{eq:full_obj_beliefs_H} to compute the expected route costs for each population \eqref{eq:expected_cost_population}.

% For population $\typeUninf$, we have the expected route costs as follows:
% \begin{align}
% \mathbb{E}[\latency{\routeindex}{\state}{\load_\routeindex^\typeUninf+\load_\routeindex^\typeInf}|\playertype^\service = \typeUninf] =& \pInc\accuracy^\typeInf \latency{\routeindex}{\stateinc}{\load_\routeindex^\typeUninf+\load_\routeindex^{\typeInf\stateinc}} + \pInc(1-\accuracy^\typeInf)\latency{\routeindex}{\stateinc}{\load_\routeindex^\typeUninf+\load_\routeindex^{\typeInf\statenorm}}\nonumber \\
% &+ (1-\pInc)(1-\accuracy^\typeInf)\latency{\routeindex}{\statenorm}{\load_\routeindex^\typeUninf+\load_\routeindex^{\typeInf\stateinc}} + (1-\pInc)\accuracy^\typeInf \latency{\routeindex}{\statenorm}{\load_\routeindex^\typeUninf+\load_\routeindex^{\typeInf\statenorm}}\label{eq:exp_latency_typeL_obj}.
% \end{align}
% For population $\typeInf$, we must account for the expected route costs for both $\typeInf\stateinc$ \& $\typeInf\statenorm$. 

The common knowledge of TIS conditional likelihoods might be an appropriate assumption when TIS providers are competitive firms who have researched their competitor's TIS. However, in a large class of route choice scenarios faced by commuters, a more realistic assumption on the information structure is the case when only marginal type distributions are common knowledge, which we consider next. 

\subsection{TIS Marginal Type Distributions are Common Knowledge}
\label{sub:marginal_likelihoods_ck}
In this case, each population does not know the conditional distribution of their opponent's type, but knows the opponent's marginal type distribution. In contrast to~\eqref{eq:belief_formimp}, each player's interim beliefs under this restrictive information structure can be expressed as follows: 
\begin{align}\label{eq:belief_marg_formimp}
\forall \service \in \serviceset, \quad\hat{\belief}^\service\left(\state,\playertype^{-\service}|\playertype^\service\right)=\frac{\pr(\state)\pr^\service(\playertype^{\service}|\state)}{\pr(\playertype^{\service})}\pr(\playertype^{-\service}).
\end{align}

%We now address the scenario where the conditional likelihoods are not common knowledge, and players only know the marginal likelihoods of encountering an opponent of a given type. We denote these beliefs as $\hat{\belief}$. In this case, players cannot compute \eqref{eq:belief_formimp} directly, as they do not know $\pr(\playertype^{-\service}|\staterv)$. Instead, they substitute the marginal likelihood $\pr(\playertype^{-\service}$ and compute their belief as:

Using~\eqref{eq:prob_of_types} in~\eqref{eq:belief_marg_formimp}, we can write the population $\typeUninf$'s interim belief $\hat{\belief}^\typeUninf$ as follows: 
%The full posteriors under this knowledge structure are calculated below. For population $\typeUninf$, the posterior beliefs are:
\begin{align*}
\hat{\belief}^\typeUninf(\state = \stateinc,\playertype^\typeInf = \typeInf\stateinc | \playertype^\typeUninf = \typeUninf\stateinc) &= \frac{\pInc \accuracy^\typeUninf}{\pr(\typeUninf\stateinc)}\pr(\typeInf\stateinc)\\
\hat{\belief}^\typeUninf(\state = \stateinc,\playertype^\typeInf = \typeInf\statenorm | \playertype^\typeUninf = \typeUninf\stateinc) &= \frac{\pInc \accuracy^\typeUninf}{\pr(\typeUninf\stateinc)}\pr(\typeInf\statenorm)\\
\hat{\belief}^\typeUninf(\state = \statenorm,\playertype^\typeInf = \typeInf\stateinc | \playertype^\typeUninf = \typeUninf\stateinc) &= \frac{(1-\pInc)(1- \accuracy^\typeUninf)}{\pr(\typeUninf\stateinc)}\pr(\typeInf\stateinc)\\
\hat{\belief}^\typeUninf(\state = \statenorm,\playertype^\typeInf = \typeInf\statenorm | \playertype^\typeUninf = \typeUninf\stateinc) &= \frac{(1-\pInc)(1- \accuracy^\typeUninf)}{\pr(\typeUninf\stateinc)}\pr(\typeInf\statenorm)\\
\hat{\belief}^\typeUninf(\state = \stateinc,\playertype^\typeInf = \typeInf\stateinc | \playertype^\typeUninf = \typeUninf\statenorm) &= \frac{\pInc(1- \accuracy^\typeUninf)}{\pr(\typeUninf\statenorm)}\pr(\typeInf\stateinc)\\
\hat{\belief}^\typeUninf(\state = \stateinc,\playertype^\typeInf = \typeInf\statenorm | \playertype^\typeUninf = \typeUninf\statenorm) &= \frac{\pInc(1- \accuracy^\typeUninf)}{\pr(\typeUninf\statenorm)}\pr(\typeInf\statenorm)\\
\hat{\belief}^\typeUninf(\state = \statenorm,\playertype^\typeInf = \typeInf\stateinc | \playertype^\typeUninf = \typeUninf\statenorm) &= \frac{\pInc\accuracy^\typeUninf}{\pr(\typeUninf\statenorm)}\pr(\typeInf\stateinc)\\
\hat{\belief}^\typeUninf(\state = \statenorm,\playertype^\typeInf = \typeInf\statenorm | \playertype^\typeUninf = \typeUninf\statenorm) &= \frac{\pInc\accuracy^\typeUninf}{\pr(\typeUninf\statenorm)}\pr(\typeInf\statenorm).
\end{align*}
Similarly, when the population $\typeInf$ knows the marginal type distribution~$\pr(\playertype^{\typeUninf})$ but does not know $\pr(\playertype^\typeUninf|\state)$, population $\typeInf$'s interim belief distribution $\hat{\belief}^\typeInf$ can be expressed as follows:
\begin{align*}
\hat{\belief}^\typeInf(\state = \stateinc,\playertype^\typeUninf = \typeUninf\stateinc | \playertype^\typeInf = \typeInf\stateinc) &= \frac{\pInc \accuracy^\typeInf}{\pr(\typeInf\stateinc)}\pr(\typeUninf\stateinc)\\
\hat{\belief}^\typeInf(\state = \stateinc,\playertype^\typeUninf = \typeUninf\statenorm | \playertype^\typeInf = \typeInf\stateinc) &= \frac{\pInc \accuracy^\typeInf}{\pr(\typeInf\stateinc)}\pr(\typeUninf\statenorm)\\
\hat{\belief}^\typeInf(\state = \statenorm,\playertype^\typeUninf = \typeUninf\stateinc | \playertype^\typeInf = \typeInf\stateinc) &= \frac{(1-\pInc) (1-\accuracy^\typeInf)}{\pr(\typeInf\stateinc)}\pr(\typeUninf\stateinc)\\
\hat{\belief}^\typeInf(\state = \statenorm,\playertype^\typeUninf = \typeUninf\statenorm | \playertype^\typeInf = \typeInf\stateinc) &= \frac{(1-\pInc) (1-\accuracy^\typeInf)}{\pr(\typeInf\stateinc)}\pr(\typeUninf\statenorm)\\
\hat{\belief}^\typeInf(\state = \stateinc,\playertype^\typeUninf = \typeUninf\stateinc | \playertype^\typeInf = \typeInf\statenorm) &= \frac{\pInc (1-\accuracy^\typeInf)}{\pr(\typeInf\statenorm)}\pr(\typeUninf\stateinc)\\
\hat{\belief}^\typeInf(\state = \stateinc,\playertype^\typeUninf = \typeUninf\statenorm | \playertype^\typeInf = \typeInf\statenorm) &= \frac{\pInc (1-\accuracy^\typeInf)}{\pr(\typeInf\statenorm)}\pr(\typeUninf\statenorm)\\
\hat{\belief}^\typeInf(\state = \statenorm,\playertype^\typeUninf = \typeUninf\stateinc | \playertype^\typeInf = \typeInf\statenorm) &= \frac{(1-\pInc) \accuracy^\typeInf}{\pr(\typeInf\statenorm)}\pr(\typeUninf\stateinc)\\
\hat{\belief}^\typeInf(\state = \statenorm,\playertype^\typeUninf = \typeUninf\statenorm | \playertype^\typeInf = \typeInf\statenorm) &= \frac{(1-\pInc) \accuracy^\typeInf}{\pr(\typeInf\statenorm)}\pr(\typeUninf\statenorm).
\end{align*}
% \begin{align}
% \begin{split}
% \belief^\typeInf(\state=\stateinc, \playertype^\typeUninf = \typeUninf|\playertype^\typeInf = \typeInf\stateinc) &= \frac{\pInc\accuracy^\typeInf}{\prob{\typeInf\stateinc}}\\
% \belief^\typeInf(\state=\statenorm, \playertype^\typeUninf = \typeUninf|\playertype^\typeInf = \typeInf\stateinc) &=\frac{(1-\pInc)(1-\accuracy^\typeInf)}{\prob{\typeInf\stateinc}}\\
% \belief^\typeInf(\state=\stateinc, \playertype^\typeUninf = \typeUninf|\playertype^\typeInf = \typeInf\statenorm) &= \frac{\pInc(1-\accuracy^\typeInf)}{\prob{\typeInf\statenorm}}\\
% \belief^\typeInf(\state=\statenorm, \playertype^\typeUninf = \typeUninf|\playertype^\typeInf = \typeInf\statenorm) &= \frac{(1-\pInc)\accuracy^\typeInf}{\prob{\typeInf\statenorm}}.
% \end{split}
% \end{align}
% The posterior beliefs of population $\typeInf$ on population $\typeUninf$ are identical to the beliefs from the case where the conditional likelihoods were common knowledge \eqref{eq:full_obj_beliefs_H}. This is because under Assumption \ref{ass:type_L_acc}, the conditional and marginal likelihoods of TIS $\typeUninf$ are identical: $\pr(\playertype^\typeUninf=\typeUninf|\staterv) = \pr(\playertype^\typeUninf=\typeUninf)=1$.

\subsection{Uninformative TIS $\typeUninf$ ($\accuracy^\typeUninf=0.5$)}
\label{sub:uninformative_tis_}

% We argue that this information structure provides a more realistic model of commuter traffic because it does not require players to know the accuracy of the other players' TIS. Of course, there are scenarios where it is appropriate to assume common knowledge of conditional likelihoods, e.g. competition between firms that have researched their competitor's TIS. However, such models fall outside the scope of this article.

%\subsection{Uninformative TIS $\typeUninf$} % (fold)

%
%The beliefs presented in Sec.~\ref{sub:conditional_likelihoods_as_common_knowledge} and Sec.~\ref{sub:marginal_likelihoods_ck} capture information heterogeneity for general $\accuracy^\typeInf$ and $\accuracy^\typeUninf$ (subject to \eqref{eq:accuracy}).

We now consider a specific case of Sec.~\ref{sub:marginal_likelihoods_ck} when $\accuracy^\typeUninf=0.5$, i.e. $\pr(\playertype^\typeUninf)=\pr(\playertype^\typeUninf|\state)$. This case is of special interest to us because it allows a population $\typeUninf$ that knowns no more about the nature state than the prior distribution $\pr(\state)$, and is also uniformed about the TIS $\typeInf$'s conditional likelihood $\pr(\playertype^\typeInf|\state)$. Hence, in this case, population $\typeUninf$ receives no private information.

%\footnote{In these assumptions are applied to the interim beliefs in Sec.~\ref{sub:conditional_likelihoods_as_common_knowledge} and Sec.~\ref{sub:marginal_likelihoods_ck}, we obtain that $\tilde{\belief}^\typeInf=\hat{\belief}^\typeInf$, but $\tilde{\belief}^\typeUninf\neq\hat{\belief}^\typeUninf$.}

%This captures environments in which population $\typeUninf$'s information service does not give any information about Nature's state, i.e. TIS $\typeUninf$ is equally likely to report any state, regardless of the true state of nature. 

%Formally, we will maintain the following assumption throughout the rest of the article:
%\begin{assumption} \label{ass:type_L_acc} TIS $\typeUninf$ is completely uninformative, i.e.:
%\begin{align}
%0.5=\accuracy^\typeUninf < \accuracy^\typeInf \leq 1.\label{eq:accuracy_uninf}
%\end{align}
%\end{assumption}

It is easy to argue that when $\accuracy^\typeUninf = 0.5$, the signal that any player of type $\typeUninf\stateinc$ or $\typeUninf\statenorm$ receives is completely uninformative of the state of nature; hence, the belief of any player subscribed to TIS $\typeUninf$ is no more informative than common prior belief~$\hat\commprior$. Since the commuters of type $\typeUninf\statenorm$ and $\typeUninf\stateinc$ have identical beliefs about the nature state and about the population $\typeInf$'s type, they also have identical strategies. Furthermore, commuters in the population~$\typeInf$ have identical beliefs about type $\typeUninf\statenorm$ and $\typeUninf\stateinc$. Thus, $\typeUninf\stateinc$ and $\typeUninf\statenorm$ no longer need to be distinguished. Henceforth, we will consider only a single equivalent type that represents both $\typeUninf\stateinc$ and $\typeUninf\statenorm$; we refer to it as type $\typeUninf$. Indeed, if $\accuracy^\typeUninf > 0.5$, it is not possible make this simplification.

%Under Assumption \ref{ass:type_L_acc}, 
Under the aforementioned simplification, population $\typeUninf$'s interim belief $\bar{\belief}^\typeUninf$ can be expressed as follows:
\begin{align}
\begin{split}
\bar{\belief}^\typeUninf(\state=\stateinc,\playertype^\typeInf=\typeInf\stateinc|\playertype^\typeUninf=\typeUninf) &= \pInc \pr(\typeInf\stateinc)\\
\bar{\belief}^\typeUninf(\state=\stateinc,\playertype^\typeInf=\typeInf\statenorm|\playertype^\typeUninf=\typeUninf) &= \pInc \pr(\typeInf\statenorm)\\
\bar{\belief}^\typeUninf(\state=\statenorm,\playertype^\typeInf=\typeInf\stateinc|\playertype^\typeUninf=\typeUninf) &= (1-\pInc) \pr(\typeInf\stateinc)\\
\bar{\belief}^\typeUninf(\state=\statenorm,\playertype^\typeInf=\typeInf\statenorm|\playertype^\typeUninf=\typeUninf) &= (1-\pInc) \pr(\typeInf\statenorm).
\end{split}\label{eq:full_subj_beliefs_L}
\end{align}
Finally, when the population $\typeInf$ players still know the marginal type distribution~$\pr(\playertype^{\typeUninf})$ and face an uninformed population $\typeUninf$ ($\accuracy^\typeUninf = 0.5$), their interim belief distribution $\bar{\belief}^\typeInf$ can be written as follows: 
\begin{align}
\begin{split}
\bar{\belief}^\typeInf(\state = \stateinc,\playertype^\typeUninf = \typeUninf | \playertype^\typeInf = \typeInf\stateinc) &= \frac{\pInc \accuracy^\typeInf}{\prob{\typeInf\stateinc}}\\
\bar{\belief}^\typeInf(\state = \statenorm,\playertype^\typeUninf = \typeUninf| \playertype^\typeInf = \typeInf\stateinc) &= \frac{(1-\pInc) (1-\accuracy^\typeInf)}{\prob{\typeInf\stateinc}}\\
\bar{\belief}^\typeInf(\state = \stateinc,\playertype^\typeUninf = \typeUninf | \playertype^\typeInf = \typeInf\statenorm) &= \frac{\pInc (1-\accuracy^\typeInf)}{\prob{\typeInf\statenorm}}\\
\bar{\belief}^\typeInf(\state = \statenorm,\playertype^\typeUninf = \typeUninf | \playertype^\typeInf = \typeInf\statenorm) &= \frac{(1-\pInc) \accuracy^\typeInf}{\prob{\typeInf\statenorm}}
\end{split}\label{eq:full_subj_beliefs_H}
\end{align}

%% file: equilibrium.tex
%!TEX root = main.tex

\section{Equilibrium}
\label{sec:equilibrium}
In this section, we present a full characterization the Bayesian Wardrop Equilibrium (BWE) of the game $\genGame_p=\left(\serviceset,{\loadset},{\stateset},{\playertypeset},{\costVar},{\commprior} \right)$ with the interim beliefs $\bar{\belief}$ as in \eqref{eq:full_subj_beliefs_L} and \eqref{eq:full_subj_beliefs_H}, i.e. under the information structure induced by $\accuracy^\typeUninf = 0.5$ and when TIS marginal type distributions are common knowledge. In our subsequent analysis, we use the notation $\load_\routeindex^{\playertype^{\service}}$ to denote the load $\load_\routeindex^{\service}(\playertype^\service)$ on route $\routeindex$ of population $\service$ when they are type $\playertype^\service$. Similarly, we use the notation $\splitfraction{\playertype^\service}{1}{}$ to denote the strategy distribution (or split fraction) $\splitfraction{\playertype^\service}{1}{}(\playertype^\service)$ on route $\routeindex$ of population $\service$ when they are type $\playertype^\service$. Recall that load and split fraction are related by \eqref{eq:splitfraction_load_relation}. 

Without loss of generality, we will assume that the total demand is sufficiently large that the cost of placing all of the demand on the first route in the normal state is costlier than the free-flow cost of the second route, i.e.:
\begin{align}
\totaldemand > \frac{\intercept{2}-\intercept{1}}{\slope{1}{\statenorm}} \label{eq:demand_assumption}.
\end{align}
This assumption guarantees that there is no equilibrium where all the demand is routed on the same route.

\subsection{Equilibrium Characterization} % (fold)
\label{sub:characterization}
In the two-route network (Fig~\ref{fig:network_cartoon}), the BWE conditions (Definition \ref{def:bayesian_wardrop}) can be interpreted as follows: if all players in a population play exclusively one route in equilibrium when they are type $\playertype^\service$, i.e. $\splitfraction{\playertype^\service}{1}{\bneIndicator}=$ 0 or 1, then the expected cost of the route taken must be less than that of the other route. On the other hand, if the players of a population split and take both routes when they are type $\playertype^\service$, i.e. $\splitfraction{\playertype^\service}{1}{\bneIndicator} \in (0,1)$, then the expected cost on both routes are equal. This creates three qualitatively different strategy distributions for each player type: exclusively take $\routeindex_1$, exclusively take $\routeindex_2$, or split some demand on each. These are expressed below:
\begin{subequations}
\begin{align}[left={\forall \service \in \serviceset, \playertype^\service \in \playertypeset_\service, \empheqlbrace\,}]
&\splitfraction{\playertype^\service}{1}{\bneIndicator}=0 \implies \mathbb{E}_{\bar{\belief}}[\latency{1}{\state}{\load_1^{\playertype^{-\service}\bneIndicator}}|\playertype^\service] < \mathbb{E}_{\bar{\belief}}[\latency{2}{\state}{\totaldemand^{\playertype\bneIndicator}+\load_2^{\playertype^{-\service}\bneIndicator}}|\playertype^\service]\label{eq:bne_ineq_0}\\
&\splitfraction{\playertype^\service}{1}{\bneIndicator}=1 \implies \mathbb{E}_{\bar{\belief}}[\latency{1}{\state}{\totaldemand^{\playertype^\service\bneIndicator}+\load_1^{\playertype^{-\service}\bneIndicator}}|\playertype^\service] > \mathbb{E}_{\bar{\belief}}[\latency{2}{\state}{\load_2^{\playertype^{-\service}\bneIndicator}}|\playertype^\service ]\label{eq:bne_ineq_1}\\
&\splitfraction{\playertype^\service}{1}{\bneIndicator}\in(0,1) \implies \mathbb{E}_{\bar{\belief}}[\latency{1}{\state}{\load_1^{\playertype^\service\bneIndicator}+\load_1^{\playertype^{-\service}\bneIndicator}}|\playertype^\service] = \mathbb{E}_{\bar{\belief}}[\latency{2}{\state}{\load_2^{\playertype^\service\bneIndicator}+\load_2^{\playertype^{-\service}\bneIndicator}}|\playertype^\service ].\label{eq:bne_equality}
\end{align}
\end{subequations}
Since there are three player types ($\typeUninf, \typeInf\statenorm, \typeInf\stateinc$), this leads to $3^3=27$ combinations of qualitatively different strategy distribution profiles that must be considered.

The expected route costs \eqref{eq:expected_cost_population} according to the interim beliefs $\bar{\belief}$ given in \eqref{eq:full_subj_beliefs_L} for players in population $\typeUninf$ on route $\routeindex$ can be expressed as follows:
\begin{align}
\mathbb{E}_{\bar{\belief}}[\latency{\routeindex}{\state}{\load_\routeindex^\typeUninf+\load_\routeindex^{\playertype^\typeInf}}|\playertype^\typeUninf= \typeUninf] =& \bar{\belief}^\typeUninf(\staterv=\stateinc,\playertype^\typeInf = \typeInf\stateinc) \latency{\routeindex}{\stateinc}{\load_\routeindex^\typeUninf+\load_\routeindex^{\typeInf\stateinc}} + \bar{\belief}^\typeUninf(\staterv=\stateinc,\playertype^\typeInf = \typeInf\statenorm) \latency{\routeindex}{\stateinc}{\load_\routeindex^\typeUninf+\load_\routeindex^{\typeInf\statenorm}} \nonumber\\
&+ \bar{\belief}^\typeUninf(\staterv=\statenorm,\playertype^\typeInf = \typeInf\stateinc) \latency{\routeindex}{\statenorm}{\load_\routeindex^\typeUninf+\load_\routeindex^{\typeInf\stateinc}}+ \bar{\belief}^\typeUninf(\staterv=\statenorm,\playertype^\typeInf = \typeInf\statenorm) \latency{\routeindex}{\statenorm}{\load_\routeindex^\typeUninf+\load_\routeindex^{\typeInf\statenorm}}\nonumber\\
=& 
\pInc\prob{\typeInf\stateinc} \latency{\routeindex}{\stateinc}{\load_\routeindex^\typeUninf+\load_\routeindex^{\typeInf\stateinc}} + \pInc\prob{\typeInf\statenorm}\latency{\routeindex}{\stateinc}{\load_\routeindex^\typeUninf+\load_\routeindex^{\typeInf\statenorm}} \nonumber\\
&+ (1-\pInc)\prob{\typeInf\stateinc}\latency{\routeindex}{\statenorm}{\load_\routeindex^\typeUninf+\load_\routeindex^{\typeInf\stateinc}}+ (1-\pInc)\prob{\typeInf\statenorm} \latency{\routeindex}{\statenorm}{\load_\routeindex^\typeUninf+\load_\routeindex^{\typeInf\statenorm}}.\label{eq:exp_latency_typeL}
\end{align}
Similarly, the expected route costs for players in population $\typeInf$ given the belief $\bar{\belief}$ given in \eqref{eq:full_subj_beliefs_H} can be expressed as:
\begin{align}
\mathbb{E}_{\bar{\belief}}[\latency{\routeindex}{\state}{\load_\routeindex^{\typeInf\stateinc}+\load_\routeindex^\typeUninf}|\playertype^\typeInf = \typeInf\stateinc] &= \frac{\pInc\accuracy^\typeInf}{\prob{\typeInf\stateinc}}\latency{\routeindex}{\stateinc}{\load_\routeindex^{\typeInf\stateinc}+\load_\routeindex^\typeUninf} +  \frac{(1-\pInc)(1-\accuracy^\typeInf)}{\prob{\typeInf\stateinc}}\latency{\routeindex}{\statenorm}{\load_\routeindex^{\typeInf\stateinc}+\load_\routeindex^\typeUninf},\label{eq:exp_latency_typeHa}
\end{align}
for type $\typeInf\stateinc$, and:
\begin{align}
\mathbb{E}_{\bar{\belief}}[\latency{\routeindex}{\state}{\load_\routeindex^{\typeInf\statenorm}+\load_\routeindex^\typeUninf}|\playertype^\typeInf = \typeInf\statenorm] &= \frac{\pInc(1-\accuracy^\typeInf)}{\prob{\typeInf\statenorm}}\latency{\routeindex}{\stateinc}{\load_\routeindex^{\typeInf\statenorm}+\load_\routeindex^\typeUninf} + \frac{(1-\pInc)\accuracy^\typeInf}{\prob{\typeInf\statenorm}}\latency{\routeindex}{\statenorm}{\load_\routeindex^{\typeInf\statenorm}+\load_\routeindex^\typeUninf}.\label{eq:exp_latency_typeHn}
\end{align} 
for type $\typeInf\statenorm$.

% Propositions
We now explore how the equilibrium is affected by the non-physical parameters of the game, i.e. those which are not characteristic of the underlying network: the probability of incident ($\pInc$), the fraction of players subscribed to TIS $\typeInf$ ($\fracInf$), and the accuracy of TIS $\typeInf$ ($\accuracy^\typeInf$). 
% Specifically, we focus on $\fracInf$ and $\accuracy^\typeInf$ because as TIS become more widespread and accurate, it is important to understand its effects on equilibrium route choices. We include $\pInc$ to ensure that our results hold under a range of incident scenarios. 
Before moving forward, let us define the four following parameter regimes: 
\begin{align}
\begin{split}
\regime{1} &= \left\{(\pInc,\fracInf,\accuracy^\typeInf) \in (0,1) \times [0,1] \times (0.5,1]\  |\ 0 \leq \fracInf < \fracBound{1} \right\}\\
\regime{2} &= \left\{(\pInc,\fracInf,\accuracy^\typeInf) \in (0,1) \times [0,1] \times (0.5,1]\  |\ \fracBound{1} \leq \fracInf \leq \fracBound{2} \right\}\\
\regime{3} &= \left\{(\pInc,\fracInf,\accuracy^\typeInf) \in (0,1) \times [0,1] \times (0.5,1]\  |\ \fracBound{2} < \fracInf < \fracBound{3} \right\}\\
\regime{4} &= \left\{(\pInc,\fracInf,\accuracy^\typeInf) \in (0,1) \times [0,1] \times (0.5,1]\  |\ \fracBound{3} \leq \fracInf \leq 1 \right\}\\
\end{split}
\end{align}
where:
\begin{align} \label{eq:regimeBounds}
\begin{split}
\fracBound{1} &:= \dfrac{\regConst_1(\widehat{\slope{1}{}}-\bar{\slope{1}{}}\prob{\typeInf\stateinc})}{\totaldemand\prob{\typeInf\statenorm}(\widehat{\slope{1}{}}+\slope{2}{}\prob{\typeInf\stateinc})},\\
\fracBound{2} &:= \dfrac{\regConst_1-\prob{\typeInf\stateinc}\regConst_2}{\totaldemand\prob{\typeInf\statenorm}},\\
\fracBound{3} &:= \dfrac{\regConst_3}{\totaldemand},\\
\bar{\slope{1}{}} &:= (1-\pInc)\slope{1}{\statenorm} + \pInc \slope{1}{\stateinc},\\
\widehat{\slope{1}{}} &:= (1-\pInc)(1-\accuracy^\typeInf)\slope{1}{\statenorm} + \pInc\accuracy^\typeInf \slope{1}{\stateinc},\\
\widetilde{\slope{1}{}} &:= (1-\pInc)\accuracy^\typeInf\slope{1}{\statenorm} + \pInc(1-\accuracy^\typeInf) \slope{1}{\stateinc},\\
\solConst &:= \slope{2}{} \totaldemand - \intercept{1} + \intercept{2},\\
\regConst_1 &:= \frac{\solConst}{\bar{\slope{1}{}}+\slope{2}{}},\\
\regConst_2 &:= \frac{\solConst\prob{\typeInf\stateinc}}{\widehat{\slope{1}{}}+\prob{\typeInf\stateinc}\slope{2}{}},\\
\regConst_3 &:=\dfrac{\solConst \prob{\typeInf\statenorm}}{\widetilde{\slope{1}{}}+\prob{\typeInf\statenorm}\slope{2}{}}.
\end{split}
\end{align}

\begin{figure}[H]
\centering
\subfloat[Population $\typeInf$ receives signal $\statenorm$]{
\includegraphics[width=0.4\linewidth]{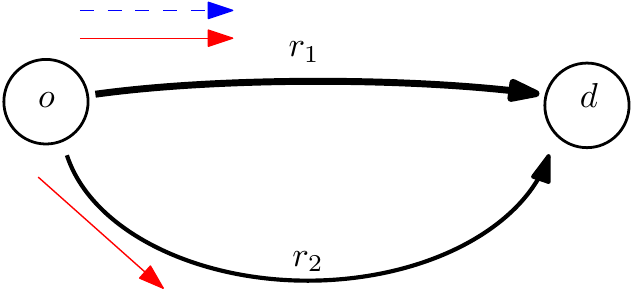}}\\
\subfloat[Population $\typeInf$ receives signal $\stateinc$]{\includegraphics[width=0.4\linewidth]{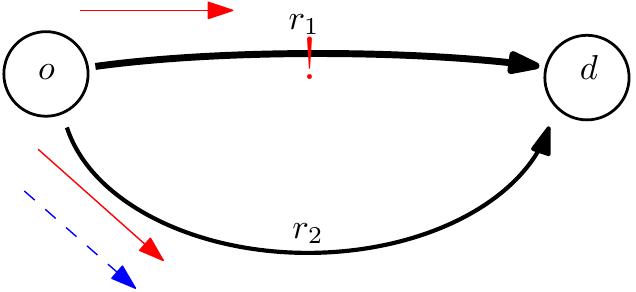}}
\caption{Illustration of the BWE in $\regime{1}$. Dashed blue represents population $\typeInf$, and solid red represents $\typeUninf$.}
\label{fig:reg1_cartoon} 
\end{figure}

The following proposition gives the BWE in Regime $\regime{1}$. This corresponds to the case where the fraction of players that belong to population $\typeInf$ is small, and all players in this population take $\routeindex_1$ when they receive signal $\statenorm$ and route $\routeindex_2$ when they receive signal $\stateinc$. Figure~\ref{fig:reg1_cartoon} illustrates this play.
\begin{proposition}{[Equilibrium for Regime $\regime{1}$]}
\label{prop:reg1equilibrium}

For $(\pInc,\fracInf,\accuracy^\typeInf)\in \regime{1}$, the BWE of $\genGame_p$ is: 
\begin{align}
(\splitfraction{\typeUninf}{1}{\bneIndicator},(\splitfraction{\typeInf}{1}{\statenorm\bneIndicator},\splitfraction{\typeInf}{1}{\stateinc\bneIndicator})) = \left(\frac{\regConst_1}{(1-\fracInf)\totaldemand}-\frac{\prob{\typeInf\statenorm}\fracInf}{1-\fracInf},(1, 0)\right),
\label{eq:bne_regime_1_rhos}
\end{align}
i.e., the fraction $\splitfraction{\typeUninf}{1}{\bneIndicator}$ of players in population $\typeUninf$ choose $\routeindex_1$; if players in population $\typeInf$ receive the signal $\statenorm$ (resp. $\stateinc$), they all take $\routeindex_1$ (resp. $\routeindex_2$).
\end{proposition}

\begin{proof}{Proof of Proposition \ref{prop:reg1equilibrium}.}

For the split fractions \eqref{eq:bne_regime_1_rhos} to be an equilibrium strategy, it must satisfy the Wardrop conditions: type $\typeUninf$'s expected cost must satisfy \eqref{eq:bne_equality}; type $\typeInf\statenorm$'s expected cost must satisfy \eqref{eq:bne_ineq_0}, and type $\typeInf\stateinc$'s expected cost must satisfy \eqref{eq:bne_ineq_1}. Since $\splitfraction{\typeInf}{1}{\statenorm\bneIndicator} = 1$ and $\splitfraction{\typeInf}{1}{\stateinc\bneIndicator} = 0$, we obtain that $\load_{1}^{\typeInf\statenorm\bneIndicator} = \fracInf \totaldemand$, and $\load_{1}^{\typeInf\stateinc\bneIndicator} = 0$. This yields the following conditions:
\begin{align*}
\load_{1}^{\typeUninf\bneIndicator} &= \regConst_1 - \fracInf\totaldemand\prob{\typeInf\statenorm} \\
\regConst_2 < \load_{1}^{\typeUninf\bneIndicator} &< \regConst_3 - \fracInf\totaldemand
\end{align*}
One can check that the aforementioned bounds on $\load_{1}^{\typeUninf\bneIndicator}$ are satisfied when $\fracInf < \fracBound{1}$, i.e. $(\pInc,\fracInf,\accuracy^\typeInf) \in \regime{1}$. \hfill \Halmos
\end{proof}

\begin{figure}[H]
\centering
\subfloat[Population $\typeInf$ receives signal $\statenorm$]{
\includegraphics[width=0.4\linewidth]{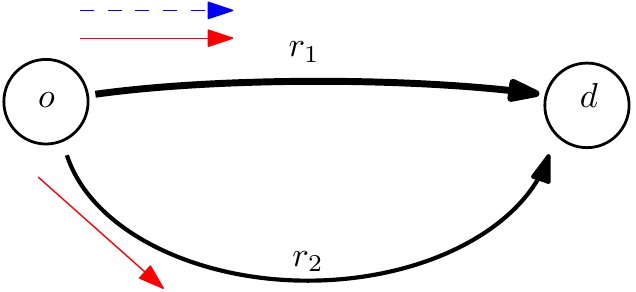}}\\
\subfloat[Population $\typeInf$ receives signal $\stateinc$]{\includegraphics[width=0.4\linewidth]{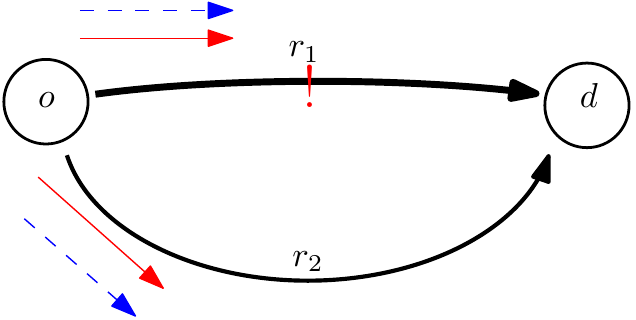}}
\caption{Illustration of the BWE in $\regime{2}$. Dashed blue represents population $\typeInf$, and solid red represents $\typeUninf$.}
\label{fig:reg2_cartoon} 
\end{figure}

Secondly, the following proposition gives the BWE in Regime $\regime{2}$. This corresponds to the case where population $\typeInf$ is sufficiently large that its players can no longer exclusively take $\routeindex_2$ when they receive signal $\stateinc$. However, the size of population $\typeInf$ is still small enough that its players can exclusively take $\routeindex_1$ when they receive signal $\statenorm$. Fig~\ref{fig:reg2_cartoon} illustrates this play.
\begin{proposition}{[Equilibrium for Regime $\regime{2}$]}
\label{prop:reg2equilibrium}

For $(\pInc,\fracInf,\accuracy^\typeInf)\in \regime{2}$, the BWE of $\genGame_p$ is:  
\begin{align}
(\splitfraction{\typeUninf}{1}{\bneIndicator},(\splitfraction{\typeInf}{1}{\statenorm\bneIndicator},\splitfraction{\typeInf}{1}{\stateinc\bneIndicator}))=\left(\frac{\regConst_1 - \fracInf\totaldemand\prob{\typeInf\statenorm} - \prob{\typeInf\stateinc}\regConst_2}{(1-\fracInf)\totaldemand\prob{\typeInf\statenorm}},\left(1, \frac{\fracInf\totaldemand\prob{\typeInf\statenorm} + \prob{\typeInf\stateinc}\regConst_2 - \regConst_1}{\fracInf\totaldemand\prob{\typeInf\statenorm}}\right)\right),
\label{eq:bne_regime_2_rhos}
\end{align}
i.e. the fraction $\splitfraction{\typeUninf}{1}{\bneIndicator}$ of players in population L choose $\routeindex_1$; if players of population $\typeInf$ receive the signal $\statenorm$, they all take $\routeindex_1$, and if they receive the signal $\stateinc$, the fraction $\splitfraction{\typeInf}{1}{\stateinc\bneIndicator}$ of population $\typeInf$ chooses the first route.
\end{proposition}

\begin{proof}{Proof of Proposition \ref{prop:reg2equilibrium}.}

For the split fractions \eqref{eq:bne_regime_2_rhos} to be an equilibrium strategy, it must satisfy the Wardrop conditions: type $\typeUninf$'s expected cost must satisfy \eqref{eq:bne_equality}; type $\typeInf\statenorm$'s expected cost must satisfy \eqref{eq:bne_ineq_0}, and the expected costs for type $\typeInf\stateinc$ must satisfy \eqref{eq:bne_equality}. Since $ \splitfraction{\typeInf}{1}{\statenorm\bneIndicator} = 1$, we obtain $\load_{1}^{\typeInf\statenorm\bneIndicator} = \fracInf \totaldemand$. Combined with the Wardrop conditions, this yields:
\begin{align}
\load_{1}^{\typeUninf\bneIndicator} + \load_{1}^{\typeInf\stateinc\bneIndicator}\prob{\typeInf\stateinc} &= \regConst_1 - \fracInf\totaldemand\prob{\typeInf\statenorm},\nonumber \\
\load_{1}^{\typeUninf\bneIndicator} + \load_{1}^{\typeInf\stateinc\bneIndicator} &= \regConst_2,\nonumber \\
\load_{1}^{\typeUninf\bneIndicator} &\leq \regConst_3 - \fracInf\totaldemand. \label{eq:prop_2_q1_bound}
\end{align}
Solving for $\load_{1}^{\typeUninf\bneIndicator}$ and $\load_{1}^{\typeInf\stateinc\bneIndicator}$ gives the following expressions: 
\begin{align*}
\load_{1}^{\typeUninf\bneIndicator} &= \frac{\regConst_1 - \fracInf\totaldemand\prob{\typeInf\statenorm} - \prob{\typeInf\stateinc}\regConst_2}{\prob{\typeInf\statenorm}},\\
\load_{1}^{\typeInf\stateinc\bneIndicator} &=\frac{\fracInf\totaldemand\prob{\typeInf\statenorm} + \prob{\typeInf\stateinc}\regConst_2 - \regConst_1}{\prob{\typeInf\statenorm}}
\end{align*}
In order for $\load_{1}^{\typeUninf\bneIndicator}$ and $\load_{1}^{\typeInf\stateinc\bneIndicator}$ to be non-negative and satisfy \eqref{eq:prop_2_q1_bound}, we find that $\fracBound{1} \leq \fracInf \leq \fracBound{2}$, i.e. $(\pInc,\fracInf,\accuracy^\typeInf) \in \regime{2}$. \hfill \Halmos
\end{proof}

\begin{figure}[H]
\centering
\subfloat[Population $\typeInf$ receives signal $\statenorm$]{
\includegraphics[width=0.4\linewidth]{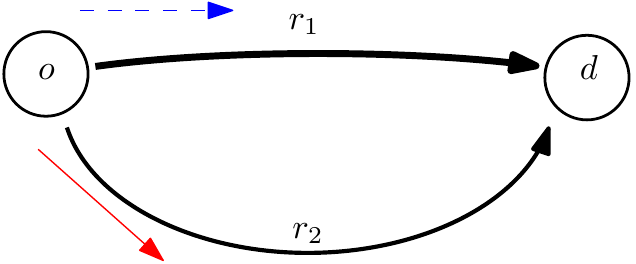}}\\
\subfloat[Population $\typeInf$ receives signal $\stateinc$]{\includegraphics[width=0.4\linewidth]{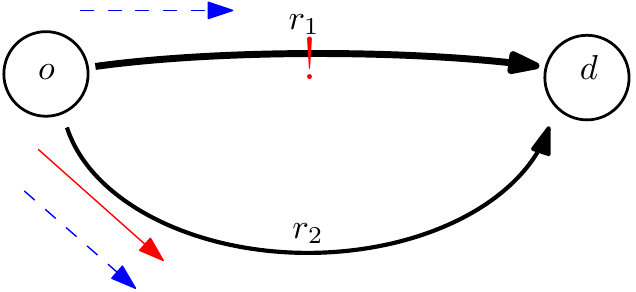}}
\caption{Illustration of the BWE in $\regime{3}$. Dashed blue represents population $\typeInf$, and solid red represents $\typeUninf$.}
\label{fig:reg3_cartoon} 
\end{figure}

The following proposition gives the BWE in Regime $\regime{3}$. This corresponds to the case where the players in population $\typeUninf$ exclusively choose $\routeindex_2$ in order to ``avoid'' the significant expected load on $\routeindex_1$ due to population $\typeInf$ players. Figure~\ref{fig:reg3_cartoon} illustrates this play.
\begin{proposition}{[Equilibrium for Regime $\regime{3}$]}
\label{prop:reg3equilibrium}

For $(\pInc,\fracInf,\accuracy^\typeInf)\in \regime{3}$, the BWE of $\genGame_p$ is:   
\begin{align}
(\splitfraction{\typeUninf}{1}{\bneIndicator},(\splitfraction{\typeInf}{1}{\statenorm\bneIndicator},\splitfraction{\typeInf}{1}{\stateinc\bneIndicator}))=\left(0,\left(1, \frac{\regConst_2}{\fracInf\totaldemand}\right) \right),
\label{eq:bne_regime_3_rhos}
\end{align}
i.e. all players in population $\typeUninf$ take $\routeindex_2$; if players in population $\typeInf$ receive the signal $\statenorm,$ they all take $\routeindex_1$, and if they receive the signal $\stateinc$, the fraction $\splitfraction{\typeInf}{1}{\stateinc\bneIndicator}$ of population $\typeInf$ take $\routeindex_1$.
\end{proposition}

\begin{proof}{Proof of Proposition \ref{prop:reg3equilibrium}.}

For the split fractions \eqref{eq:bne_regime_3_rhos} to be an equilibrium strategy, it must satisfy the Wardrop conditions: population $\typeUninf$'s expected cost must satisfy \eqref{eq:bne_ineq_0}; type $\typeInf\statenorm$'s expected cost must satisfy \eqref{eq:bne_ineq_1} and the expected cost for $\typeInf\stateinc$ must satisfy \eqref{eq:bne_equality}. Since $\splitfraction{\typeUninf}{1}{\bneIndicator}=0$ and $\splitfraction{\typeInf}{1}{\statenorm\bneIndicator}=1$, we obtain $\load_{1}^{\typeUninf\bneIndicator} = 0$ and $\load_{1}^{\typeInf\statenorm\bneIndicator} = \fracInf \totaldemand$. Substituting into the Wardrop conditions yields:
\begin{align*}
\load_{1}^{\typeInf\stateinc\bneIndicator} &= \regConst_2, \\
\frac{\regConst_1 - \prob{\typeInf\stateinc}\regConst_2}{\prob{\typeInf\statenorm}\totaldemand} < \fracInf &< \frac{\regConst_3}{\totaldemand}.
\end{align*}
The aforementioned bounds on $\fracInf$ correspond to $\fracBound{2} < \fracInf < \fracBound{3}$, i.e. $(\pInc,\fracInf,\accuracy^\typeInf) \in \regime{3}$. \hfill \Halmos
\end{proof}

The following proposition gives the BWE in Regime $\regime{4}$. This corresponds to the case where population $\typeInf$ is sufficiently large that its players can no longer exclusively take $\routeindex_1$ when they receive signal $\statenorm$ and thus a fraction of the players take $\routeindex_2$ when they receive $\statenorm$. Figure~\ref{fig:reg1_cartoon} illustrates this play.

\begin{proposition}{[Equilibrium for Regime $\regime{4}$]}
\label{prop:reg4equilibrium}

For $(\pInc,\fracInf,\accuracy^\typeInf)\in \regime{3}$, the BWE of $\genGame_p$ is:  
\begin{align}
(\splitfraction{\typeUninf}{1}{\bneIndicator},(\splitfraction{\typeInf}{1}{\statenorm\bneIndicator},\splitfraction{\typeInf}{1}{\stateinc\bneIndicator})) = \left(0,\left(\dfrac{\regConst_3}{\fracInf\totaldemand}, \dfrac{\regConst_2}{\fracInf\totaldemand}\right)\right),
\label{eq:bne_regime_4_rhos}
\end{align}
i.e. all players in population $\typeUninf$ take $\routeindex_2$; if players in population $\typeInf$ receive the signal $\statenorm$ (resp. $\stateinc$), the fraction $\splitfraction{\typeInf}{1}{\statenorm\bneIndicator}$ (resp. $\splitfraction{\typeInf}{1}{\stateinc\bneIndicator}$) take $\routeindex_1$.
\end{proposition}

\begin{figure}[H]
\centering
\subfloat[Population $\typeInf$ receives signal $\statenorm$]{
\includegraphics[width=0.4\linewidth]{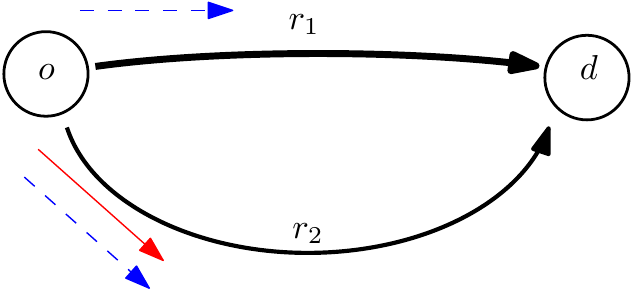}}\\
\subfloat[Population $\typeInf$ receives signal $\stateinc$]{\includegraphics[width=0.4\linewidth]{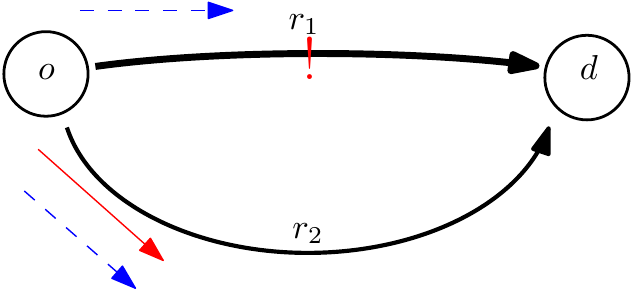}}
\caption{Illustration of the BWE in $\regime{4}$. Dashed blue represents population $\typeInf$, and solid red represents $\typeUninf$.}
\label{fig:reg4_cartoon} 
\end{figure}

\begin{proof}{Proof of Proposition \ref{prop:reg4equilibrium}.}

For the split fractions \eqref{eq:bne_regime_3_rhos} to be an equilibrium strategy, it must satisfy the Wardrop conditions: the expected cost of population $\typeUninf$ players must satisfy \eqref{eq:bne_ineq_1}, and both $\typeInf\statenorm$ and $\typeInf\stateinc$'s expected costs must satisfy \eqref{eq:bne_equality}. Since $\splitfraction{\typeUninf}{1}{\bneIndicator}=0$, we obtain $\load_{1}^{\typeUninf\bneIndicator} = 0$. Combining this with the Wardrop conditions yields:
\begin{align}
\load_{1}^{\typeInf\statenorm\bneIndicator} &= \regConst_3\nonumber \\
\load_{1}^{\typeInf\stateinc\bneIndicator} &= \regConst_2\nonumber \\
\prob{\typeInf\statenorm}\load_{1}^{\typeInf\statenorm\bneIndicator} + \prob{\typeInf\stateinc}\load_{1}^{\typeInf\stateinc\bneIndicator} &> \regConst_1 \label{eq:prop4_qH_cond}
\end{align}
The condition \eqref{eq:prop4_qH_cond} is satisfied if $0<\pInc<1$ and $0.5<\eta^H\leq1$, which is true by assumption. The equilibrium demands must be no greater than the total demand of that population, i.e. $\load_{1}^{\typeInf\statenorm\bneIndicator}, \load_{1}^{\typeInf\stateinc\bneIndicator} \leq \playerFrac{\typeInf} \totaldemand$. This is satisfied if $\fracInf \geq \fracBound{3}$, i.e. $(\pInc,\fracInf,\accuracy^\typeInf) \in \regime{4}$. \hfill \Halmos
\end{proof}

We find that the rest of the 27 possible strategy profiles are not equilibrium profiles. Some of these profiles are infeasible, such as profiles that result in negative loads on routes. Other strategies are strictly dominated. Let us argue that strategy profiles where $\splitfraction{\typeInf}{1}{\stateinc\bneIndicator}\geq\splitfraction{\typeInf}{1}{\statenorm\bneIndicator}$ are strictly dominated by those where $\splitfraction{\typeInf}{1}{\stateinc\bneIndicator}<\splitfraction{\typeInf}{1}{\statenorm\bneIndicator}$. Since population $\typeUninf$ always plays exactly the same strategy, population $\typeInf$ expects the same demand on $\routeindex_1$ due to the population $\typeUninf$ players. Additionally, population $\typeInf$ expects the slope of $\routeindex_1$ to be greater when they receive signal $\stateinc$, since they now know it's more likely that $\routeindex_1$ is in the incident state. Thus, their expected cost of routing the same demand on $\routeindex_1$ will be higher when they receive signal $\stateinc$ than when they receive signal $\statenorm$. Therefore strategies where $\splitfraction{\typeInf}{1}{\stateinc\bneIndicator}\geq\splitfraction{\typeInf}{1}{\statenorm\bneIndicator}$ are dominated by those where $\splitfraction{\typeInf}{1}{\stateinc\bneIndicator}<\splitfraction{\typeInf}{1}{\statenorm\bneIndicator}$.

We omit an exhaustive analysis of all 27 profiles for the sake of brevity. Table~\ref{tab:profile_summary} summarizes the strategy profiles and whether or not they can exist in equilibrium. Interestingly enough, we find that for all parameter values, at least one population type exclusively takes one route in equilibrium. That is, under the information structure assumed in $\bar{\belief}$, the game $\genGame_p$ does not admit an equilibrium where all types split their demand: there are no parameter values where all of $\splitfraction{\typeUninf}{1}{\bneIndicator}, \splitfraction{\typeInf}{1}{\statenorm\bneIndicator},\splitfraction{\typeInf}{1}{\stateinc\bneIndicator} \in (0,1)$. 

\begin{table}[h]
\centering
\subfloat[$\typeUninf$ play $\routeindex_1$ exclusively ($\splitfraction{\typeUninf}{1}{\bneIndicator} = 1$)]{
\begin{tabular}{l|c|c|c|}
    \diaghead(2,-1){$\strat_\typeInf(\stateinc)--$}%
    {$\splitfraction{\typeInf}{1}{\statenorm\bneIndicator}$}{$\splitfraction{\typeInf}{1}{\stateinc\bneIndicator}$}&  $=0$  & $\in (0,1) $& $= 1 $\\    \hline
    $=0$& \xmark& \xmark & \xmark\\ \hline
    $\in (0,1)$& \xmark &\xmark& \xmark\\ \hline
    $=1$&\xmark&\xmark&\xmark \\ \hline
  \end{tabular}}
\subfloat[$\typeUninf$ split between $\routeindex_1, \routeindex_2$ ($\splitfraction{\typeUninf}{1}{\bneIndicator} \in (0,1)$)]{
  \begin{tabular}{l|c|c|c|}
    \diaghead(2,-1){$\strat_\typeInf(\stateinc)--$}%
    {$\splitfraction{\typeInf}{1}{\statenorm\bneIndicator}$}{$\splitfraction{\typeInf}{1}{\stateinc\bneIndicator}$}&  $=0$  & $\in (0,1) $& $= 1 $\\    \hline
    $=0$& \xmark& \xmark & \xmark\\ \hline
    $\in (0,1)$& \xmark &\xmark& \xmark\\ \hline
    $=1$& $\regime{1}$ & $\regime{2}$ &\xmark \\ \hline
  \end{tabular}}
\subfloat[$\typeUninf$ play $\routeindex_2$ exclusively ($\splitfraction{\typeUninf}{1}{\bneIndicator} = 0$)]{
  \begin{tabular}{l|c|c|c|}
    \diaghead(2,-1){$\strat_\typeInf(\stateinc)--$}%
    {$\splitfraction{\typeInf}{1}{\statenorm\bneIndicator}$}{$\splitfraction{\typeInf}{1}{\stateinc\bneIndicator}$}&  $=0$  & $\in (0,1) $& $= 1 $\\    \hline
    $=0$& \xmark& \xmark & \xmark\\ \hline
    $\in (0,1)$& \xmark &$\regime{4}$& \xmark\\ \hline
    $=1$&\xmark&$\regime{3}$&\xmark \\ \hline
  \end{tabular}}
\caption{Strategy profiles for the Bayesian Congestion Game. Equilibrium strategy profiles are marked with their corresponding parameter regimes. Non-equilibrium profiles are marked with a cross.}
\label{tab:profile_summary}
\end{table}

% \begin{theorem}{[Existence and Uniqueness of equilibrium]}
% \label{thm:exist_unique_bne}

% There exists a unique equilibrium of $\genGame$ for every value of $\fracInf$ and $\pInc$.
% \end{theorem}
% \begin{proof}{Proof of Theorem \ref{thm:exist_unique_bne}}

% Propositions \ref{prop:reg1equilibrium} -- \ref{prop:reg4equilibrium} show that there exist four regimes where equilibrium exist. It can be shown that the union of these spans the entire $\fracInf,\pInc$ parameter space. Figure \ref{fig:bne_regions} illustrates these regimes. Therefore, an equilibrium exists for all values $\fracInf,\pInc$.

% The four regimes $\regime{1}$--$\regime{4}$ do not overlap, and furthermore, there are no other strategy profiles that are equilibrium (Proposition \ref{prop:dominated_strats}, Table \ref{tab:profile_summary}). Therefore the equilibrium is unique. \hfill \Halmos
% \end{proof}
\begin{table}
\centering
\begin{tabular}{r|l|l|l}
\hline
Quantity & symbol & value & units\\ \hline \hline
$\routeindex_1$ slope, state $\statenorm$ & $\slope{1}{\statenorm}$ & 1& min/(veh hr$^{-1}$)\\
$\routeindex_1$ slope, state $\stateinc$ & $\slope{1}{\stateinc}$ & 3 &min/(veh hr$^{-1}$)\\
$\routeindex_2$ slope & $\slope{2}{}$ & 2&min/(veh hr$^{-1}$)\\
$\routeindex_1$ intercept & $\intercept{1}$ & 19 & min\\
$\routeindex_2$ intercept & $\intercept{2}$ & 21 & min\\
Total demand & $\totaldemand$ & 5 & $10^3$ veh hr$^{-1}$\\
\hline
% Incident probability & $\pInc$ & 0.2\\
\end{tabular}
\caption{Parameter values for two-route network example.}
\label{tab:param_vals}
\end{table}

\begin{figure}[htb]
\centering
\subfloat[Equilibrium regimes in the $\fracInf-\pInc$ plane]{
\includegraphics[trim={2.5cm 0cm 2.5cm 1cm},clip,width=0.45\linewidth]{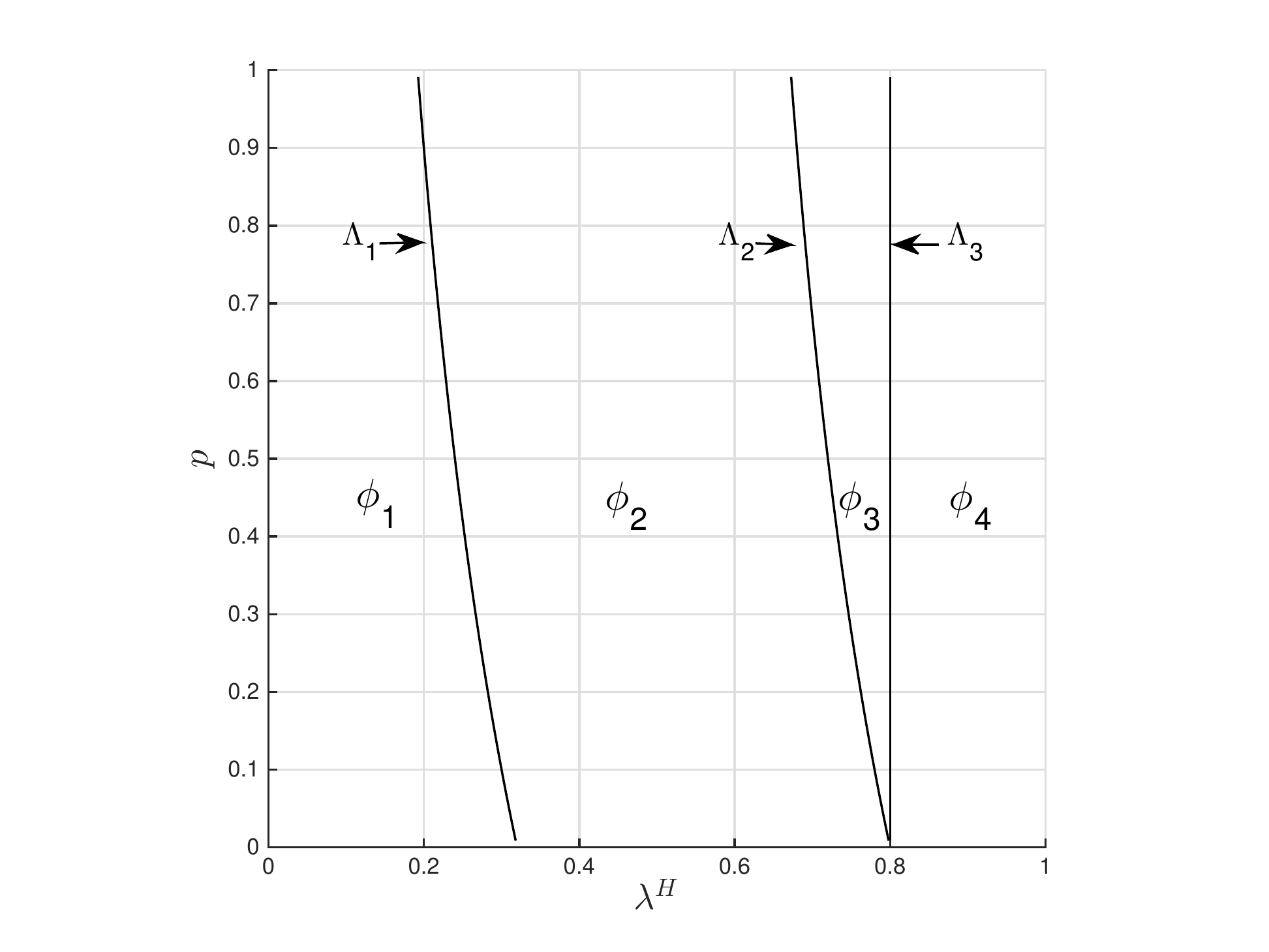}\label{fig:regBound_etaH1}}
\subfloat[Equilibrium split fractions vs. $\fracInf$, $\pInc=0.2$]{
\includegraphics[trim={2.5cm 0cm 2.5cm 1cm},clip,width=0.45\linewidth]{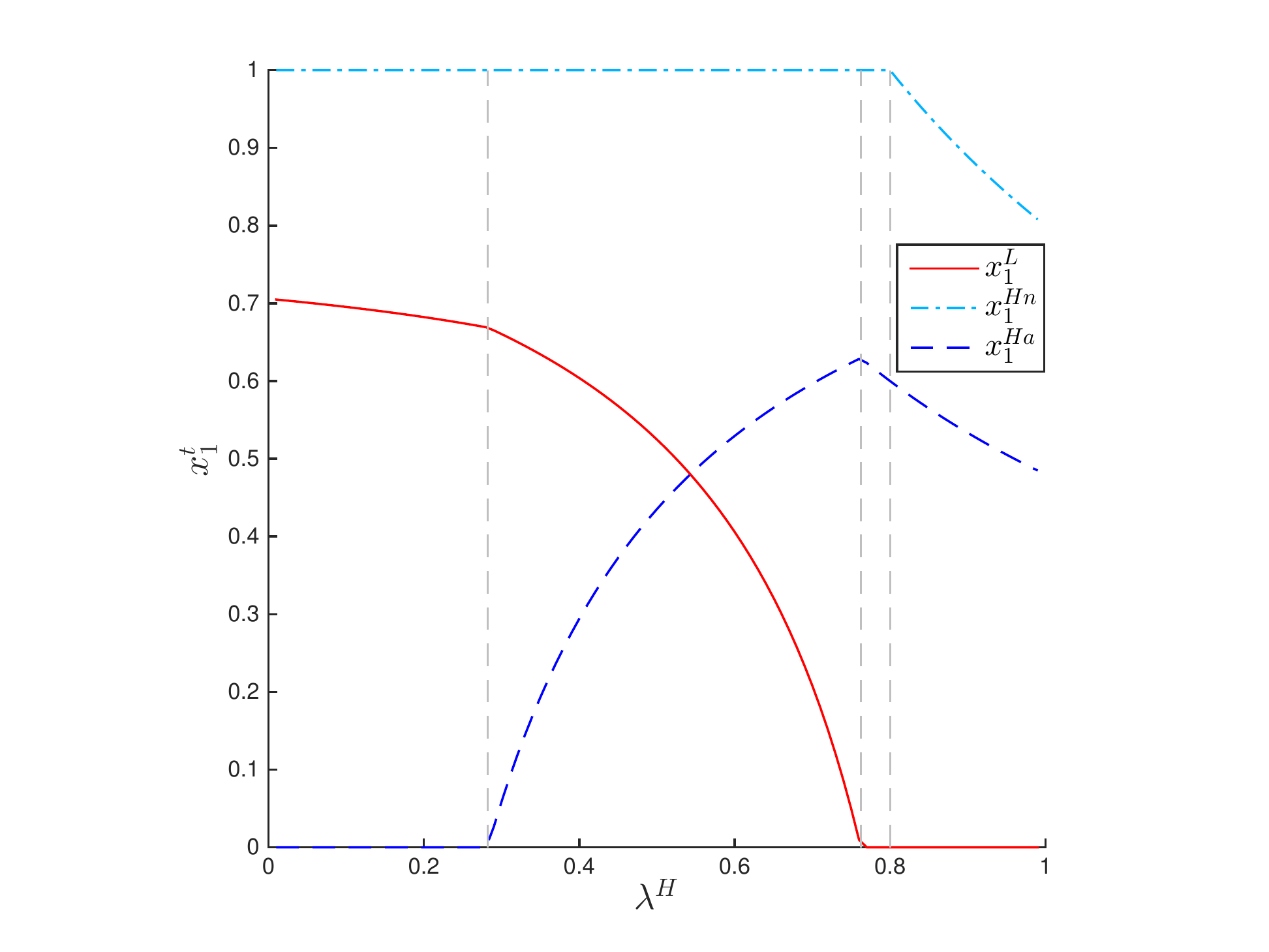}\label{fig:splitfrac_etaH1}}
\caption{Equilibrium regimes and split fractions for $\accuracy^\typeInf=1$.}
\label{fig:equilib_etaH_1} 
\end{figure}

\subsection{Remarks on Equilibrium Structure} % (fold)
\label{sub:equilibrium_interpretations}

We illustrate in Figure~\ref{fig:regBound_etaH1} the four equilibrium regimes under the assumption that population $\typeInf$ has perfectly accurate reports of the state, i.e. $\accuracy^\typeInf = 1$. Unless otherwise specified, the parameters for the two route network in our numerical examples take the values given in Table \ref{tab:param_vals}. Notably, the width of $\regime{1}$ decreases and the width of $\regime{3}$ increases as $\pInc$ increases. Additionally, the right boundary of $\regime{3}$, $\fracBound{3}$, is a constant with respect to incident probability, $\pInc$.

We illlustrate the equilibrium split fraction in all four regimes for each population and state signals in Figure~\ref{fig:splitfrac_etaH1}. In regime $\regime{1}$, the players of population $\typeUninf$ take both routes; in this numerical example, approximately 70\% of the population $\typeUninf$ players choose the first route, and the rest take the second route. Increasing the fraction of population $\typeInf$ players in this regime has a small effect of linearly decreasing the fraction of population $\typeUninf$ that chooses $\routeindex_1$. In regime $\regime{2}$, increasing the fraction of population $\typeInf$ has a more dramatic effect, and the population $\typeUninf$ players shift toward taking the second route, until in regimes $\regime{3}$ and $\regime{4}$, the population $\typeUninf$ players take route 2 exclusively.

For population $\typeInf$ players, their play depends on the signal received. When they receive the signal $\statenorm$, population $\typeInf$ players take route 1 exclusively for the first three regimes. Unlike population $\typeUninf$ players, they know that the first route will be uncongested and thus it will be less costly than taking the second route. It is only when the fraction of population $\typeInf$ players gets sufficiently high in regime $\regime{4}$, that population $\typeInf$ players have to start taking second route.

On the other hand, when population $\typeInf$ players receive the signal $\stateinc$, the equilibrium strategy in regime $\regime{1}$ is to exclusively take $\routeindex_2$. In this regime, the size of the population $\typeInf$ population is small enough that they can all take the second route and avoid the congestion due to the incident. However, in regime $\regime{2}$, population $\typeInf$ is large enough that it's no longer possible for its players to exclusively take $\routeindex_2$. This reflects the phenomenon of \emph{concentration} in the sense of \cite{ben1991dynamic}, where commuters that receive the same information all take the same route, increasing the cost. Thus, in regime $\regime{2}$, we observe that the split fraction increases as some population $\typeInf$ players take route $\routeindex_1$ even when they receive the signal $\stateinc$. In regimes $\regime{3}$--$\regime{4}$, the population $\typeInf$ players react to the population $\typeUninf$ players exclusively taking route 2; in these regimes, the split fraction decreases as $\fracInf$ increases.

We note that we recover the classical game formulations at the boundaries of our game. First, in the limiting case where $\fracInf=0$, we recover the classical imperfect information formulation of the game, where all players belong to population $\typeUninf$. The equilibrium split fraction for population $\typeUninf$ approaches $\regConst_1/\totaldemand$ as $\fracInf\rightarrow 0$, which is the Wardrop equilibrium strategy for the classical imperfect information case. Second, we recover the classical perfect information formulation of the game in the limiting case of $\fracInf = 1$ and $\accuracy^\typeInf = 1$, i.e., everyone belongs to population $\typeInf$ and receives perfectly accurate information. Under these conditions, the equilibrium split fractions for each of the player types $\splitfraction{\typeInf}{1}{\statenorm\bneIndicator},\splitfraction{\typeInf}{1}{\stateinc\bneIndicator}$ correspond to the Wardrop equilibrium of subgames for $\state=\statenorm$ and $\state=\stateinc$ respectively: taking the limits $\fracInf\rightarrow 1$ and $\accuracy^\typeInf \rightarrow 1$ yield $\splitfraction{\typeInf}{1}{\statenorm\bneIndicator} = \dfrac{\solConst}{\totaldemand(\slope{1}{\statenorm}+\slope{2}{})},$ and $\splitfraction{\typeInf}{1}{\stateinc\bneIndicator}= \dfrac{\solConst}{\totaldemand(\slope{1}{\stateinc}+\slope{2}{})}$.

% subsubsection  (end)

% \begin{figure}[htb]
% \centering
% \includegraphics[trim={1cm 6cm 1cm 6cm},clip,width=0.8\linewidth]{qs_n.pdf}
% \includegraphics[trim={1cm 6cm 1cm 6cm},clip,width=0.8\linewidth]{qs_a.pdf}
% \caption{Equilibrium demands in state $\statenorm$ (top) and state $\stateinc$ (bottom) respectively}
% \label{fig:eq_demands} 
% \end{figure}

% subsection equilibrium_interpretations (end)

%% file: welfare.tex
%!TEX root = main.tex

\section{Value of Information}
\label{sec:cost_analysis}

In this section, we evaluate the equilibrium strategies presented in the previous section. First, we compute the equilibrium costs for each population in Section \ref{sub:equilibrium_costs}. Next, we analyze the individual value of information in Section \ref{sub:individual_value_of_information}. Finally, we examine the social value of information in Section \ref{sub:social_value_of_information}.

For simplicity and ease of presentation, we restrict our attention to information environments where $\accuracy^\typeInf = 1$ and $\accuracy^\typeUninf=0.5$, i.e. population $\typeInf$ players receive the exact state of nature, while population $\typeUninf$ players only know common knowledge. However, our analysis below easily generalizes to cases where $\accuracy^\typeUninf=0.5<\accuracy^\typeInf<1$. Recall that when $\accuracy^\typeInf=1$ and the fraction of players belonging to population $\typeInf$ equals unity ($\fracInf = 1)$, we recover the classical complete and perfect information game. Similarly, when the fraction of players belonging to population $\typeInf$ equals zero $(\fracInf = 0)$ we recover the classical imperfect information game. For notational ease, we denote the regime boundaries given in \eqref{eq:regimeBounds}, $\fracBoundSym_j(\accuracy^\typeInf=1,\pInc)$, as $\fracBoundEta{j}$, for $j \in \{1,2,3\}$.

For $\accuracy^\typeInf=1,$ the marginal type distributions \eqref{eq:prob_of_types} simplify to $\prob{\typeInf\stateinc} = \pInc $ and $
\prob{\typeInf\statenorm} = 1-\pInc$. The beliefs $\bar{\belief}$ from \eqref{eq:full_subj_beliefs_L} and \eqref{eq:full_subj_beliefs_H} simplify to:
\begin{align}
\begin{split}
\bar{\belief}^\typeUninf(\state=\stateinc,\playertype^\typeInf=\typeInf\stateinc|\playertype^\typeUninf=\typeUninf) &= \pInc^2,\\
\bar{\belief}^\typeUninf(\state=\stateinc,\playertype^\typeInf=\typeInf\statenorm|\playertype^\typeUninf=\typeUninf) &= \pInc(1-\pInc),\\
\bar{\belief}^\typeUninf(\state=\statenorm,\playertype^\typeInf=\typeInf\stateinc|\playertype^\typeUninf=\typeUninf) &= \pInc(1-\pInc),\\
\bar{\belief}^\typeUninf(\state=\statenorm,\playertype^\typeInf=\typeInf\statenorm|\playertype^\typeUninf=\typeUninf) &= (1-\pInc)^2,
\end{split}\label{eq:full_subj_beliefs_L_simplified}
\end{align}
for population $\typeUninf$, and
\begin{align}
\begin{split}
\bar{\belief}^\typeInf(\state = \stateinc,\playertype^\typeUninf = \typeUninf | \playertype^\typeInf = \typeInf\stateinc) &= 1,\\
\bar{\belief}^\typeInf(\state = \statenorm,\playertype^\typeUninf = \typeUninf| \playertype^\typeInf = \typeInf\stateinc) &= 0,\\
\bar{\belief}^\typeInf(\state = \stateinc,\playertype^\typeUninf = \typeUninf | \playertype^\typeInf = \typeInf\statenorm) &= 0,\\
\bar{\belief}^\typeInf(\state = \statenorm,\playertype^\typeUninf = \typeUninf | \playertype^\typeInf = \typeInf\statenorm) &= 1,
\end{split}\label{eq:full_subj_beliefs_H_simplified}
\end{align}
for population $\typeInf$. Additionally, the quantities defined in \eqref{eq:regimeBounds} can be particularized as follows:
\begin{align*}
\tilde{\slope{1}{}} &= (1-\pInc)\slope{1}{\statenorm}\\
\widehat{\slope{1}{}} &= \pInc\slope{1}{\stateinc}\\
\regConst_1 &= \frac{\solConst}{\bar{\slope{1}{}}+\slope{2}{}}\\
\regConst_2 &= \frac{\solConst}{\slope{1}{\stateinc}+\slope{2}{}}\\
\regConst_3 &= \frac{\solConst}{\slope{1}{\statenorm}+\slope{2}{}}.
\end{align*}

% Plugging in the values from the information services' accuracies and the common prior to population $\typeUninf$'s belief give us:
% \begin{align*}
% \belief_s^{\typeUninf}(\playertype^\typeInf = \typeInf\stateinc|\playertype^\typeUninf = \typeUninf) &= \pInc,\quad
% \belief_s^{\typeUninf}( \playertype^\typeInf = \typeInf\statenorm|\playertype^\typeUninf = \typeUninf ) = 1 - \pInc
% \end{align*}

\subsection{Equilibrium costs} % (fold)
\label{sub:equilibrium_costs}
Utilizing the equilibrium strategy distributions from Propositions \ref{prop:reg1equilibrium}--\ref{prop:reg4equilibrium}, we obtain the following expressions for the equilibrium cost, $\costVar_{\staterv}^{\service\bneIndicator}$, for a player of population $\service$ in each state, as given by \eqref{eq:state_type_cost_form}:
\begin{align}
\costVar^{\typeUninf\bneIndicator}_\statenorm &=
\begin{cases}
  \left(\dfrac{\regConst_1}{(1-\fracInf)\totaldemand}-\dfrac{(1-\pInc)\fracInf}{1-\fracInf}\right)\left(\slope{1}{\statenorm}\left(\regConst_1 + \pInc \fracInf\totaldemand\right)+\intercept{1}\right)\\
  +\left(1-\left(\dfrac{\regConst_1}{(1-\fracInf)\totaldemand}-\dfrac{(1-\pInc)\fracInf}{1-\fracInf}\right)\right)\left(\slope{2}{} \left(\totaldemand - (\regConst_\typeUninf^\rom{1} + \pInc \fracInf\totaldemand) \right) +\intercept{2}\right), &\text{ in } \regime{1}\\
  \left(\dfrac{\regConst_4}{(1-\fracInf)\totaldemand} - \dfrac{\fracInf}{1-\fracInf} \right)\left(\slope{1}{\statenorm}\left(\regConst_4\right) + \intercept{1} \right)\\
  + \left(1 - \left(\dfrac{\regConst_4}{(1-\fracInf)\totaldemand}-\dfrac{\fracInf}{1-\fracInf}\right) \right)\left(\slope{2}{}\left(\totaldemand - \regConst_4\right) + \intercept{2} \right) , &\text{ in } \regime{2}\\
  \slope{2}{}(1- \fracInf)\totaldemand+\intercept{2}, &\text{ in } \regime{3}\\
  \slope{2}{}\left(\totaldemand-\regConst_3\right)+\intercept{2}, &\text{ in } \regime{4}
\end{cases}\label{eq:cost_L_n}\\
\costVar^{\typeUninf\bneIndicator}_\stateinc &=
\begin{cases}
  \left(\dfrac{\regConst_1}{(1-\fracInf)\totaldemand}-\dfrac{(1-\pInc)\fracInf}{1-\fracInf}\right)\left(\slope{1}{\stateinc}\left(\regConst_1-(1-\pInc)\fracInf\totaldemand \right)+\intercept{1}\right)\\ 
  + \left(1- \left(\dfrac{\regConst_1}{(1-\fracInf)\totaldemand}-\dfrac{(1-\pInc)\fracInf}{1-\fracInf}\right)\right)\left(\slope{2}{} \left(\totaldemand - \left(\regConst_1-(1-\pInc)\fracInf\totaldemand \right)\right)+\intercept{2}\right), &\text{ in } \regime{1}\\
  \slope{1}{\stateinc}\left(\regConst_2\right) + \intercept{1}
  , &\text{ in } \regime{2},\regime{3},\regime{4}
\end{cases}\label{eq:cost_L_a}\\
\costVar^{\typeInf\bneIndicator}_\statenorm &= 
\begin{cases}
  \slope{1}{\statenorm}\left(\regConst_1 + p\fracInf\totaldemand\right) + \intercept{1}, &\text{ in }\regime{1}\\
  \slope{1}{\statenorm}\left(\regConst_4\right) + \intercept{1}, &\text{ in }\regime{2}\\
  \slope{1}{\statenorm}\left(\fracInf\totaldemand\right)+\intercept{1}, &\text{ in }\regime{3}\\
  \slope{1}{\statenorm}\left(\regConst_3 \right) + \intercept{1} , &\text{ in }\regime{4}\\
\end{cases}\label{eq:cost_Hn}\\
\costVar^{\typeInf\bneIndicator}_\stateinc &= \begin{cases}
  \slope{2}{}\left(\totaldemand - \left(\regConst_1 - (1-\pInc)\fracInf\totaldemand\right)\right)+\intercept{2}, &\text{ in } \regime{1}\\
  \slope{1}{\stateinc}\left(\regConst_2\right) + \intercept{1}, &\text{ in } \regime{2},\regime{3},\regime{4}
\end{cases}\label{eq:cost_Ha}
\end{align}
where $\regConst_4:=\regConst_2\left(1+\dfrac{\slope{1}{\stateinc}-\slope{1}{\statenorm}}{\bar{\slope{1}{}}+\slope{2}{}}\right)$.

The expressions for the expected population-dependent cost ($\costVar^{\service\bneIndicator}$), the state-dependent social cost ($\costVar^\bneIndicator_\staterv$), and the expected social cost ($\bar{\costVar}^\bneIndicator$) can be computed by plugging the expressions \eqref{eq:cost_L_n}--\eqref{eq:cost_Ha} into \eqref{eq:exp_pop_cost}, \eqref{eq:state_soc_cost}, and \eqref{eq:exp_soc_cost}, respectively.

% Recall that the social optimum cost is the minimum cost for a given set of parameters. The price of anarchy (PoA) for the Bayesian congestion game is the aggregate equilibrium cost normalized by the social optimum cost. We calculated the social optimum for our choice of parameters in each state. In our plots, we normalize all costs by the respective socially optimal cost. This allows us to compare the individual and social effect of information on welfare against a ``best-case'' benchmark.

% subsubsection type_and_state_dependent_cost (end)

\subsection{Individual value of information} % (fold)
\label{sub:individual_value_of_information}
We now analyze the relative individual value of the information environment $(\pInc,\fracInf,\accuracy^\typeInf=1,\accuracy^\typeUninf=0.5)$ and its relationship with $\fracInf$. 

First, we present two useful lemmas. The first lemma provides that if the players in population $\typeInf$ exclusively choose one route in a given state, then the realized cost for a population $\typeInf$ player in that state is less than that of a population $\typeUninf$ player:

\begin{lemma}
\label{lem:ineq_equilib_cost}
Consider the information environment $(\pInc,\fracInf,\accuracy^\typeInf=1,\accuracy^\typeUninf=0.5)$. In a given state $\staterv\in\{\stateinc,\statenorm\}$, if $\splitfraction{\playertype^\typeInf}{1}{\bneIndicator} = 0$ or $1$, then $\costVar_{\staterv}^{\typeUninf\bneIndicator} > \costVar_{\staterv}^{\typeInf\bneIndicator}$, and the relative individual value of information in that state, $\valueVar_\staterv$, is positive.
\end{lemma}

\begin{proof}{Proof of lemma \ref{lem:ineq_equilib_cost}.}
In order for routing all demand through one route ($\splitfraction{\playertype^\typeInf}{1}{\bneIndicator} = 0$ or $1$) to be an equilibrium action for population $\typeInf$, the expected cost of the route chosen must be less than the other route. Since $\accuracy^\typeInf=1$, the signal that population $\typeInf$ players receive will always be the true state of nature, i.e. they will always receive signal $\statenorm$ (resp. $\stateinc$) when the true state is $\statenorm,$ (resp. $\stateinc$). Therefore, the expected cost calculated by population $\typeInf$ players coincides with the realized cost. Thus, in a given state, the realized cost of routing all of population $\typeInf$'s demand through the chosen route is less than the realized cost of taking the other route.

Since we assume that the demand is sufficiently high (see \eqref{eq:demand_assumption}), there is no equilibrium where all players take the same route. Thus, there must be some players of population $\typeUninf$ on the route not chosen by population $\typeInf$; these population $\typeUninf$ players necessarily receive a higher cost in equilibrium. Therefore, $\costVar^{\typeUninf\bneIndicator}_\staterv > \costVar^{\typeInf\bneIndicator}_\staterv$, and by \eqref{eq:ind_val_info_struct_state}, we obtain that the relative individual value of information in that state is positive.  \hfill \Halmos

% $\mathbb{E}_{\bar{\belief}}[\latency{\routeindex}{\staterv}{\load_\routeindex^{\playertype^\typeInf} + \load_\routeindex^{\playertype^\typeUninf}}|\playertype^\typeInf=\typeInf\statenorm] = \latency{\routeindex}{\statenorm}{\load_\routeindex^{\playertype^\typeInf} + \load_\routeindex^{\playertype^\typeUninf}}.$ Similarly, we can show 
\end{proof}

From the second lemma, we obtain that that if the population $\typeInf$ players split their demand along both routes in a given state, then the realized cost for a population $\typeInf$ player is equal to that of a population $\typeUninf$ player.

\begin{lemma}
\label{lem:eq_equilib_cost}
Consider the information environment $(\pInc,\fracInf,\accuracy^\typeInf=1,\accuracy^\typeUninf=0.5)$. In a given state $\staterv\in\{\stateinc,\statenorm\}$, if $\splitfraction{\playertype^\typeInf}{1}{\bneIndicator}\in(0,1)$, then the realized equilibrium cost in that state for a population $\typeUninf$ player is equal to that of a population $\typeInf$ player, i.e. $\costVar^{\typeUninf\bneIndicator}_\staterv = \costVar^{\typeInf\bneIndicator}_\staterv$. Therefore, the relative individual value of information in that state, $\valueVar_\staterv$, is zero. 
\end{lemma}

\begin{proof}{Proof of Lemma \ref{lem:eq_equilib_cost}.}

In order for $\splitfraction{\playertype^\typeInf}{1}{\bneIndicator}\in(0,1)$ to be an equilibrium action for population $\typeInf$, the expected cost on both routes must be equal. Recall that since $\accuracy^\typeInf=1$, the signal that population $\typeInf$ players receive will always coincide with the true state of nature, and thus the expected cost will be equal to the realized cost on both routes. Therefore, in that state, the realized equilibrium costs on both routes are equal, and players of either population receive the same cost, i.e. $\costVar^{\typeUninf\bneIndicator}_\staterv = \costVar^{\typeInf\bneIndicator}_\staterv$. Following \eqref{eq:rel_ind_val_inf_state}, the relative individual value of information for that state is zero. \hfill \Halmos
\end{proof}

We present in Theorem \ref{thm:ind_val_of_info} our results on the relative expected individual value of information for the information environment $(\pInc,\fracInf,\accuracy^\typeInf=1,\accuracy^\typeUninf=0.5)$. We find that the relative expected value, $\valueVar$, is positive when $\fracInf < \fracBound{3}$, and zero otherwise. In other words, in parameter regimes $\regime{1}$--$\regime{3}$, a population $\typeInf$ player will always have a lower expected equilibrium cost than a population $\typeUninf$ player, and in regime $\regime{4}$, all players will have the same expected equilibrium costs. Formally:

%%%%%%%%%%%%%Theorem individual value of info%%%%%%%%%%%%
\begin{theorem}{}
\label{thm:ind_val_of_info}

For a given $\pInc$ and $\accuracy^\typeInf=1$, if $\fracInf < \fracBoundEta{3}$, then the relative expected individual value of information, $\valueVar>0$. If $\fracInf \geq \fracBoundEta{3}$, then $\valueVar = 0$.
\end{theorem}
%%%%%%%

\begin{proof}{Proof of Theorem \ref{thm:ind_val_of_info}}

To show that the relative expected value of information is positive in $\regime{1}$ -- $\regime{3}$, we must show that the relative value of information is positive in one state and zero in the other, or positive in both states. By Lemma \ref{lem:ineq_equilib_cost} it suffices to show that population $\typeInf$ players play $\splitfraction{\playertype^\typeInf}{1}{\bneIndicator}\in\{0,1\}$ in one state, and $\splitfraction{\playertype^\typeInf}{1}{\bneIndicator} \in (0,1)$ in the other, or $\splitfraction{\playertype^\typeInf}{1}{\bneIndicator}\in\{0,1\}$ in both states. 

To show that the relative expected value of information is zero in $\regime{4}$, we must show that the the relative value of information is zero in both states. By Lemma \ref{lem:eq_equilib_cost} it suffices to show that players in population $\typeInf$ play $\splitfraction{\playertype^\typeInf}{1}{\bneIndicator} \in (0,1)$ for both states.

% Recall that the individual value of information is the difference in realized equilibrium costs between the type $\typeUninf$ players and population $\typeInf$ players. It represents the marginal cost improvement that a less informed player could achieve if they could become type $\typeInf$.

% From lemma \ref{lem:eq_equilib_cost}, it follows that in any regime where $0<\splitfraction{\playertype^\typeInf}{1}{\bneIndicator}<1$, the realized equilibrium costs for a given state on both routes will be equal, and thus $\valueVar_{ind} = 0$. Additionally, lemma \ref{lem:ineq_equilib_cost} implies that in any regime where $\splitfraction{\playertype^\typeInf}{1}{\bneIndicator}\in\{0,1\}$, the cost for population $\typeInf$ players in that state will be less than the cost for type $\typeUninf$ players, and thus $\valueVar_{ind} > 0$. 

In regime $\regime{1}$: the equilibrium split fractions for population $\typeInf$ players are $\splitfraction{\typeInf}{1}{\statenorm\bneIndicator} = 1$ and $\splitfraction{\typeInf}{1}{\stateinc\bneIndicator} = 0$ (Proposition \ref{prop:reg1equilibrium}). By Lemma \ref{lem:ineq_equilib_cost}, in this regime the relative individual value of information is positive in both states. It follows that the relative expected individual value of information is also positive. 

In regimes $\regime{2},\regime{3}$: Propositions \ref{prop:reg2equilibrium} and \ref{prop:reg3equilibrium} state that $\splitfraction{\typeInf}{1}{\statenorm\bneIndicator} = 1$, and $\splitfraction{\typeInf}{1}{\stateinc\bneIndicator}\in(0,1)$. Therefore, the relative individual value is positive in state $\statenorm$ $(\valueVar_\statenorm > 0)$, and zero in state $\stateinc$ $(\valueVar_\stateinc=0)$. It follows that the relative expected value of information is positive in these regimes. 

In regime $\regime{4}$: the equilibrium play of type $\typeInf$ for both states is to split, i.e. $\splitfraction{\typeInf}{1}{\statenorm\bneIndicator},\splitfraction{\typeInf}{1}{\stateinc\bneIndicator}\in(0,1)$, see Proposition \ref{prop:reg4equilibrium}. Therefore, by Lemma \ref{lem:eq_equilib_cost}, the realized equilibrium costs must be equal for all players in both states, thus the relative expected individual value of information is zero in this regime. \hfill \Halmos
\end{proof}

We illustrate the results of this theorem in Figures~\ref{fig:type_costs_state} and \ref{fig:type_dep_exp_cost} with numerical results generated using the parameters in Table \ref{tab:param_vals}. The costs were normalized by the social optimum cost for ease of comparison. In Figure~\ref{fig:type_dep_exp_cost}, we see that the relative expected individual value of information (the difference between the solid red line and dashed blue line) is positive when $\fracInf < \fracBoundEta{3}$, and zero when $\fracInf>\fracBoundEta{3}$.

\begin{figure}[htb]
\centering
\subfloat[$\dfrac{\costVar^{\typeUninf\bneIndicator}_\statenorm}{\costVar_\statenorm^\socOpt}$, $\dfrac{\costVar^{\typeInf\bneIndicator}_\statenorm}{\costVar_\statenorm^\socOpt}$ vs. $\fracInf$, $\pInc=0.2$]{
\includegraphics[trim={1cm 6cm 1cm 6cm},clip,width=0.4\linewidth]{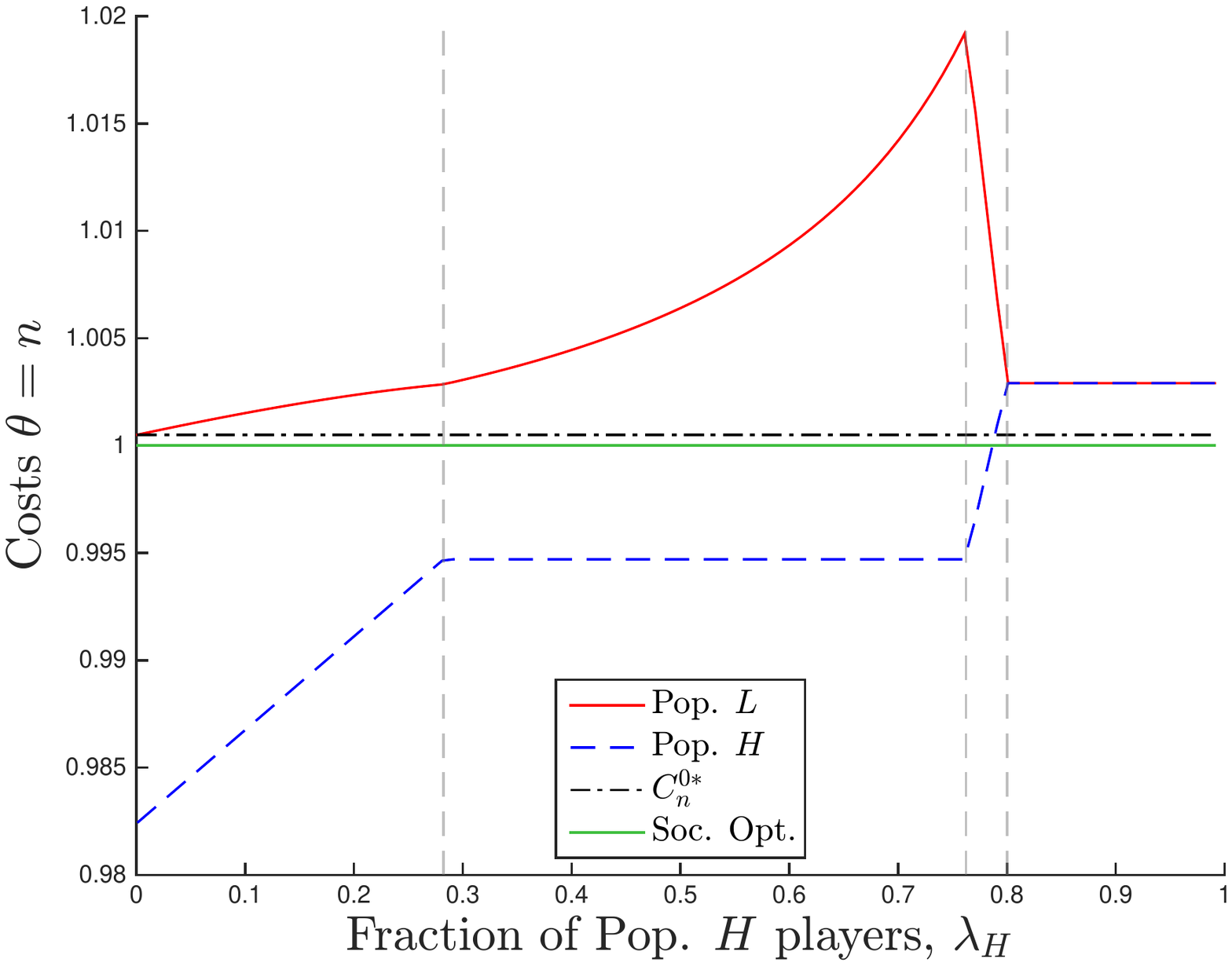}}
\subfloat[$\dfrac{\costVar^{\typeUninf\bneIndicator}_\statenorm}{\costVar_\statenorm^\socOpt}$, $\dfrac{\costVar^{\typeInf\bneIndicator}_\statenorm}{\costVar_\statenorm^\socOpt}$ vs. $\fracInf$, $\pInc=0.6$]{
\includegraphics[trim={1cm 6cm 1cm 6cm},clip,width=0.4\linewidth]{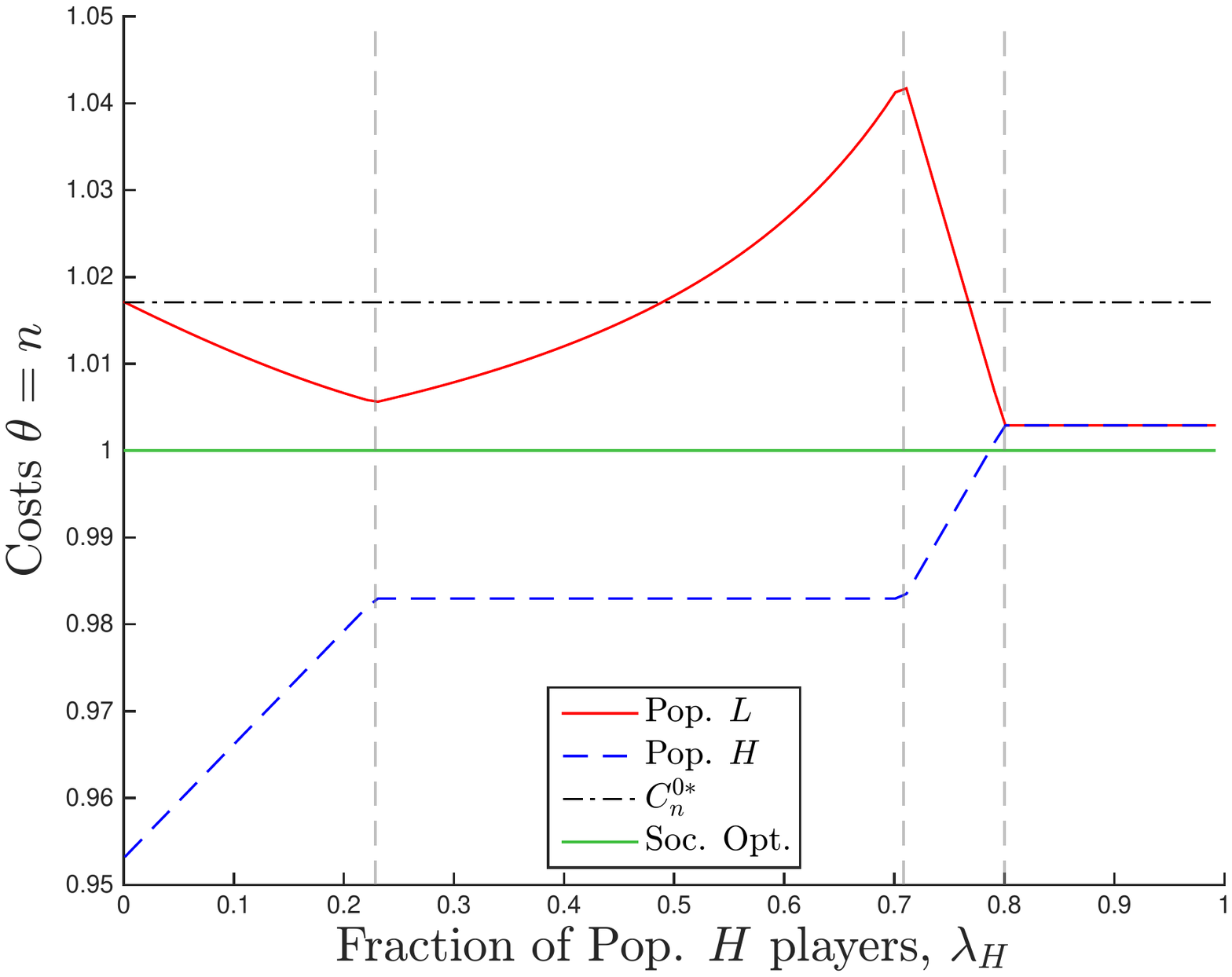}}\\
\subfloat[$\dfrac{\costVar^{\typeUninf\bneIndicator}_\stateinc}{\costVar_\stateinc^\socOpt}$ $\dfrac{\costVar^{\typeInf\bneIndicator}_\stateinc}{\costVar_\stateinc^\socOpt}$ vs. $\fracInf$, $\pInc=0.2$]{
\includegraphics[trim={1cm 6cm 1cm 6cm},clip,width=0.4\linewidth]{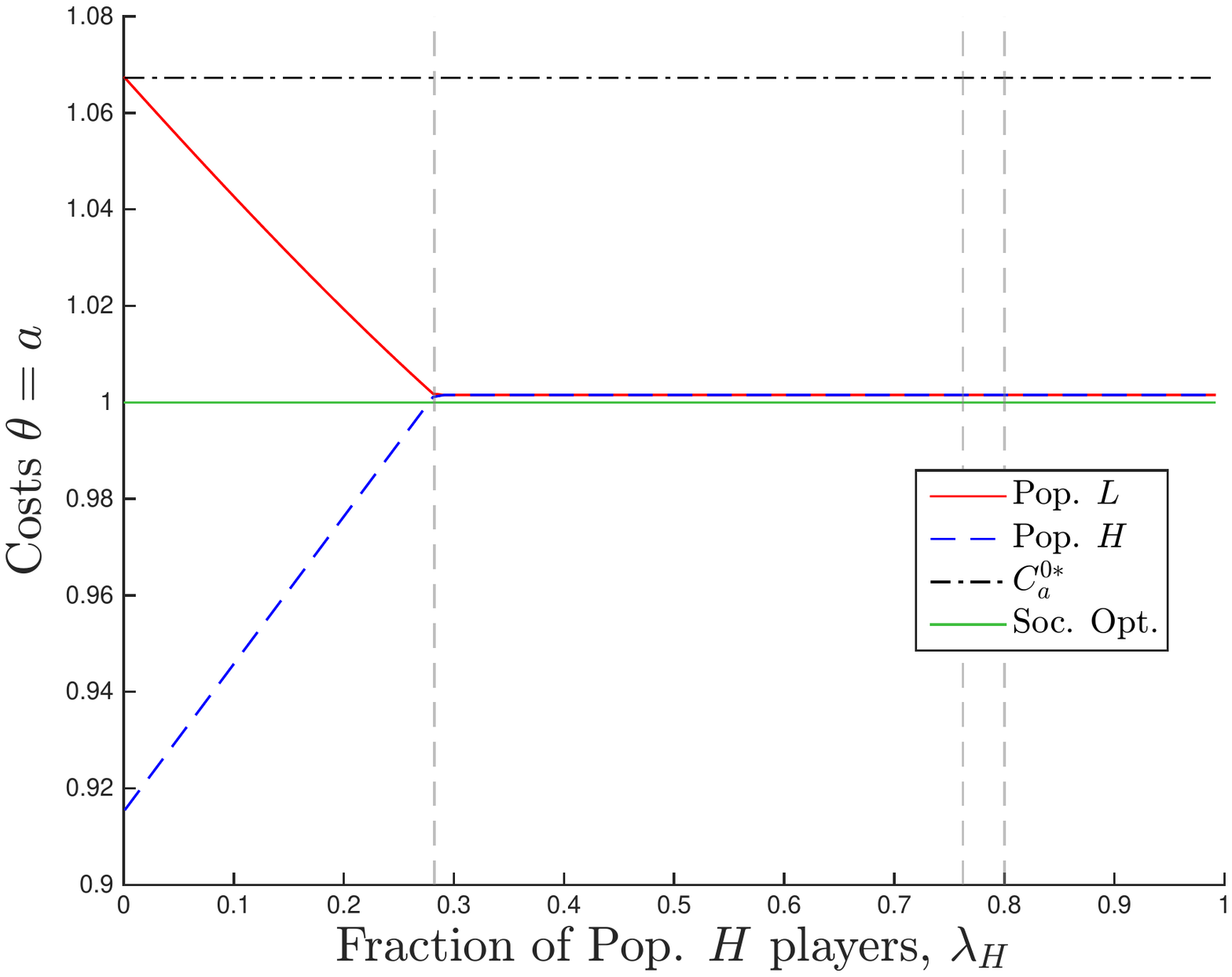}}
\subfloat[$\dfrac{\costVar^{\typeUninf\bneIndicator}_\stateinc}{\costVar_\stateinc^\socOpt}$ $\dfrac{\costVar^{\typeInf\bneIndicator}_\stateinc}{\costVar_\stateinc^\socOpt}$ vs. $\fracInf$, $\pInc=0.6$]{
\includegraphics[trim={1cm 6cm 1cm 6cm},clip,width=0.4\linewidth]{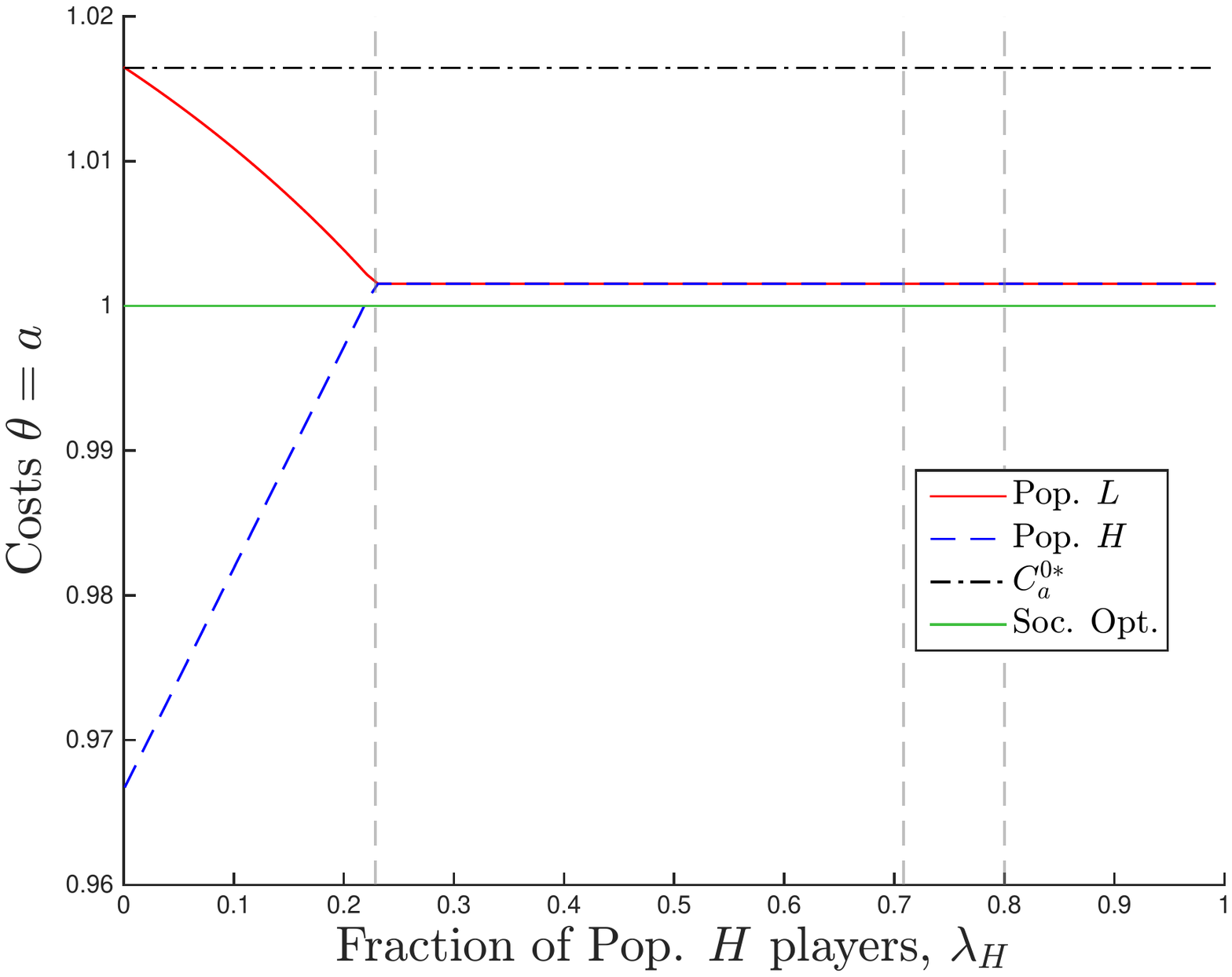}}
\caption{Population- and state-specific costs vs. the fraction of informed users; normalized by soc. opt. cost.}
\label{fig:type_costs_state}
\end{figure}

\begin{figure}[htb]
\centering
\subfloat[ $\dfrac{\costVar^{\typeUninf\bneIndicator}}{\costVar^\socOpt}$, $\dfrac{\costVar^{\typeInf\bneIndicator}}{\costVar^\socOpt}$ vs. $\fracInf$, $\pInc=0.2$]{
\includegraphics[trim={1cm 6cm 1cm 6cm},clip,width=0.4\linewidth]{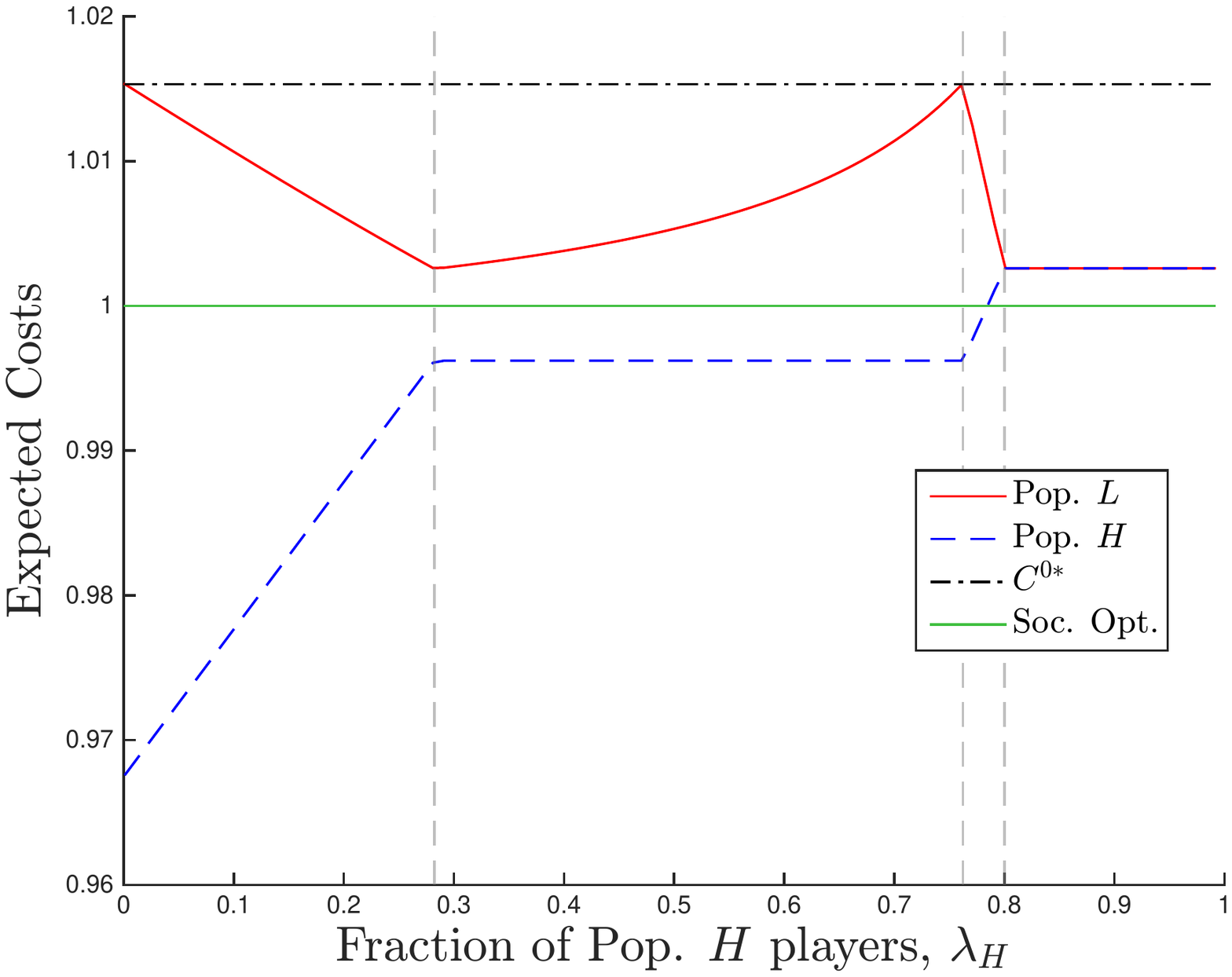}}
\subfloat[ $\dfrac{\costVar^{\typeUninf\bneIndicator}}{\costVar^\socOpt}$, $\dfrac{\costVar^{\typeInf\bneIndicator}}{\costVar^\socOpt}$ vs. $\fracInf$, $\pInc=0.6$]{
\includegraphics[trim={1cm 6cm 1cm 6cm},clip,width=0.4\linewidth]{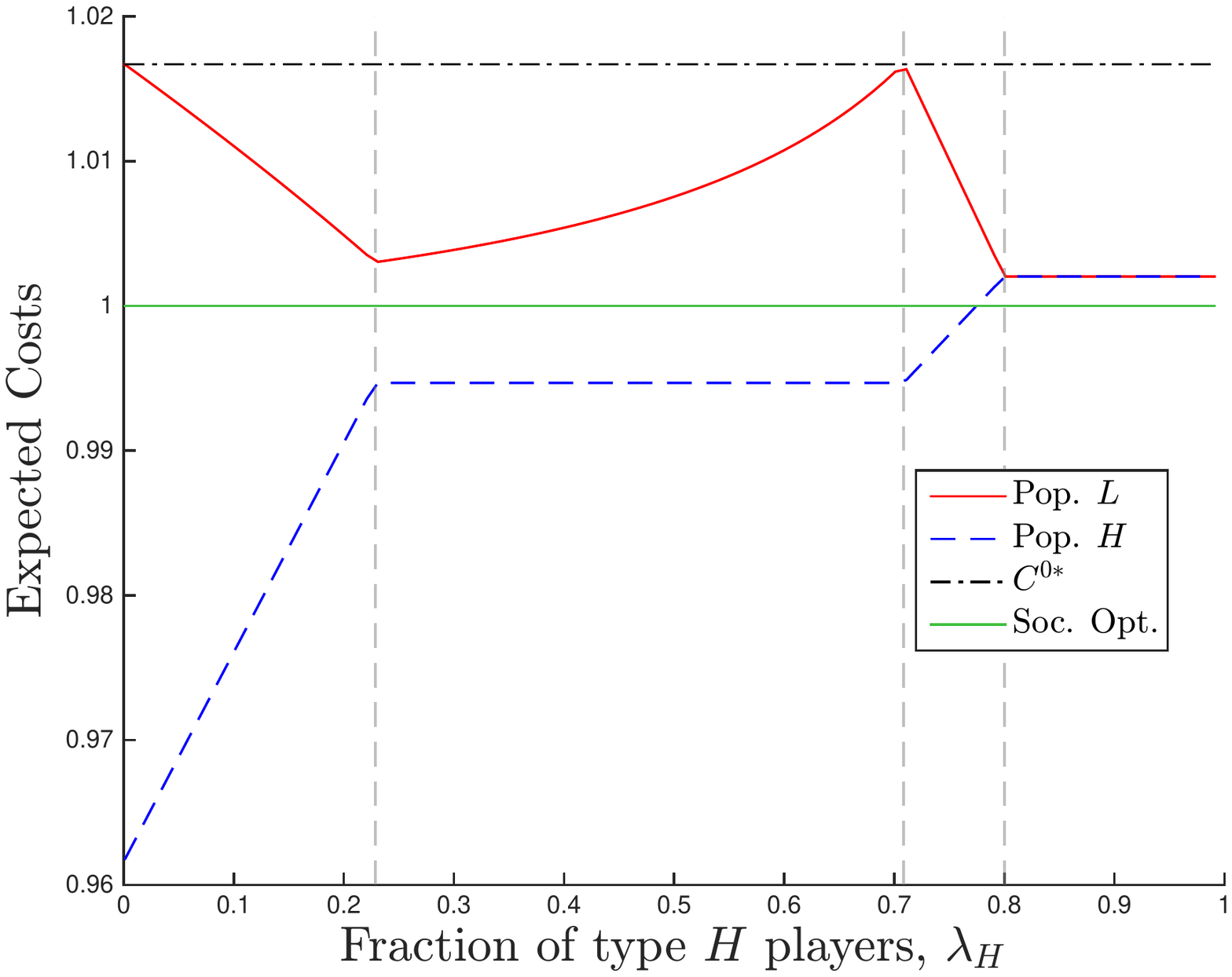}}
\caption{Expected cost in equilibrium for players in each population; normalized by soc. opt. cost.}
\label{fig:type_dep_exp_cost}
% subsection subsection_name (end)
\end{figure}

It is interesting to note that the expected cost for population $\typeInf$ is non-decreasing in $\fracInf$. That is, the more of the population that has access to the network state information, the higher the expected costs for type $\typeInf$ (Figure \ref{fig:type_dep_exp_cost}). Additionally, players that don't receive information still enjoy a reduction in expected equilibrium cost when others are receiving information. These effects were also noted by \cite{mahmassani1991system} in their simulations. This poses a potential conundrum for TIS providers: providers may want to increase market share to maximize revenue, but in doing so, they may limit or even nullify the usefulness of their service to their subscribers. Indeed, in regime $\regime{4}$, players with access to information experience no additional benefit compared to those without. 

We note that the relative expected individual value of information is greatest when $\fracInf$ is close to zero (Figure \ref{fig:type_dep_exp_cost}). This is intuitive, since information provides the most value when the fewest number of players can take advantage of it. One may also expect that the relative expected individual value of information would decrease as $\fracInf$ increases, since more players are utilizing the information. Indeed, this is the case in $\regime{1}$ and $\regime{3}$ it does decrease. However, in $\regime{2},$ the relative expected individual value of information increases with $\fracInf$. While the expected cost for population $\typeInf$ players does not change in $\regime{2}$, the expected cost for population $\typeUninf$ players increases as $\fracInf$ increases.

% subsubsection aggregate_cost (end)

\subsection{Social value of information} % (fold)
\label{sub:social_value_of_information}

% subsection social_value_of_information (end)
Recall from \eqref{eq:exp_soc_val_of_info} that the expected social value of information, $\welfare$, represents the reduction in the expected equilibrium cost that an average player enjoys under the given information environment, $(\pInc,\fracInf,\accuracy^\typeInf=1,\accuracy^\typeUninf=0.5)$, compared to the baseline situation where nobody has any access to information (i.e. $\fracInf=0$). Theorem \ref{thm:soc_val_of_info} concretizes the relationship between the expected social value of information and the fraction of population $\typeInf$ players.

%%%%%%%%%%Theorem - social value of info%%%%%%%%%%%%%
\begin{theorem}{}
\label{thm:soc_val_of_info}

For the information environment $(\pInc,\fracInf,\accuracy^\typeInf=1,\accuracy^\typeUninf=0.5)$, as the fraction of population $\typeInf$ players increases:
\begin{enumerate}
  \item[i.)] in $\regime{1}$, $\welfare$ increases in $\fracInf$,
  \item[ii.)] in $\regime{2}$, $\welfare$ is constant in $\fracInf$,
  \item[iii.)] in $\regime{3}$, $\begin{cases}
  \welfare \text{ decreases in } \fracInf, &\text{ if } \fracBoundEta{2} \geq \dfrac{2\slope{2}{}\totaldemand+\intercept{2}-\intercept{1}}{2\totaldemand(\slope{1}{\statenorm}+\slope{2}{})}\\
  % \text{increases}, &\text{ if } \dfrac{2\slope{2}{}\totaldemand+\intercept{2}-\intercept{1}}{2\totaldemand(\slope{1}{\statenorm}+\slope{2}{})}<\dfrac{\fracBoundEta{2}+\fracBoundEta{3}}{2}\\
    \welfare \text{ increases then decreases in } \fracInf, &\text{ if } \fracBoundEta{2}<\dfrac{2\slope{2}{}\totaldemand+\intercept{2}-\intercept{1}}{2\totaldemand(\slope{1}{\statenorm}+\slope{2}{})} < \fracBoundEta{3} \\
    \welfare \text{ increases in } \fracInf, &\text{ if }\fracBoundEta{3} \leq \dfrac{2\slope{2}{}\totaldemand+\intercept{2}-\intercept{1}}{2\totaldemand(\slope{1}{\statenorm}+\slope{2}{})}
  \end{cases}$,
  \item[iv.)] in $\regime{4}$, $\welfare$ is constant in $\fracInf$
\end{enumerate}
\end{theorem}
%%%%%%%%%

\begin{proof}{Proof of Theorem \ref{thm:soc_val_of_info}}
Recall from \eqref{eq:exp_soc_val_of_info} that the expected social value of information is the difference between the baseline expected cost, $\costVar^{0\bneIndicator}$, and the expected equilibrium cost $\bar{\costVar}^\bneIndicator$. Since $\costVar^{0\bneIndicator}$ is a constant, in order to show that the expected social value of information increases in $\fracInf$ in regime $\regime{1}$, is constant in regime $\regime{2}$, etc., it suffices to show that the expected social cost $\bar{\costVar}^\bneIndicator$ decreases in $\fracInf$ in regime $\regime{1},$ etc. From Propositions \ref{prop:reg1equilibrium}--\ref{prop:reg4equilibrium}, we can calculate the expected social cost $\bar{\costVar}^\bneIndicator$ for each regime \eqref{eq:exp_soc_cost}:
\begin{align}
\bar{\costVar}^\bneIndicator=\begin{cases}
  (\fracInf)^2 \left(\totaldemand \pInc(1-\pInc)(\slope{1}{\stateinc}+\slope{1}{\statenorm}+\slope{2}{}-\bar{\slope{1}{}})  \right) + \fracInf \regConst_1 \left((1-\pInc)\slope{1}{\statenorm}+(2\pInc-1)\bar{\slope{1}{}} - \pInc\slope{1}{\stateinc}\right)\\
+ \frac{\regConst_1^2}{\totaldemand}\bar{\slope{1}{}}+\slope{2}{}\left(\totaldemand-2\regConst_1+\frac{\regConst_1^2}{\totaldemand}+\frac{\regConst_1}{\totaldemand}\intercept{1}+\left(1-\frac{\regConst_1}{\totaldemand}\right)\right), &\text{in } \regime{1}\\
  \pInc\left(\slope{2}{}\left(\totaldemand - \frac{\solConst}{\slope{1}{\stateinc}+\slope{2}{}}\right) + \intercept{2}\right) + (1-\pInc)\left(\frac{\regConst_4}{\totaldemand}\left(\slope{1}{\statenorm} \regConst_4+\intercept{1}\right) + \left(1 - \frac{\regConst_4}{\totaldemand}\right)\left(\slope{2}{}(\totaldemand - \regConst_4)+\intercept{2}\right) \right), &\text{in } \regime{2}\\
  \pInc (\slope{1}{\statenorm} \regConst_2+\intercept{1}) + (1-\pInc)((\fracInf)^2  \slope{1}{\statenorm} \totaldemand + \fracInf\intercept{1} + (1-\fracInf)^2\slope{2}{}\totaldemand+(1-\fracInf)\intercept{2}), &\text{in } \regime{3}\\
  \pInc\left(\slope{2}{} \left(\totaldemand - \frac{\solConst}{\slope{1}{\stateinc}+\slope{2}{}} \right) + \intercept{2}\right) + (1-\pInc)\left(\slope{2}{}\left(\totaldemand-\frac{\solConst}{\slope{1}{\statenorm}+\slope{2}{}}\right)+\intercept{2}\right), &\text{in } \regime{4}
\end{cases} \label{eq:exp_soc_cost_analytic_form}
\end{align}

First, we would like to show that in regime $\regime{1}$, $\bar{\costVar}^\bneIndicator$ decreases in $\fracInf$. We note that in regime $\regime{1}$, $\bar{\costVar}^\bneIndicator$ is quadratic in $\fracInf$ and is minimized at $\widehat{\fracInf} = \dfrac{\regConst_1\left(\pInc\slope{1}{\stateinc}-(1-\pInc)\slope{1}{\statenorm}-(2\pInc-1)\bar{\slope{1}{}}\right)}{2\left(\totaldemand \pInc(1-\pInc)(\slope{1}{\stateinc}+\slope{1}{\statenorm}+\slope{2}{} - \bar{\slope{1}{}}\right)}$. For $\fracInf$ below this threshold, the derivative of the quadratic is negative. Therefore, if the right boundary of regime $\regime{1}$ is at or below this threshold, i.e. $\fracBoundEta{1} \leq \widehat{\fracInf}$, then $\bar{\costVar}^\bneIndicator$ is decreasing in $\fracInf$ in this regime. Solving the inequality $\fracBoundEta{1} \leq \widehat{\fracInf}$ yields the condition $\slope{1}{\statenorm}\leq\slope{1}{\stateinc}$, which is satisfied by assumption $(A\ref{ass:route_structure})_b$. Hence, we obtain that the expected social cost, $\bar{\costVar}^\bneIndicator$, is decreasing in $\fracInf$ in $\regime{1}$.

It is trivial to observe in \eqref{eq:exp_soc_cost_analytic_form} that in $\regime{2}$ and $\regime{4}$ the expected social cost $\bar{\costVar}^\bneIndicator$ is constant in $\fracInf$.

Lastly, we argue that in $\regime{3}$, $\bar{\costVar}^\bneIndicator$ has different behavior depending on the parameter values. We note from \eqref{eq:exp_soc_cost_analytic_form} that in regime $\regime{3}$, $\bar{\costVar}^\bneIndicator$ is quadratic in $\fracInf$ and is minimized at $\widetilde{\fracInf} = \dfrac{2\slope{2}{}\totaldemand+\intercept{2}-\intercept{1}}{2\totaldemand(\slope{1}{\statenorm}+\slope{2}{})}$. Below this threshold, the derivative of the quadratic is negative, and above this threshold, it is positive. Hence, if this threshold is at or below the left boundary of regime $\regime{3}$, i.e. $\widetilde{\fracInf} \leq \fracBoundEta{2}$, then the derivative is positive throughout $\regime{3}$ and the expected social cost increases in $\fracInf$.
If this threshold falls between the boundaries of regime $\regime{3}$, i.e. $\fracBoundEta{2}<\widetilde{\fracInf}<\fracBoundEta{3}$, then the derivative is negative on the interval $[\fracBoundEta{2},\widetilde{\fracInf})$, zero at $\widetilde{\fracInf}$, and positive on the interval $(\widetilde{\fracInf}, \fracBoundEta{3}]$. Therefore, in regime $\regime{3}$, the expected equilibrium cost, $\bar{\costVar}^\bneIndicator$, decreases then increases in $\fracInf$. Finally, if the threshold occurs at or above the right boundary of $\regime{3}$, i.e. $\widetilde{\fracInf}\geq \fracBoundEta{3}$, then the derivative is negative throughout regime $\regime{3}$, and $\bar{\costVar}^\bneIndicator$ is decreasing in $\fracInf$. \hfill \Halmos
\end{proof}

% \begin{corollary}{\textbf{Change in expected social cost from }$\regime{2}$\textbf{ to } $\regime{4}$}

% The change in expected social cost from $\regime{2}$ to $\regime{4}$ is given by:
% \begin{align}
% \Delta \bar{\costVar}_{2\rightarrow4} := (1-\pInc)\left(\slope{2}{}\left(\totaldemand - \frac{\solConst}{\slope{1}{\statenorm} + \slope{2}{}}\right) + \intercept{2} - \left(\frac{\regConst_4}{\totaldemand}\left(\slope{1}{\statenorm} \regConst_2+\intercept{1}\right) + \left(1 - \frac{\regConst_4}{\totaldemand}\right)\left(\slope{2}{}(\totaldemand - \regConst_4)+\intercept{2}\right) \right) \right).
% \end{align}

% $\Delta \bar{\costVar}_{2\rightarrow4}$ is nonnegative iff $\dfrac{2\slope{2}{}\totaldemand+\intercept{2}-\intercept{1}}{2\totaldemand(\slope{1}{\statenorm}+\slope{2}{})}\leq\dfrac{\fracBoundEta{2}+\fracBoundEta{3}}{2}$
% \end{corollary}

Corollary \ref{cor:global_min_soc_cost} below provides the expressions for the $\fracInf$ where the expected social cost is minimized (and consequently, the expected social value of information is maximized). We denote this fraction as $\fracInf_{\min}$. Above this fraction of highly-informed players, there is no additional reduction of expected social cost in increasing the fraction of population $\typeInf$ players. One can view this fraction as the ``optimal'' level of information distribution for reduction in social cost. A similar result was observed in simulations by \cite{mahmassani1991system}. Notably, this fraction is less than or equal to the threshold where the relative individual value of information goes to zero (see Thm. \ref{thm:ind_val_of_info}). If $\fracInf_{\min}<\fracBoundEta{3}$, then there exists a range of $\fracInf$ where it is personally advantageous for population $\typeUninf$ players to become population $\typeInf$ players, but it is harmful to society for them to do so.

\begin{corollary}{\textbf{Global minimum of expected social cost in equilibrium}.} \label{cor:global_min_soc_cost}

For the information environment $(\pInc,\fracInf,\accuracy^\typeInf=1,\accuracy^\typeUninf=0.5)$, the smallest $\fracInf$ where the minimum expected social cost is achieved is given below:
\begin{align}
\fracInf_{\min}:=
\begin{cases}
  \fracBoundEta{1}, &\text{ if }\fracBoundEta{2} \geq \dfrac{2\slope{2}{}\totaldemand+\intercept{2}-\intercept{1}}{2\totaldemand(\slope{1}{\statenorm}+\slope{2}{})}\\
  \dfrac{2\slope{2}{}\totaldemand+\intercept{2}-\intercept{1}}{2\totaldemand(\slope{1}{\statenorm}+\slope{2}{})} &\text{ if }\fracBoundEta{2} < \dfrac{2\slope{2}{}\totaldemand+\intercept{2}-\intercept{1}}{2\totaldemand(\slope{1}{\statenorm}+\slope{2}{})}<\fracBoundEta{3}\\
  \fracBoundEta{3} &\text{ if }\fracBoundEta{3}\leq \dfrac{2\slope{2}{}\totaldemand+\intercept{2}-\intercept{1}}{2\totaldemand(\slope{1}{\statenorm}+\slope{2}{})}.
\end{cases}
\end{align}
\end{corollary}

The social costs for each state are illustrated in Figure~\ref{fig:agg_costs_state} for $\pInc=0.2$ and $0.6,$ and the expected social costs are illustrated in Figure~\ref{fig:aggCosts}. When few players have access to incident information in regime $\regime{1}$, providing information to players significantly reduces the expected social cost. However, in $\regime{2}$, there is no change in social cost from providing information to more players. In regime $\regime{3}$, under certain conditions, there can be a modest reduction of expected social cost (Fig.~\ref{fig:soc_exp_cost_p60}). Under other conditions, however, increasing $\fracInf$ increases the expected social cost (Fig.~\ref{fig:soc_exp_cost_p20}). In $\regime{4}$, there is no change in expected cost as $\fracInf$ increases.

\begin{figure}[H]
\centering
\subfloat[$\dfrac{\costVar^\bneIndicator_\statenorm}{\costVar_\statenorm^\socOpt}$ vs. $\fracInf$, $\pInc=0.2$]{
\includegraphics[trim={1cm 6cm 1cm 6cm},clip,width=0.4\linewidth]{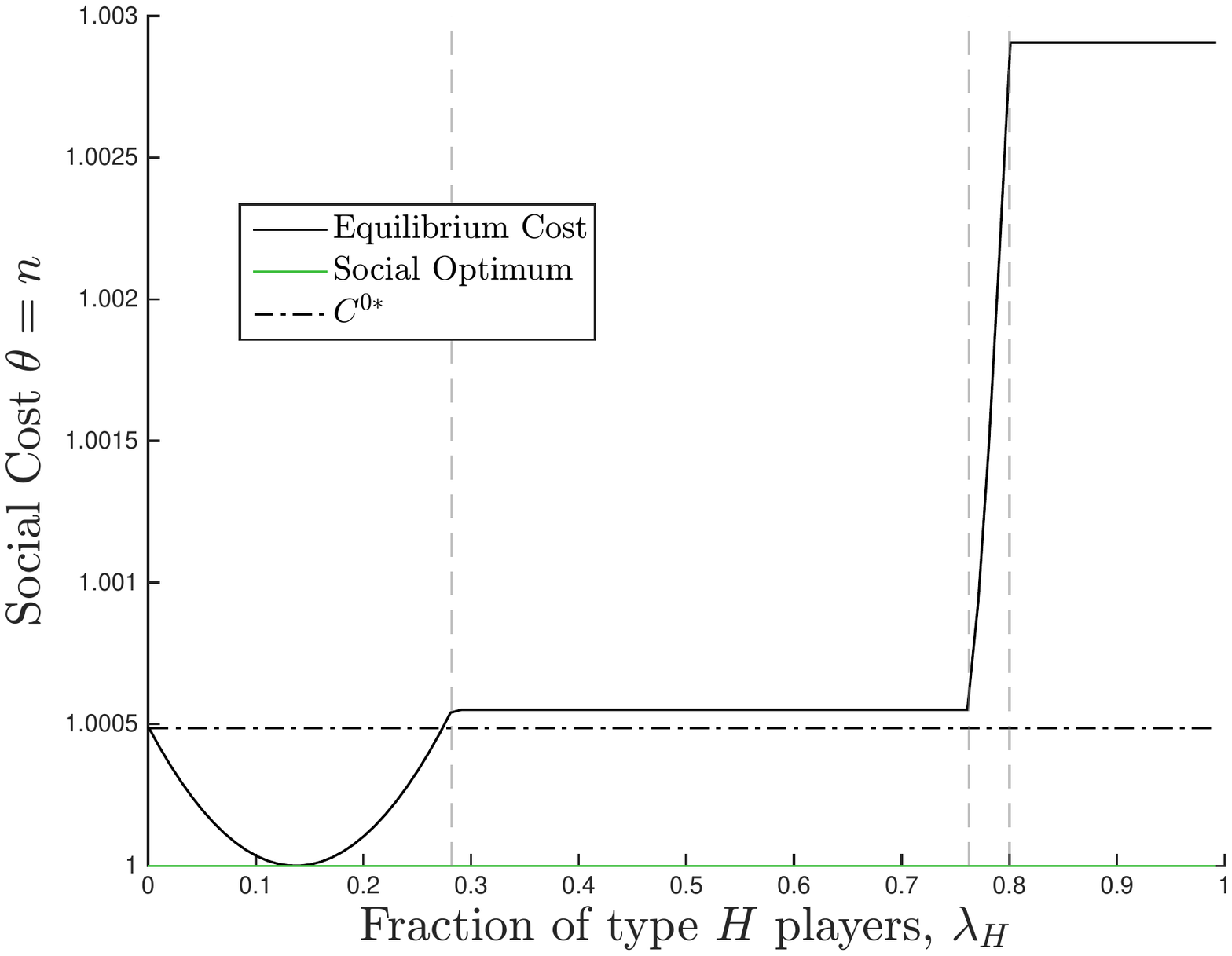}\label{fig:soc_norm_cost_p20}}
\subfloat[$\dfrac{\costVar^\bneIndicator_\statenorm}{\costVar_\statenorm^\socOpt}$ vs. $\fracInf$, $ \pInc=0.6$]{
\includegraphics[trim={1cm 6cm 1cm 6cm},clip,width=0.4\linewidth]{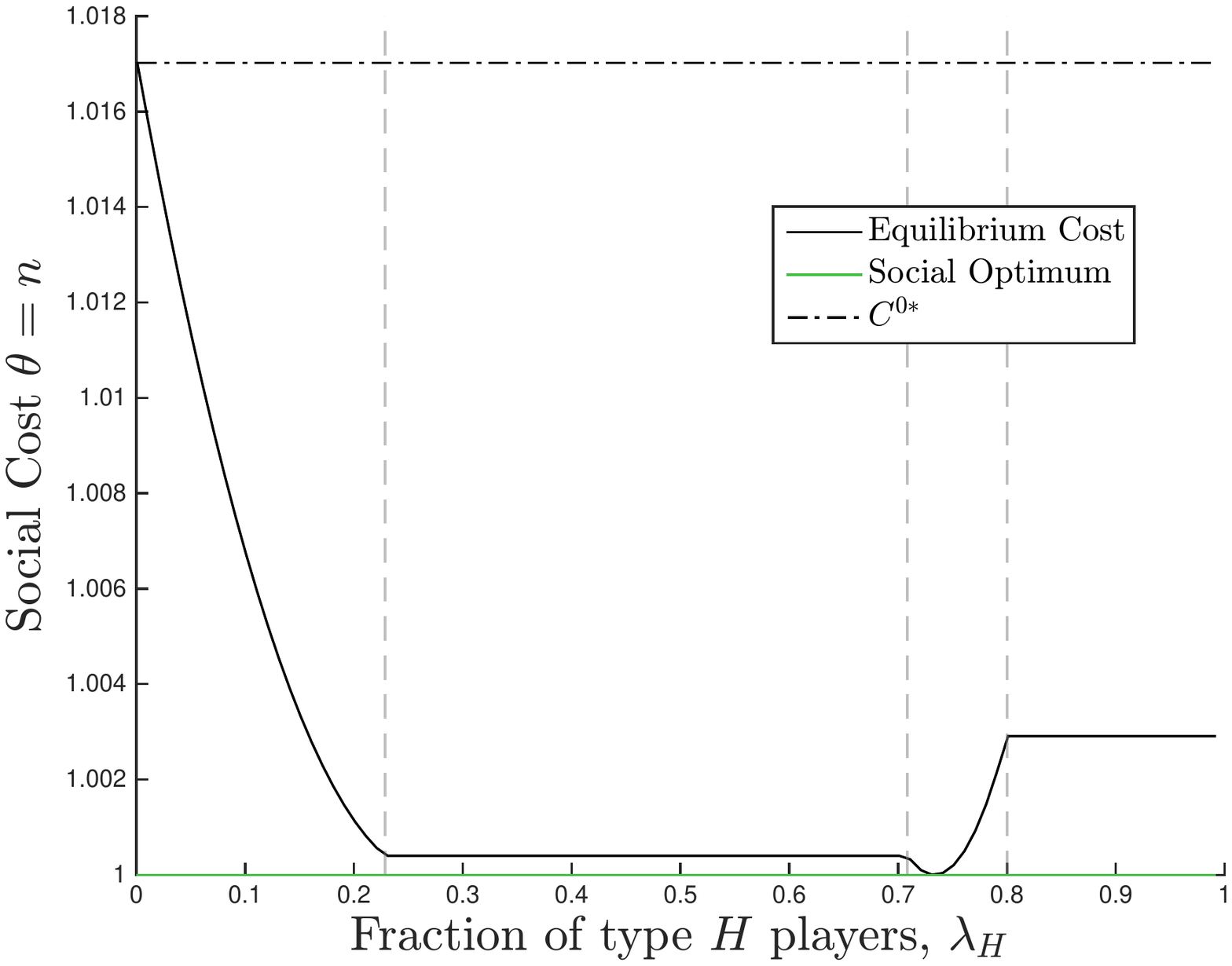}\label{fig:soc_norm_cost_p60}}\\
\subfloat[$\dfrac{\costVar^\bneIndicator_\stateinc}{\costVar_\stateinc^\socOpt}$ vs. $\fracInf$, $\pInc=0.2$]{
\includegraphics[trim={1cm 6cm 1cm 6cm},width=0.4\linewidth]{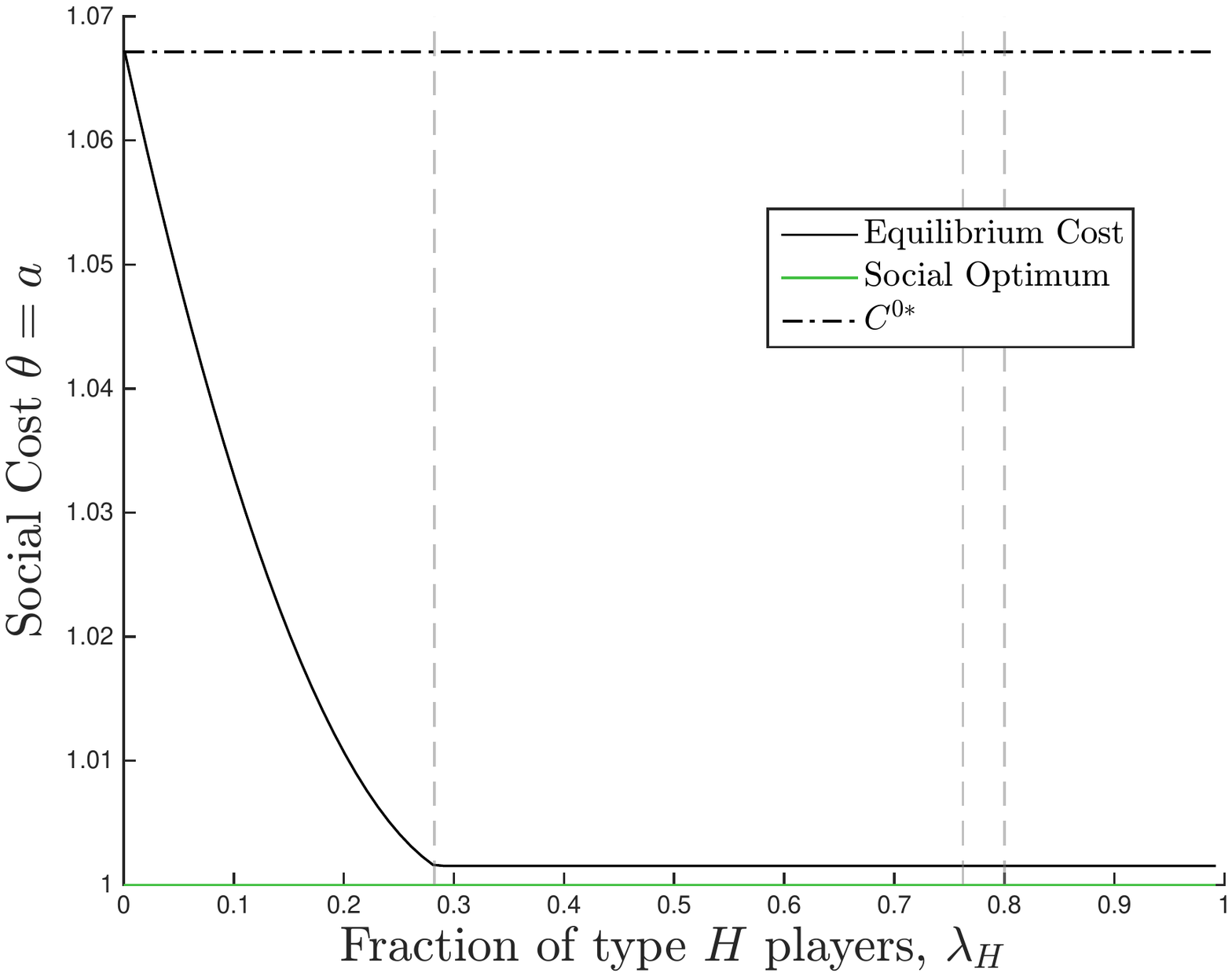}\label{fig:soc_inc_cost_p20}}
\subfloat[$\dfrac{\costVar^\bneIndicator_\stateinc}{\costVar_\stateinc^\socOpt}$ vs. $\fracInf$, $\pInc=0.6$]{
\includegraphics[trim={1cm 6cm 1cm 6cm},width=0.4\linewidth]{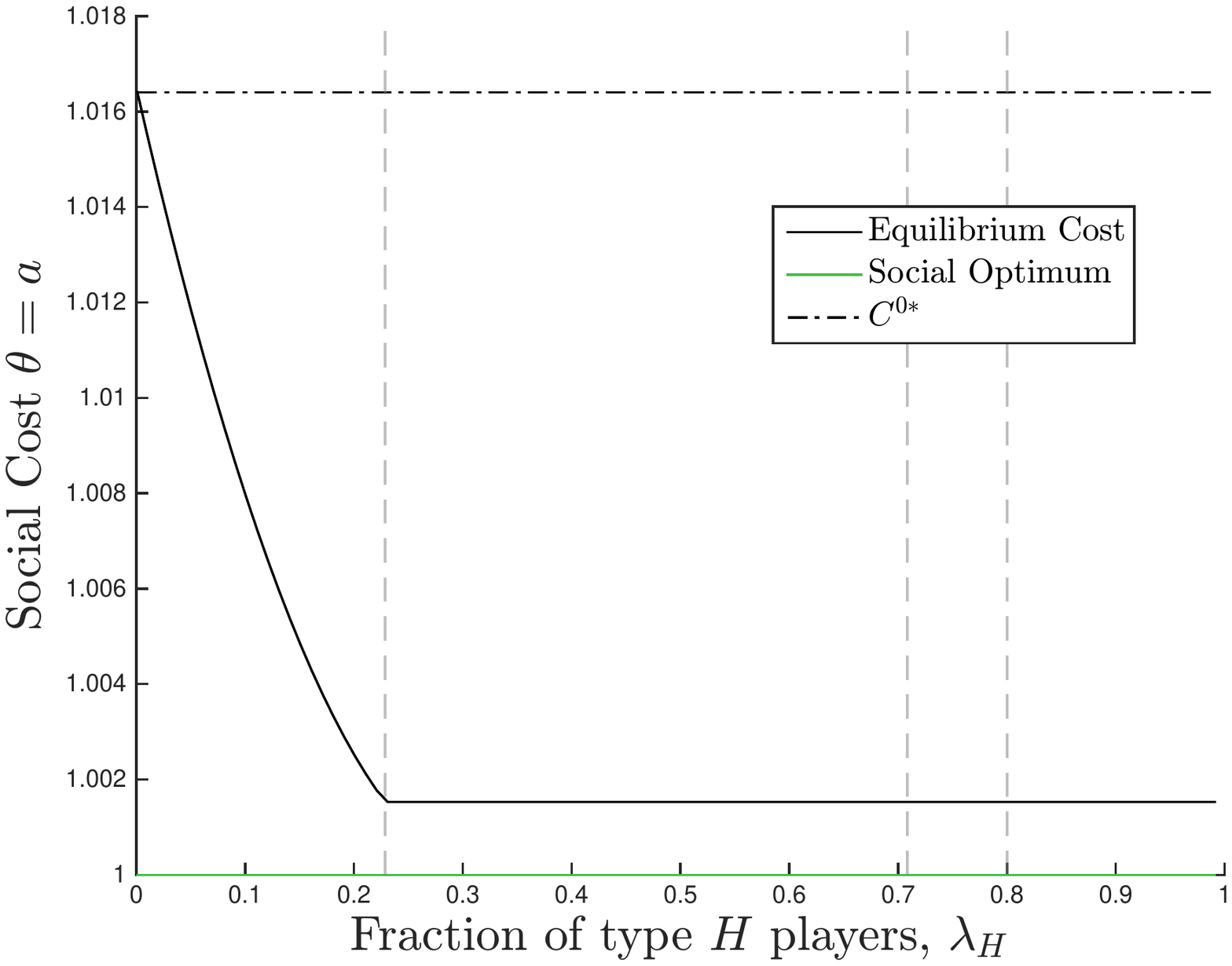}\label{fig:soc_inc_cost_p60}}
\caption{Social costs vs. fraction of informed users $\fracInf$; normalized by soc. opt. cost.}
\label{fig:agg_costs_state}
\end{figure}

\begin{figure}[H]
\centering
\subfloat[$\dfrac{\bar{\costVar}^\bneIndicator}{\costVar^\socOpt}$ vs. $\fracInf$, $\pInc=0.2$]{
\includegraphics[trim={1cm 6cm 1cm 6cm},clip,width=0.4\linewidth]{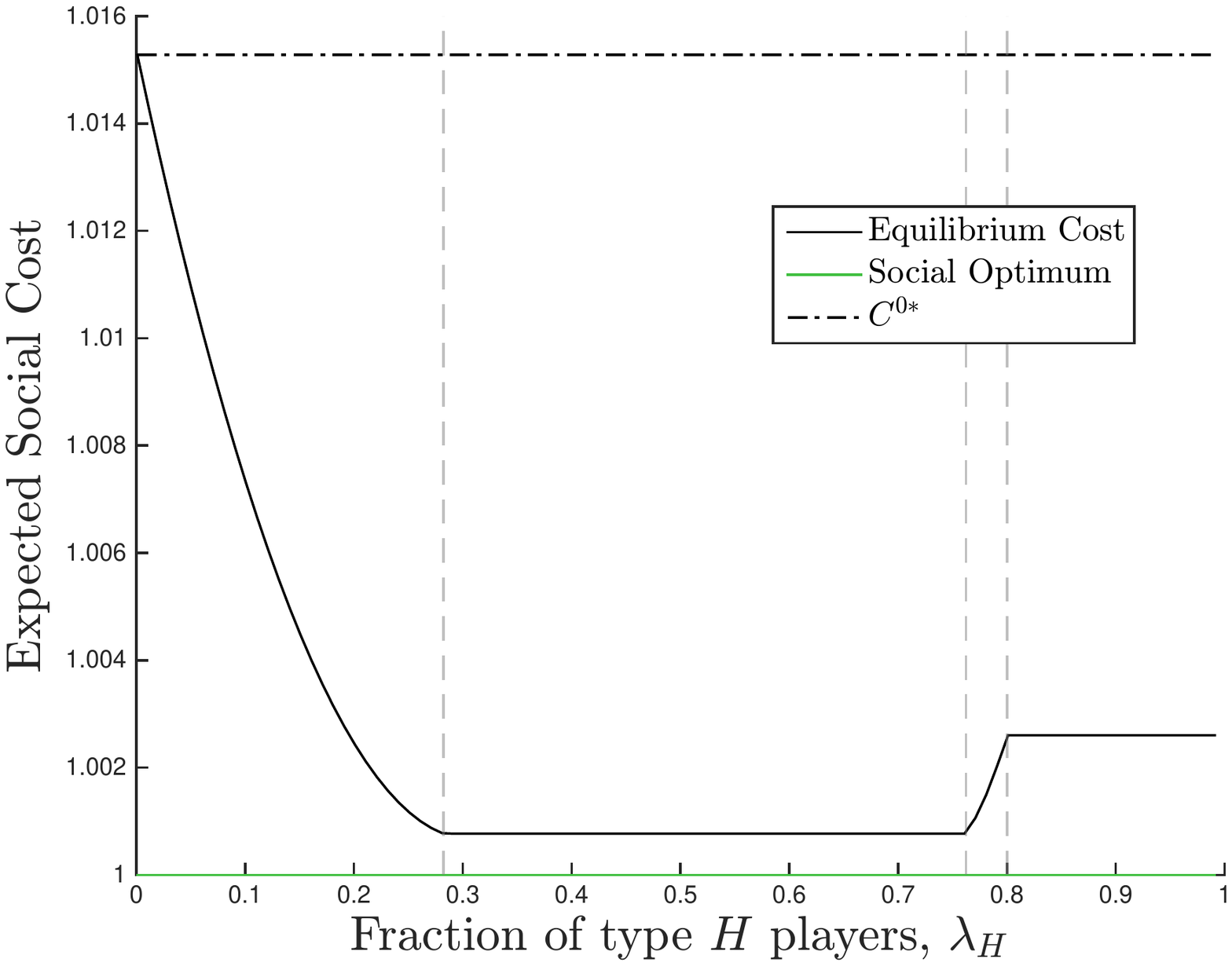}\label{fig:soc_exp_cost_p20}}
\subfloat[$\dfrac{\bar{\costVar}^\bneIndicator}{\costVar^\socOpt}$ vs. $\fracInf$, $\pInc=0.6$]{
\includegraphics[trim={1cm 6cm 1cm 6cm},clip,width=0.4\linewidth]{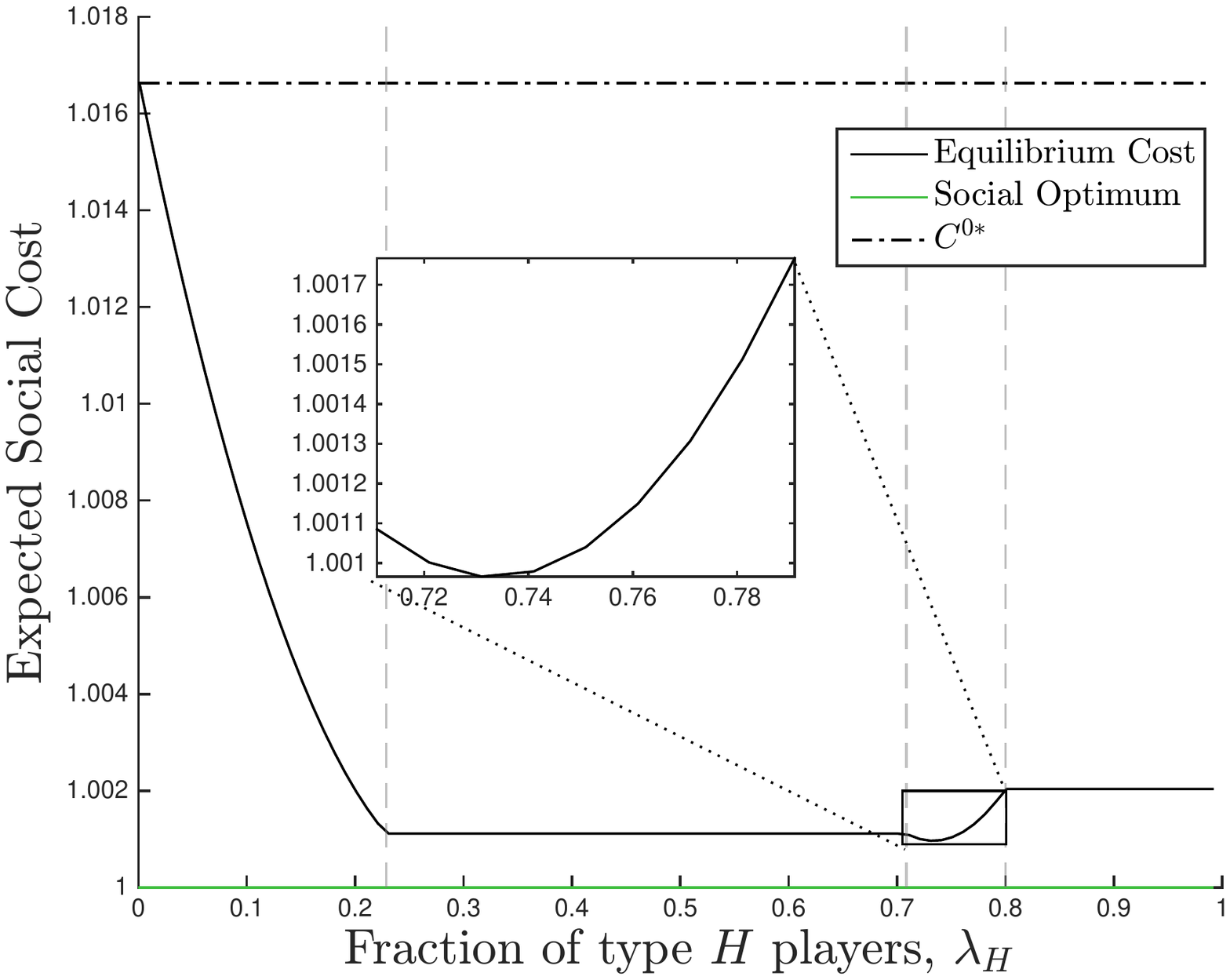}\label{fig:soc_exp_cost_p60}}
\caption{Expected social cost vs. fraction of informed users $\fracInf$; normalized by soc. opt. cost.}
\label{fig:aggCosts}
\end{figure}

% Figure \ref{fig:aggCosts} illustrates our result in Theorem \ref{thm:soc_val_of_info} that above a critical fraction coinciding with $\fracBoundEta{1}$, there is no improvement in social welfare. Notably, this critical fraction is below the critical fraction where improvement in individual welfare ceases ($\fracBoundEta{3}$); this means that in regimes $\regime{2},\regime{3}$, there is an individual incentive to switch from type $\typeUninf$ to type $\typeInf$ if possible, but there is no societal benefit in doing so. This is especially interesting since we observe that in regime $\regime{3}$, increasing the fraction of informed players is actually harmful to social welfare. This implies that there may be a regime where there is personal incentive to adopt a behavior that is harmful to society. The results from this simple model illustrate the non-trivial effects of information on both individual and social welfare. As traffic information becomes increasingly available to both travelers and system operators, it is especially important to model and predict the effects of information.

%% file: discussion.tex
%!TEX root = main.tex

\section{Concluding Remarks}
\label{sec:discussion}
As TIS become more ubiquitous, it becomes even more important to study the effects of information. This paper introduces a Bayesian congestion game model that allows for the study of the effects of heterogeneous information about traffic incidents on route choices. Our model provides a unifying analytical framework which considers incident probability, information distribution, and information accuracy. We are able to capture many of the effects of information observed by \cite{mahmassani1991system}, namely: (i) there exists an optimal fraction (less than 1) of commuters with information that achieves the maximum reduction in social cost; (ii) the cost for commuters with information increases as the fraction of informed commuters increases; and (iii) even commuters without information receive a reduction in cost when others have access to information.

In our analysis of the value of information, we only focused on the case where population $\typeInf$ players receive perfectly accurate information about the incident state. However, real TIS may be noisy and inaccurate, and additional research is needed to fully characterize the effects of the distribution of imperfect information. While our equilibrium results hold for general $\accuracy^\typeInf$ (see figure Figure~\ref{fig:equilib_etaH_075} for the equilibrium regimes and split fractions for $\accuracy^\typeInf=0.75$), further work is necessary to fully characterize and understand the effects of noisy information on equilibrium route choices and costs.

\begin{figure}[htb]
\centering
\subfloat[Equilibrium regimes in the $\fracInf-\pInc$ plane]{
\includegraphics[trim={2.5cm 0cm 2.5cm 1cm},clip,width=0.45\linewidth]{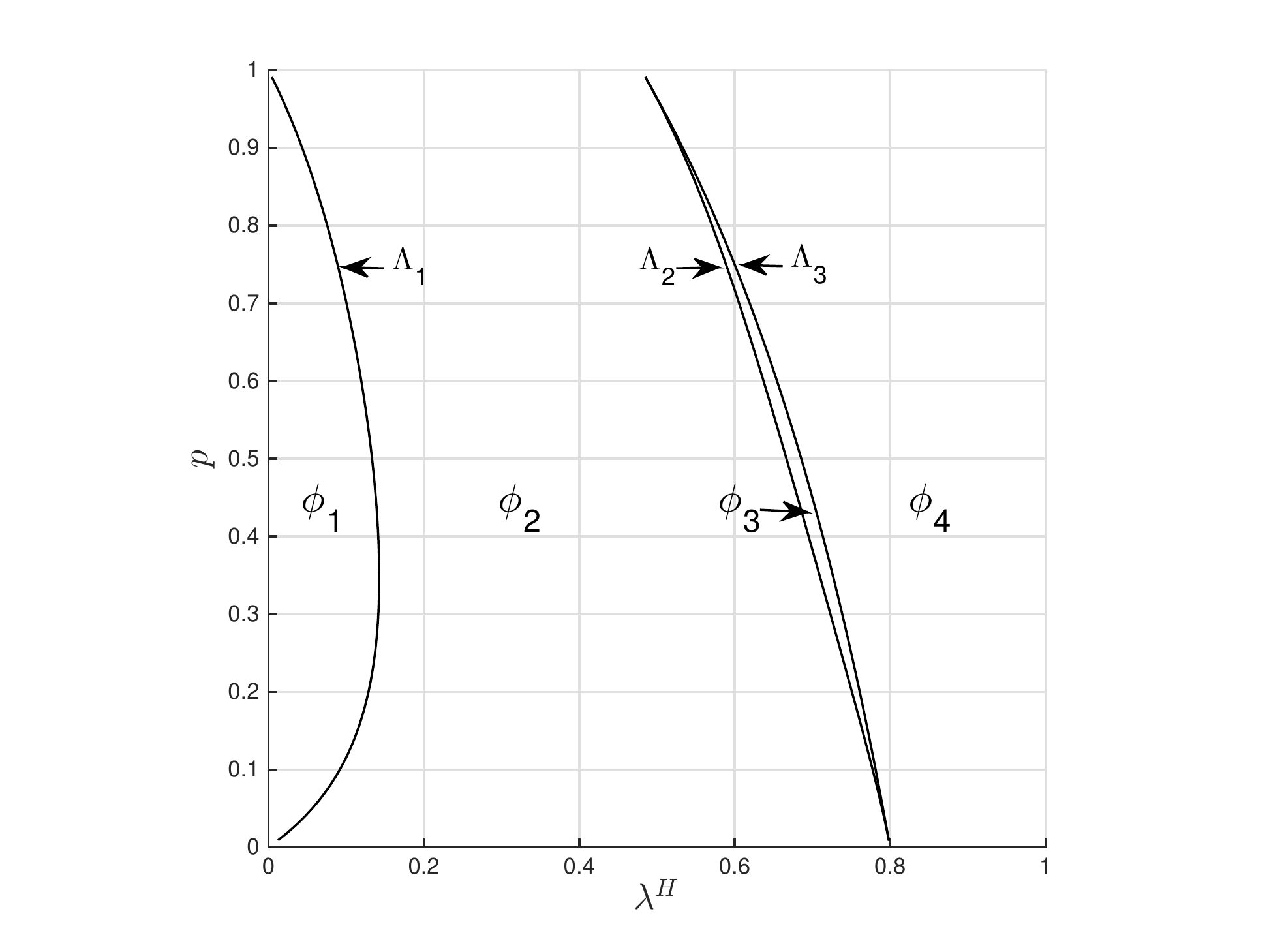}\label{fig:regBound_etaH075}}
\subfloat[Equilibrium split fractions vs. $\fracInf$, $\pInc=0.2$]{
\includegraphics[trim={2.5cm 0cm 2.5cm 1cm},clip,width=0.45\linewidth]{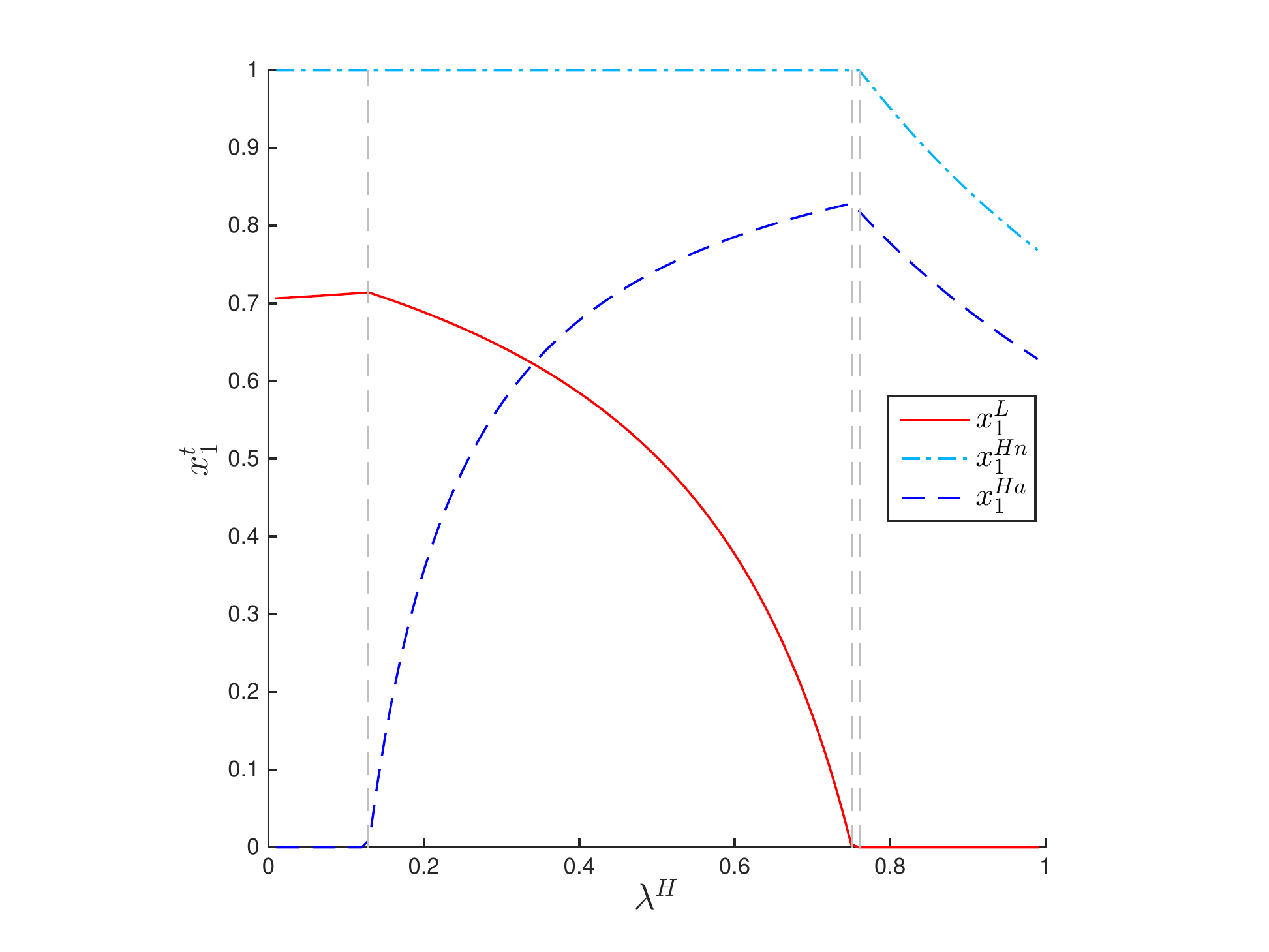}\label{fig:splitfrac_etaH075}}
\caption{Equilibrium regimes and split fractions for $\accuracy^\typeInf=0.75$.}
\label{fig:equilib_etaH_075} 
\end{figure}

The results of this paper raise some practical considerations. TIS providers must take into account that the value of information to informed commuters decreases as the fraction of highly-informed commuters increase. This is especially pertinent for services like Waze, which use crowdsourced data. As more people use such services, the quality of information may improve due to having more users reporting information, but the relative benefit of having access to information may decrease as more commuters take advantage of it. TIS providers must balance these effects in order to provide a useful service. 

We also identify that there are conditions where it would be advantageous for an individual to switch from being a less-informed to a highly-informed commuter (i.e. relative value of information is positive), but it would be detrimental for society if they did so. Our model can help social planners by identifying conditions where this may occur.